\documentclass[11pt]{article}
\usepackage[utf8]{inputenc}
\usepackage[english]{babel}
\usepackage{graphicx}         
\usepackage[pdftex, colorlinks=true, linkcolor=black, citecolor=black]{hyperref}
\usepackage{braket}
\usepackage{amsmath}
\usepackage{amssymb}
\usepackage{amsthm}
\usepackage{mathtools}
\usepackage{bbm}		
\usepackage{array}		
\usepackage{color}
\usepackage[usenames,dvipsnames]{xcolor}
\usepackage{caption}
\usepackage{enumitem}	

\allowdisplaybreaks

\theoremstyle{definition}\newtheorem{definition}{Definition}[section]

\theoremstyle{plain}\newtheorem{thm}{Theorem}
\newtheorem{lem}{Lemma}[section]
\newtheorem{cor}[lem]{Corollary}
\newtheorem{prop}[lem]{Proposition}

\theoremstyle{remark}\newtheorem{remark}{Remark}

\newcommand{\lemit}[1]{\begin{enumerate}[label={(\alph*)}, ref={\thelem\alph*}]{#1}\end{enumerate}}	
\newcommand{\remit}[1]{\begin{enumerate}[label={(\alph*)}, ref={\theremark\alph*}]{#1}\end{enumerate}}

\newcommand{\bo}{{\tilde \beta}}

\renewcommand{\hat}[1]{\widehat{#1}}											
\renewcommand{\tilde}[1]{\widetilde{#1}}										

\newcommand{\ls}{\lesssim}													
\newcommand{\gs}{\gtrsim}													

\newcommand{\Vp}{V^\parallel}
\newcommand{\phe}{\varphi^\varepsilon}
\newcommand{\chie}{\chi^\varepsilon}

\newcommand{\z}{z}															

\newcommand{\wb}{w_\beta}
\newcommand{\bb}{b_\beta}													

\newcommand{\na}{\nabla}														

\newcommand{\pp}{p^\Phi}														
\newcommand{\pc}{p^{\chie}}														
\newcommand{\qp}{q^\Phi}														
\newcommand{\qc}{q^{\chie}}														

\newcommand{\gb}{g_\bo}	
\newcommand{\gbot}{g_\bo^{(12)}}
\newcommand{\gbij}{g_\bo^{(ij)}}											
\newcommand{\fb}{f_\bo}	
\newcommand{\fbot}{f_\bo^{(12)}}

\newcommand{\fboth}{f_\bo^{(13)}}		
\newcommand{\fblm}{f_\bo^{(lm)}}		
\newcommand{\fblk}{f_\bo^{(lk)}}	
\newcommand{\fbrs}{f_\bo^{(rs)}}	
\newcommand{\fbij}{f_\bo^{(ij)}}	
\newcommand{\fblms}{f_\bo^{(l'm')}}											

\newcommand{\A}{\mathcal{A}}
\newcommand{\Abar}{\overline{\mathcal{A}}}
\newcommand{\Ao}{{\A_1}}
\newcommand{\Abaro}{{\Abar_1}}
\newcommand{\B}{\mathcal{B}}
\newcommand{\Bbar}{\overline{\mathcal{B}}}
\newcommand{\Bo}{{\B_1}}
\newcommand{\Bbaro}{{\Bbar_1}}
\newcommand{\Cbar}{\overline{\mathcal{C}}}
\newcommand{\Cbaro}{{\Cbar_1}}
\newcommand{\charAbaro}{\mathbbm{1}_\Abaro}
\newcommand{\charAo}{\mathbbm{1}_\Ao}
\newcommand{\charBbaro}{\mathbbm{1}_\Bbaro}
\newcommand{\charBo}{\mathbbm{1}_\Bo}
\newcommand{\charCbaro}{\mathbbm{1}_\Cbaro}
\newcommand{\charAbarox}{\mathbbm{1}_{\overline{\mathcal{A}}_1^x}}
\newcommand{\charAox}{\mathbbm{1}_{\mathcal{A}_1^x}}

\newcommand{\he}{h_\varepsilon}												
\newcommand{\heot}{\he^{(12)}}

\newcommand{\heij}{\he^{(ij)}}

\newcommand{\hbo}{\overline{h}_{\beta_1}}									
\newcommand{\hboot}{\overline{h}_{\beta_1}^{(12)}}

\newcommand{\te}{\Theta_\varepsilon}											
\newcommand{\teij}{\te^{(ij)}}
\newcommand{\teot}{\te^{(12)}}

\newcommand{\tb}{\overline{\Theta}_{\beta_1}}									
\newcommand{\tbot}{\overline{\Theta}_{\beta_1}^{(12)}}

\newcommand{\tbij}{\overline{\Theta}_{\beta_1}^{(ij)}}

\newcommand{\tz}{\overline{\Theta}_0}
\newcommand{\tzot}{\tz^{(12)}}

\newcommand{\hz}{\overline{h}_0}
\newcommand{\hzot}{\hz^{(12)}}

\newcommand{\hb}{\overline{h}_\bo}
\newcommand{\hbot}{\hb^{(12)}}

\newcommand{\lr}[1]{\left\langle #1 \right\rangle} 							
\newcommand{\llr}[1]{\left\llangle #1 \right\rrangle}							
\newcommand{\norm}[1]{\lVert#1\rVert}   										
\newcommand{\onorm}[1]{\lVert#1\rVert_\mathrm{op}}							

\newcommand{\R}{\mathbb{R}}													

\newcommand\mydots{,\makebox[1em][c]{.\hfil.\hfil.},}							
\newcommand\mycdots{\makebox[1em][c]{$\cdot$\hfil$\cdot$\hfil$\cdot$}}		

\newcommand{\Tr}{\mathrm{Tr}}
\renewcommand{\d}{\mathop{}\!\mathrm{d}}

\renewcommand{\i}{\mathrm{i}}

\DeclareMathOperator*{\supp}{\mathrm{supp}}

\makeatletter
\DeclareFontFamily{OMX}{MnSymbolE}{}
\DeclareSymbolFont{MnLargeSymbols}{OMX}{MnSymbolE}{m}{n}
\SetSymbolFont{MnLargeSymbols}{bold}{OMX}{MnSymbolE}{b}{n}
\DeclareFontShape{OMX}{MnSymbolE}{m}{n}{
    <-6>  MnSymbolE5
   <6-7>  MnSymbolE6
   <7-8>  MnSymbolE7
   <8-9>  MnSymbolE8
   <9-10> MnSymbolE9
  <10-12> MnSymbolE10
  <12->   MnSymbolE12
}{}
\DeclareFontShape{OMX}{MnSymbolE}{b}{n}{
    <-6>  MnSymbolE-Bold5
   <6-7>  MnSymbolE-Bold6
   <7-8>  MnSymbolE-Bold7
   <8-9>  MnSymbolE-Bold8
   <9-10> MnSymbolE-Bold9
  <10-12> MnSymbolE-Bold10
  <12->   MnSymbolE-Bold12
}{}

\let\llangle\@undefined
\let\rrangle\@undefined
\DeclareMathDelimiter{\llangle}{\mathopen}%
                     {MnLargeSymbols}{'164}{MnLargeSymbols}{'164}
\DeclareMathDelimiter{\rrangle}{\mathclose}%
                     {MnLargeSymbols}{'171}{MnLargeSymbols}{'171}
\makeatother

 \newcommand\smallO[1]{
        \mathchoice
            {
                \ensuremath{\mathop{}\mathopen{}{\scriptstyle\mathcal{O}}\mathopen{}\left(#1\right)}
            }
            {
                \ensuremath{\mathop{}\mathopen{}{\scriptstyle\mathcal{O}}\mathopen{}\left(#1\right)}
            }
            {
                \ensuremath{\mathop{}\mathopen{}{\scriptscriptstyle\mathcal{O}}\mathopen{}\left(#1\right)}
            }
            {
                \ensuremath{\mathop{}\mathopen{}{o}\mathopen{}\left(#1\right)}
            }
    }

\title{Derivation of the 1d Gross--Pitaevskii equation from the 3d quantum many-body dynamics of strongly confined bosons}
\author{Lea Boßmann\thanks{Fachbereich Mathematik, Eberhard Karls Universität Tübingen\newline
	\indent\hspace{4pt} Auf der Morgenstelle 10, 72076 Tübingen, Germany\newline
	\indent\hspace{4pt} E-mail: lea.bossmann@uni-tuebingen.de, stefan.teufel@uni-tuebingen.de}~ 
	and Stefan Teufel\footnotemark[1]}
\date{\today}

\begin{document}
\maketitle

\begin{abstract}\noindent
We consider the dynamics of $N$ interacting bosons initially forming a Bose--Einstein condensate. Due to an external trapping potential, the bosons are strongly confined in two dimensions, where the transverse extension of the trap is of order $\varepsilon$. 
The non-negative interaction potential is scaled such that its range and its scattering length  are both  of order $(N/\varepsilon^2)^{-1}$, corresponding to the Gross--Pitaevskii scaling of a dilute Bose gas.
We show that in the simultaneous limit $N\rightarrow\infty$ and $\varepsilon\rightarrow 0$, the dynamics preserve condensation and the time evolution is asymptotically described by a Gross--Pitaevskii equation in one dimension. The strength of the nonlinearity is given by the scattering length of the unscaled interaction, multiplied with a factor depending on the shape of the confining potential. For our analysis, we adapt a method by Pickl \cite{pickl2015} to the problem with dimensional reduction and rely on  the derivation of the one-dimensional NLS equation for interactions with softer scaling behaviour  in \cite{NLS}. 
\end{abstract}

\section{Introduction}
We consider $N$ identical bosons in $\R^3$ interacting through a repulsive pair interaction. The bosons are trapped within a cigar-shaped potential, which effectively confines the particles in two directions to a region of order $\varepsilon$.  
Using the coordinates 
$$\z=(x,y)\in\R^{1+2}\,,$$
the confinement in the $y$-directions is generated by a scaled potential 
$\frac{1}{\varepsilon^2}V^\perp\left(\tfrac{y}{\varepsilon}\right)$, where $V^\perp:\R^2\rightarrow\R$ and    $0<\varepsilon\ll 1$.    
The Hamiltonian describing the system is
\begin{equation}\label{H:micro:coord}
H(t)=\sum\limits_{j=1}^N\left(-\Delta_j+\frac{1}{\varepsilon^2}V^\perp\left(\frac{y_j}{\varepsilon}\right)+\Vp(t,\z_j)\right)+\sum\limits_{1\leq i<j\leq N}w_\mu(\z_i-\z_j),
\end{equation}
where $\Delta$ denotes the Laplace operator on $\R^3$ and $\Vp$ is an additional unscaled external potential. The units are chosen such that $\hbar=1$ and $m=\frac12$. 

The interaction between the particles is described by the potential 
\begin{equation}\label{interaction}
w_\mu(\z)=\mu^{-2} \,w\left(\frac{\z}{\mu}\right)\quad \mbox{ with }\quad \mu := \frac{\varepsilon^2}{N}  
\end{equation}
and for some compactly supported, spherically symmetric, non-negative potential $w$. This scaling of the   interaction describes a dilute gas in the Gross--Pitaevskii regime, which will be explained in detail below.

We are interested in the dynamics of the system in the simultaneous limit $(N,\varepsilon)\rightarrow(\infty,0)$.
The state $\psi^{N,\varepsilon}(t)$ of the system at time $t$ is given as the solution of the $N$-body Schrödinger equation
\begin{equation}\label{SE}
\i\tfrac{\d}{\d t}\psi^{N,\varepsilon}(t)=H(t)\psi^{N,\varepsilon}(t)
\end{equation}
with initial datum $\psi^{N,\varepsilon}(0)=\psi^{N,\varepsilon}_0\in L^2_+(\R^{3N}):=\otimes_\mathrm{sym}^N L^2(\R^3).$
We assume that the bosons initially form a Bose--Einstein condensate. Mathematically, this means that the one-particle reduced density matrix $\gamma^{(1)}_{\psi_0^{N,\varepsilon}}$ of $\psi^{N,\varepsilon}_0$, 
\begin{equation}\label{eqn:k:particle:RDM}
\gamma^{(k)}_{\psi_0^{N,\varepsilon}}:=\Tr_{k+1\mydots N}|\psi_0^{N,\varepsilon}\rangle\langle\psi_0^{N,\varepsilon}|
\end{equation}
for $k=1$, is asymptotically close to a projection $|\phe_0\rangle\langle\phe_0|$ onto a one-body state $\phe_0$. Because of the strong confinement, this condensate state factorises at low energies and is of the form  $\phe_0(z)=\Phi_0(x)\chie(y)\in L^2(\R^3)$ (see Remark \ref{rem:LSSY}). Here, $\Phi_0$ denotes the wavefunction along the $x$-axis and $\chie$ is the normalised ground state of $-\Delta_y+\frac{1}{\varepsilon^2}V^\perp(\frac{y}{\varepsilon})$. Due to the rescaling by $\varepsilon$, $\chie$ is given by
\begin{equation}\label{eqn:chie}
\chie(y)=\tfrac{1}{\varepsilon}\chi(\tfrac{y}{\varepsilon}),
\end{equation}
where $\chi$ is the normalised ground state of $-\Delta_y+V^\perp(y)$. 

In Theorem~\ref{thm}, we show that if the system initially condenses into a factorised state, i.e.\
$$\lim\limits_{(N,\varepsilon)\rightarrow(\infty,0)}\Tr_{L^2(\R^{3})}\Big|\gamma^{(1)}_{\psi^{N,\varepsilon}_0}-|\phe_0\rangle\langle\phe_0| \Big|=0$$
with $\phe_0=\Phi_0\chie$ and $\Phi_0\in H^2(\R)$ (where  the limit $(N,\varepsilon)\rightarrow(\infty,0)$ is taken in an appropriate way), then the condensation into a factorised state is preserved by the dynamics, i.e.\ for all $t\in\R$ and $k\in\mathbb{N}$
$$\lim\limits_{(N,\varepsilon)\rightarrow(\infty,0)}\Tr_{L^2(\R^{3k})}\Big|\gamma^{(k)}_{\psi^{N,\varepsilon}(t)}-|\phe(t)\rangle\langle\phe(t)|^{\otimes k}\Big|=0$$
with $\phe(t)=\Phi(t)\chie$. Moreover, $\Phi(t)$ is the solution of the one-dimensional Gross--Pitaevskii equation
\begin{equation}\label{NLS}
\i\tfrac{\partial}{\partial t}\Phi(t,x)=\left(-\tfrac{\partial^2}{\partial x^2}+\Vp(t,(x,0))+b|\Phi(t,x)|^2\right)\Phi(t,x)=:h(t)\Phi(t,x)
\end{equation}
with $\Phi(0)=\Phi_0 $ and
$$b=8\pi a\int_{\R^2}|\chi(y)|^4\d y=8\pi a \,\varepsilon^2\int_{\R^2}|\chie(y)|^4\d y,$$
where $a$ denotes the scattering length of the unscaled potential $w$. \\

To prove Theorem~\ref{thm}, we follow the approach developed by Pickl for the problem without strong confinement \cite{pickl2015}, which is outlined in Section~\ref{sec:proof}. 
To handle the   singular scaling of the interaction, he first shows the convergence for interactions with softer (but still singular) scaling behaviour, and as a second step uses this result to prove the Gross--Pitaevskii case.

The derivation of the one-dimensional NLS equation for softer scalings of the interaction combined with dimensional reduction was done in \cite{NLS}. In the present paper, we extend the result  from \cite{NLS} to treat the Gross--Pitaevskii regime. As in \cite{NLS}, the strong asymmetry of the problem requires non-trivial adjustments to the method by Pickl.
A description of the differences between our proof and \cite{pickl2015} is given in Remark \ref{rem:differences}.

In the remaining part of the introduction, we will first motivate the scaling \eqref{interaction} of the interaction. This scaling is physically relevant since, written in suitable coordinates, it describes an $(N,\varepsilon)$-independent interaction. 
Subsequently, we comment on related literature.

We wish to study $N$ three-dimensional bosons in an asymmetric trap, which confines in two directions to a length scale $L^\perp$ that is much smaller then the length scale $L^\parallel$ of the remaining direction\footnote{In this paragraph, the capital letters $L^\parallel$, $L^\perp$ and $A$ indicate length scales. In Theorem~\ref{thm} and the remainder of the paper, we use units where $L^\parallel=1$.}.
Hence, we have
$$L^\perp=\varepsilon L^\parallel$$
with $\varepsilon\ll 1$. 
The transverse confinement on the scale $L^\perp$ is achieved by the  potential $\tfrac{1}{(L^\perp)^2}V^\perp(\tfrac{\cdot}{L^\perp})$, where   $-\Delta+V^\perp$ is assumed to have a localised ground state.
In the remaining direction, the system is assumed to be localised in a region of length $L^\parallel$. The particle density is thus
$$\varrho_\mathrm{3d}\sim\tfrac{N}{L^\parallel (L^\perp)^2}=\tfrac{N}{\varepsilon^2 (L^\parallel)^3}\,.$$
To observe Gross--Pitaevskii dynamics in the longitudinal direction in the limit $(N,\varepsilon)\rightarrow(\infty,0)$, we require the kinetic energy per particle in this direction, $E_\mathrm{kin,p.p.}\sim (L^\parallel)^{-2}$, to remain comparable to the total internal energy per particle, i.e.~the total energy without the contributions from the confinement. For a dilute gas, the latter is given by $E_\mathrm{p.p.}\sim A\varrho_\mathrm{3d}$ \cite[Chapter 2]{LSSY}, where $A$ denotes the ($s$-wave) scattering length of the interaction. 
The physical significance of this parameter is the following: the scattering of a slow and sufficiently distant particle at some other particle is to leading order described by its scattering at a hard sphere with radius $A$. Consequently, the length scale determined by $A$ is the relevant length scale for the two-body correlations.
The condition $E_\mathrm{kin,p.p.}\sim E_\mathrm{p.p.}$ implies the scaling condition
\begin{equation}\label{GP:scaling:condition}
\tfrac{A}{L^\parallel}\sim\tfrac{\varepsilon^2}{N}.
\end{equation}
It seems physically reasonable to fix $A\sim1$ since $A$ describes the two-body scattering process and should therefore be independent of $N$ and $\varepsilon$. We will call this choice the microscopic frame of reference. By \eqref{GP:scaling:condition}, the length scales of the problem with respect to this frame are given by $L^\parallel=\tfrac{N}{\varepsilon^2}$ and $L^\perp=\tfrac{N}{\varepsilon}$, hence  both tend to infinity as $(N,\varepsilon)\rightarrow(\infty,0)$. 
$\varrho_\mathrm{3d}$ is of order $\varepsilon^4N^{-2}$ and converges to zero, which shows that we indeed consider a dilute gas. A useful characterisation of the low density regime is the requirement that the mean (three-dimensional) inter-particle distance $\varrho_\mathrm{3d}^{-\frac13}$ be much larger than the scattering length, i.e.~$A^3\varrho_\mathrm{3d}\rightarrow0$. The gas is also dilute with respect to the one-dimensional density $\varrho_\mathrm{1d}\sim\tfrac{N}{L^\parallel}$ because $A\varrho_\mathrm{1d}\sim\varepsilon^2\rightarrow0$, where $\varrho_{1d}^{-1}$ describes the mean one-dimensional inter-particle distance.

For the mathematical analysis, we follow the common practice  to choose coordinates where the longitudinal length scale $L^\parallel=1$ is fixed. Consequently, $L^\perp=\varepsilon$ and the scattering length shrinks as $A=a\,\frac{\varepsilon^2}{N}$. 
This frame of reference arises from the microscopic frame by the coordinate rescaling $\z\mapsto \frac{\varepsilon^2}{N}z$ and $t\mapsto (\tfrac{\varepsilon^2}{N})^2t$ in the Schrödinger equation \eqref{SE}, which yields the rescaled interaction \eqref{interaction}. 
Note that times of order one with respect to this frame correspond to extremely long times on the microscopic time scale, which relates to the low density of the gas.

We admit an external field $\Vp$ varying on the length scale $L^\parallel$. Consequently, it depends on $(N,\varepsilon)$ with respect to the microscopic frame of reference and is $(N,\varepsilon)$-independent in our coordinates.
As $L^\parallel\gg A$, the external potential is asymptotically constant on the scale of the interaction and therefore does not affect the scaling condition \eqref{GP:scaling:condition}.

Due to this scaling condition, the system always remains within the second of the five regions defined by Lieb, Seiringer and Yngvason in \cite{lieb2004}. In that paper, the authors prove that the ground state energy and density of a dilute Bose gas in a highly elongated trap can be obtained by minimising the energy functional corresponding to the Lieb--Liniger Hamiltonian with coupling constant $g=\frac{A}{\varepsilon^2}\int|\chi(y)|^4\d y$ \cite[Theorem~1.1]{lieb2004}. If $g\overline{\varrho}^{-1}\rightarrow0$, where $\overline{\varrho}$ denotes the mean one-dimensional density, the system can be described as one-dimensional limit of a three-dimensional effective theory. In particular, if $g\overline{\varrho}^{-1}\sim N^{-2}$, which is true for our system due to \eqref{GP:scaling:condition}, the ground state is described by a one-dimensional Gross--Pitaevskii energy functional \cite[Theorem~2.2]{lieb2004}. The other regions can be reached by scaling $A$  differently.\footnote{Let us assume that the external field $\Vp$ is given by a homogeneous function of degree $s>0$ acting only in the $x$-direction. The ideal gas case (region 1) is then obtained by the scaling $A\ll \varepsilon^2N^{-1}$ and the Thomas--Fermi case (region 3) by choosing $\varepsilon^2N^{-1}\ll A\ll \varepsilon^2 N^{\frac{s}{s+2}}$. Also the truly one-dimensional regime can be reached: $A\sim\varepsilon^2N^{\frac{s}{s+2}}$ corresponds to region 4 and $A\gg\varepsilon^2N^{\frac{s}{s+2}}$ yields a Girardeau--Tonks gas (region 5).}

It is also instructive to consider softer scaling interactions of the form 
\begin{equation}\label{eqn:small:beta}
\wb(\z):=(\tfrac{N}{\varepsilon^2})^{-1+3\beta}w\left((\tfrac{N}{\varepsilon^2})^\beta\z\right),
\end{equation}
where the scaling parameter $\beta\in(0,1)$ interpolates between the Hartree ($\beta=0$) and the Gross--Pitaevskii ($\beta=1$) regime. In this case, the scattering length still scales as $(\frac{N}{\varepsilon^2})^{-1}$ \cite[Lemma~A.1]{erdos2007} whereas the effective range of $\wb$ is now of order $(\frac{N}{\varepsilon^2})^{-\beta}$. This means that as $(N,\varepsilon)\rightarrow(\infty,0)$, the scattering length becomes negligible compared to the range of the interaction, i.e.~the two-body correlations become invisible on the length scale of the interaction. Consequently, the scattering length is well approximated by the first order Born approximation and the corresponding effective equation is  the one-dimensional NLS equation \eqref{NLS} with $b$ replaced by  $\norm{w}_{L^1(\R^3)}\int_{\R^2}|\chi(y)|^4\d y$ \cite{NLS}.
\\

Quasi one-dimensional bosons in highly elongated traps have been experimentally probed \cite{gorlitz2001,henderson2009} and the dynamics of such systems are physically very interesting \cite{esteve2006,kinoshita2006,meinert2017}. The first rigorous derivation of NLS and Gross--Pitaevskii equations for three-dimensional bosons using BBGKY hierarchies is due to Erd{\H o}s, Schlein and Yau \cite{erdos2007, erdos2010}. A different approach was  proposed by Pickl \cite{pickl2008,pickl2010,
pickl2015, jeblick2018}, who 
also obtained  rates for the convergence of the reduced density matrices. A third method for the Gross--Pitaevskii case, using Bogoliubov transformations and coherent states on Fock space, was proposed by Benedikter, De Oliveira and Schlein \cite{benedikter2015}. Extending this approach, Brennecke and Schlein \cite{brennecke2017} recently proved an optimal rate of the convergence. 
Several further results concern bosons in one \cite{adami2007, chen2016} and two \cite{kirkpatrick2011,jeblick2016,jeblick2017} dimensions.
The problem of dimensional reduction for the NLS equation was treated by M{\'e}hats and Raymond \cite{mehats2017}, who study the cubic NLS equation in a quantum waveguide.
In \cite{abdallah2005_2}, Ben Abdallah, Méhats, Schmeiser and Weishäupl consider an $(n + d)$-dimensional NLS equation subject to a strong confinement in $d$ directions and derive an effective $n$-dimensional NLS evolution.

There are few works on the derivation of lower-dimensional time-dependent NLS equations from the three-dimensional $N$-body dynamics. Chen and Holmer consider three-dimensional bosons with pair interactions in a strongly confining potential in one \cite{chen2013} and two \cite{chen2017} directions. For repulsive interactions scaling with $\beta\in (0,\frac25)$ in case of a disc-shaped and for attractive interactions with $\beta\in(0,\frac37)$ in case of a cigar-shaped confinement, they show that the dynamics are effectively described by two- and one-dimensional NLS equations. 
In \cite{keler2016}, von Keler and Teufel prove this for a Bose gas which is confined to a quantum waveguide with non-trivial geometry for $\beta\in(0,\tfrac13)$. 
In \cite{NLS}, Boßmann considers bosons interacting through a potential scaling with $\beta\in(0,1)$, but apart from this in the same setting as here, and shows that the evolution of the system is well captured by a one-dimensional NLS equation.\\

\noindent{\bf Notation.} We use the notation $A\ls B$ to indicate that there exists a constant $C>0$ independent of $\varepsilon, N, t, \psi^{N,\varepsilon}_0,\Phi_0$ such that $A\leq CB$. This constant may, however, depend on the quantities fixed by the model, such as $V^\perp$, $\chi$ and $\Vp$.
Besides, we will exclusively use the symbol $\,\hat{\cdot}\,$ to denote the weighted many-body operators from Definition~\ref{def:hat} (see also Remark~\ref{rem:notation}) and use the abbreviations $$\llr{\cdot,\cdot}:=\lr{\cdot,\cdot}_{L^2(\R^{3N})},\quad \norm{\cdot}:=\norm{\cdot}_{L^2(\R^{3N})}
\quad\text{and}\quad\onorm{\cdot}:=\norm{\cdot}_{\mathcal{L}(L^2(\R^{3N}))}.$$

\section{Main Result}\label{sec:main}
To study the effective dynamics of the many-body system in the limit $(N,\varepsilon)\rightarrow(\infty,0)$, we consider families of initial data $\psi^{N,\varepsilon}_0$ along the following sequences $(N_n,\varepsilon_n)\rightarrow(\infty,0)$:
\begin{definition}\label{def:admissible}
A sequence $(N_n,\varepsilon_n)$ in $\mathbb{N}\times(0,1)$ is called \textit{admissible} if
$$\lim\limits_{n\rightarrow\infty}(N_n,\varepsilon_n)=(\infty,0) \qquad\text{ and }\qquad 
\lim\limits_{n\rightarrow\infty}\tfrac{\varepsilon^{2+\delta}_n}{\mu_n}
=0
\quad\text{ for }\mu_n:=\left(\tfrac{N_n}{\varepsilon^2_n}\right)^{-1}$$
for some $0<\delta<\frac25$.
\end{definition}
The second condition ensures that the energy gap of order $\varepsilon^{-2}$ above the transverse ground state $\chie$ grows sufficiently fast. 
In the proof, this will be used to control transverse excitations into states orthogonal to $\chie$ (see also Remark~\ref{rem:admissibility}).
Since
$$\tfrac{\varepsilon^{2+\delta}}{\mu}=N\varepsilon^\delta\to0,$$ 
$\delta$ must be strictly positive, otherwise $N\varepsilon^\delta\rightarrow0$ would be impossible. 

To formulate our main theorem, we need two different one-particle energies:
\begin{itemize}
\item The \emph{``renormalised'' energy per particle}: for $\psi\in\mathcal{D}(H(t)^\frac12)$,
\begin{equation}\label{E^psi}
E^\psi(t):=\tfrac{1}{N}\llr{\psi,H(t)\psi}-\tfrac{E_0}{\varepsilon^2},
\end{equation}
where $E_0$ denotes the lowest eigenvalue of $-\Delta_y+V^\perp(y)$. By rescaling, the lowest eigenvalue of $-\Delta_y+\frac{1}{\varepsilon^2}V^\perp(\frac{y}{\varepsilon})$ is $\frac{E_0}{\varepsilon^2}$.
\item The \emph{effective energy per particle}: for $\Phi\in H^1(\R)$,
\begin{equation}\label{E^Phi}
\mathcal{E}^\Phi(t):=\lr{\Phi,\left(-\tfrac{\partial^2}{\partial x^2}+\Vp(t,(x,0))+\tfrac{b}{2}|\Phi|^2\right)\Phi}_{L^2(\R)}.
\end{equation}
\end{itemize}
Further, define the function $\mathfrak{e}:\R\rightarrow [1,\infty)$ by
\begin{equation}\label{def:e}
\mathfrak{e}^2(t):=1+|E^{\psi^{N,\varepsilon}_0}(0)|+|\mathcal{E}^{\Phi_0}(0)|+\int\limits_0^t \norm{\dot{\Vp}(s,\cdot)}_{L^\infty(\R^3)}\d s
+\sup\limits_{\substack{i,j\in\{0,1\}\\k\in\{1,2\}}}\norm{\partial_t^i\partial_{y_k}^j\Vp(t,\cdot)}_{L^\infty(\R^3)}\,.
\end{equation}
Note that $\mathfrak{e}(t)$ is for each $t\in\R$ uniformly bounded in $N$ and $\varepsilon$ because we will assume that $E^{\psi^{N,\varepsilon}_0}(0)\to \mathcal{E}^{\Phi_0}(0)$ as $(N,\varepsilon)\to(\infty,0)$ (see assumption A4 below)
and boundedness of $\Vp$ and its derivatives (see assumption A3).
The function $\mathfrak{e}$ will be useful because, by the fundamental theorem of calculus,
\begin{equation}
\big|E^{\psi^{N,\varepsilon}(t)}(t)\big|\leq \mathfrak{e}^2(t)-1 \quad\text{ and }\quad \big|\mathcal{E}^{\Phi(t)}(t)\big|\leq \mathfrak{e}^2(t)-1 
\end{equation}
for any $t\in\R$. Note that for a time-independent external field $\Vp$, it follows that $\mathfrak{e}^2(t)\ls 1$ for any $t$, hence $E^{\psi^{N,\varepsilon}(t)}(t)$  and $\mathcal{E}^{\Phi(t)}(t)$ are in this case bounded uniformly in time.
\\

\noindent Let us now state our assumptions.
\begin{itemize}
\item[A1] \emph{Interaction.} Let the unscaled interaction $w\in L^\infty(\R^3,\R)$ be spherically symmetric, non-negative and let $\supp w\subseteq \{z\in\R^3:|z|\leq 1\}$.
\item[A2] \emph{Confining potential.} Let $V^\perp:\R^2\rightarrow\R$ such that $-\Delta_y+V^\perp$ is self-adjoint and has a non-degenerate ground state $\chi$ with energy $E_0<\inf\sigma_\mathrm{ess}(-\Delta_y+V^\perp)$. 
Assume that the negative part of $V^\perp$ is bounded and that $\chi\in\mathcal{C}^2_\mathrm{b}(\R^2)$, i.e.~$\chi$ is bounded and twice continuously differentiable with bounded derivatives. We choose  $\chi$ normalised and real. 
\item[A3] \emph{External field.} Let $\Vp:\R\times\R^3\rightarrow\R$ such that for fixed $z\in\R^3$, $\Vp(\cdot,z)\in\mathcal{C}^1(\R)$. 
Further, assume that for each fixed $t\in \R$, $\Vp(t,(\cdot,0))\in H^4(\R)$, $\Vp(t,\cdot),\dot{\Vp}(t,\cdot)\in L^\infty(\R^3)\cap \mathcal{C}^1(\R^3)$ and $\nabla_y\Vp(t,\cdot),\nabla_y\dot{\Vp}(t,\cdot)\in L^\infty(\R^3)$.
\item[A4] \emph{Initial data.} Assume that the family of initial data, $\psi^{N,\varepsilon}_0\in\mathcal{D}(H(0))\cap L^2_+(\R^{3N})$ with $\norm{\psi^{N,\varepsilon}_0}^2=1$, is close to a condensate with condensate wavefunction $\phe_0=\Phi_0\chie$ for some normalised $\Phi_0\in H^2(\R)$ in the following sense: for some admissible sequence $(N,\varepsilon)$, it holds that
\begin{equation}\label{A4:1}
\lim\limits_{(N,\varepsilon)\rightarrow(\infty,0)}\Tr_{L^2(\R^3)}\Big|\gamma^{(1)}_{\psi^{N,\varepsilon}_0}-|\Phi_0\chie\rangle\langle\Phi_0\chie|\Big|=0
\end{equation}
and
\begin{equation}\label{A4:2}
\lim\limits_{(N,\varepsilon)\rightarrow(\infty,0)}\left|E^{\psi^{N,\varepsilon}_0}(0)-\mathcal{E}^{\Phi_0}(0)\right|=0.
\end{equation}
\end{itemize}

\begin{thm}\label{thm}
Assume that $w$, $V^\perp$ and $\Vp$ satisfy A1 -- A3. Let $\psi^{N,\varepsilon}_0$ be a family of initial data satisfying A4, let $\psi^{N,\varepsilon}(t)$ denote the solution of \eqref{SE} with initial datum $\psi^{N,\varepsilon}(0)=\psi^{N,\varepsilon}_0$ and let $\gamma^{(k)}_{\psi^{N,\varepsilon}(t)}$ denote its $k$-particle reduced density matrix as in \eqref{eqn:k:particle:RDM}. 
Then for any $T\in\R$ and $k\in\mathbb{N}$,
\begin{equation}\label{T1}
\lim\limits_{(N,\varepsilon)\rightarrow(\infty,0)}\,\sup\limits_{t\in[-T,T]}
\Tr_{L^2(\R^{3k})}\Big|\gamma^{(k)}_{\psi^{N,\varepsilon}(t)}-|\Phi(t)\chie\rangle\langle\Phi(t)\chie|^{\otimes k}\Big|=0
\end{equation}
and
\begin{equation}\label{T2}
\lim\limits_{(N,\varepsilon)\rightarrow(\infty,0)}\,\sup\limits_{t\in[-T,T]}\left|E^{\psi^{N,\varepsilon}(t)}(t)-\mathcal{E}^{\Phi(t)}(t)\right|=0,
\end{equation}
where  $\Phi(t)$ is the solution of \eqref{NLS} with initial datum $\Phi(0)=\Phi_0$ and with
\begin{equation}\label{b}
b=8\pi a \int\limits_{\R^2}|\chi(y)|^4\d y\,.
\end{equation}
Here, $a$ denotes the scattering length of $w$ and the limits in \eqref{T1} and \eqref{T2} are taken along the sequence from A4. 
\end{thm}

\begin{remark}
\remit{
	\item Assumption~A4 differs from the corresponding statement in \cite{NLS} in that we impose a weaker admissibility condition than the condition $\varepsilon^2/\mu\rightarrow0$ from \cite{NLS}, which cannot hold for $\beta=1$.
	\item A2 is fulfilled, e.g., by a harmonic potential or by any smooth potential with at least one bound state below the essential spectrum. According to \cite[Theorem~1]{griesemer2004}, A2 implies that the ground state $\chi$ of $-\Delta_y+V^\perp$ decays exponentially. Thus, $\chi^\varepsilon$ is indeed exponentially localised on a scale of order $\varepsilon$.
The regularity condition on $\Vp(t,(\cdot,0))$ is needed to ensure the global existence of $H^2$ solutions of \eqref{NLS} (see \cite[Appendix~A]{NLS}).
Due to assumptions A1--A3, the operators $H(t)$ are for any $t\in\R$ self-adjoint on the time-independent domain $\mathcal{D}(H)$ and generate a strongly continuous unitary evolution on $\mathcal{D}(H)$. 
	\item In \cite{lieb2004}, it is shown that the ground state of $H(0)$ with a homogeneous external field $\Vp(z,0)$ satisfies assumption A4 (Theorem~2.2 and Theorem~5.1). Note that to observe non-trivial dynamics in this case, it is important that we admit a time-dependent external potential $\Vp$.
	\label{rem:LSSY}
	\item Our proof yields an estimate of the rate of convergence of \eqref{T1}, which is given in Corollary~\ref{cor:rates}. This rate is not uniform in time but, contrarily, depends on it in form of a double exponential. 
	\item Our result is restricted to sequences where $\varepsilon^\delta\ll N^{-1}$ for some $\delta\in(0,\frac25)$ (Assumption~A4). Similar conditions appear also in comparable works \cite{NLS, chen2013, chen2017} for $\beta<1$. However, for the ground state analysis in~\cite{lieb2004}, no analogue of this admissibility condition is required. On a formal level, together with the result of the strong confinement limit of the three-dimensional NLS in~\cite{abdallah2005_2}, this suggests that our dynamical result could be extended to hold without imposing a condition on the rate of convergence of $\varepsilon$. As remarked before, in our proof this condition is crucial to control the transverse excitations by an a priori energy estimate. A possible approach to weaken the condition might be to replace the transverse ground state $\chie$ of the linear operator $-\Delta_y+\tfrac{1}{\varepsilon^2} V^\perp(\frac{\cdot}{\varepsilon})$ by the $x$-dependent ground state of the nonlinear functional 
	$$\lr{\tilde{\chie}(x,\cdot),\left(-\Delta_y+\tfrac{1}{\varepsilon^2}V^\perp(\tfrac{\cdot}{\varepsilon})+\varepsilon^2 8\pi a|\Phi(x)|^2|\tilde{\chie}(x,\cdot)|^2\right)\tilde{\chie}(x,\cdot)}_{L^2(\R^2)} $$
and to prove the smallness of transverse excitations by  adiabatic-type arguments.\label{rem:admissibility}
	\item We expect that our proof can be extended to cover systems that are trapped to quantum waveguides with non-trivial geometry as in~\cite{keler2016}. However, this is not straightforward as a Taylor expansion of the interaction was used in~\cite{keler2016} and the kinetic term now includes an additional vector potential due to the twisting of the waveguide.
	\item Further, we expect the same strategy to be applicable to one-dimensional confining potentials resulting in effectively two-dimensional condensates. The solution of this problem is not obvious since many of our estimates depend on the dimension and cannot be directly transferred. For instance, Green's function is different in two dimensions and the ratio of $N$ and $\varepsilon$ changes (the corresponding effective range is $\mu_\mathrm{2d}=\varepsilon/N$), making some key estimates invalid.
}
\end{remark}

\section{Proof of the main theorem}\label{sec:proof}
To prove Theorem~\ref{thm}, we must show that the expressions in \eqref{T1} and \eqref{T2} vanish in the limit $(N,\varepsilon)\rightarrow(\infty,0)$ for suitable initial data. Instead of directly estimating these differences, we follow the approach of Pickl  \cite{pickl2008,pickl2010,pickl2011,pickl2015}. As one crucial  first step, we  define a functional 
$$\alpha_\xi^<:\R\times L^2(\R^{3N}) \times L^2(\R^3)\to \R\,, \quad(t,\psi^{N,\varepsilon},\phe)\mapsto\alpha_\xi^<(t,\psi^{N,\varepsilon},\phe)$$
measuring the part of $\psi^{N,\varepsilon}$ which has not condensed into $\phe$.
This functional is chosen in such a way that $\alpha_\xi^<(t,\psi^{N,\varepsilon}(t),\phe(t))\rightarrow0$ is equivalent to \eqref{T1} and \eqref{T2}.
While we  roughly follow  \cite{pickl2015}, the strong asymmetry of the setup and the more singular scaling of the interaction require a non-trivial adaptation of the formalism. We also heavily rely  on the result in \cite{NLS} for the case $\beta\in(0,1)$.
The functional $\alpha^<_\xi$ is constructed as follows:

\begin{definition}\label{def:p}
Let $\varphi\in L^2(\R^3)$ be of the form $\varphi(z) =\Phi(x)\chi(y)$ for some $\Phi\in L^2(\R)$ and $\chi\in L^2(\R^2)$ and let
$$p^\varphi:=|\varphi\rangle\langle\varphi| \quad\mbox{and}\quad q^\varphi:=\mathbbm{1}-p^\varphi
\quad \in\mathcal{L}\left(L^2(\R^3)\right).$$
Further, define the orthogonal projections on $L^2(\R^3)$ 
\begin{align*}
\pp&:=|\Phi\rangle\langle\Phi| \otimes\mathbbm{1}_{L^2(\R^2)},  &  \qp&:=\mathbbm{1}_{L^2(\R^3)}-\pp,\\
p^\chi&:=\mathbbm{1}_{L^2(\R)}\otimes|\chi\rangle\langle\chi|  , &  q^\chi&:=\mathbbm{1}_{L^2(\R^3)}-p^\chi\,.
\end{align*}
Note that $p^\varphi=\pp p^\chi$, $q^{\Phi/\chi}q^\varphi=q^{\Phi/\chi}$, $q^\varphi=q^\chi+\qp p^\chi$ and $p^{\Phi/\chi}q^\varphi=p^{\Phi/\chi}q^{\chi/\Phi}$. \\[1mm]
These one-body projections are lifted to many-body projections on $L^2(\R^{3N})$ by defining   
$$p_j^\varphi:=\underbrace{\mathbbm{1}\otimes\cdots\otimes\mathbbm{1}}_{j-1}\otimes\, p^\varphi\otimes \underbrace{\mathbbm{1}\otimes\cdots\otimes\cdots\mathbbm{1}}_{N-j} \quad\text{and}\quad q^\varphi_j:=\mathbbm{1}-p^\varphi_j\quad\mbox{ for  $j\in\{1,\dots,N\}$},
$$
and analogously  $\pp_j$, $\qp_j$, $p^\chi_j$ and $q^\chi_j$. We will also write $p^\varphi_j=|\varphi(z_j)\rangle\langle\varphi(z_j)|$. \\[1mm]
Finally, for $0\leq k\leq N$, define the symmetrised many-body projections  
$$ P_k^\varphi=\big(q_1^\varphi\cdots q_k^\varphi p_{k+1}^\varphi\cdots p_N^\varphi\big)_\mathrm{sym}:=\sum\limits_{\substack{J\subseteq\{1,\dots,N\}\\|J|=k}}\prod\limits_{j\in J}q_j^\varphi\prod\limits_{l\notin J}p_l^\varphi $$
and $P^\varphi_k=0$ for $k<0$ and $k>N$. 

\end{definition}

\begin{definition}\label{def:hat}
Let $f: \mathbb{N}_0\rightarrow\R_0^+$ and $d\in\mathbb{Z}$. Using the projections $P_k^\varphi$ from Definition~\ref{def:p}, we define the operators $\hat{f}^\varphi,\hat{f}_d^\varphi\in\mathcal{L}\left(L^2(\R^{3N})\right)$ by 
$$\hat{f}^\varphi:= \sum\limits_{k=0}^N f(k)P_k^\varphi, \qquad \hat{f}_d^\varphi:=\sum\limits_{j=-d}^{N-d} f(j+d)P_j^\varphi.$$
\end{definition}

\begin{definition}\label{def:alpha} 
For $\xi\in(0,\frac12)$, define the functional $$\alpha^<_\xi:\R\times L^2(\R^{3N})\times L^2(\R^3 )\supset \R\times\mathcal{D}(H^{\frac{1}{2}})\times (H^1(\R )\times L^2(\R^2))\rightarrow \R$$ by
$$\alpha^<_\xi(t,\psi,\varphi=\Phi\chi):=\llr{\psi,\hat{m}^\varphi\psi}+\left|E^{\psi}(t)-\mathcal{E}^{\Phi}(t)\right|,$$
where the weight function $m:\mathbb{N}_0\rightarrow\R^+_0$ is given by
$$ m(k):=\begin{cases}
	\sqrt{\frac{k}{N}} & \text{for } k\geq N^{1-2\xi},\vspace{0.1cm}\\
	\tfrac12\left(N^{-1+\xi}k+N^{-\xi}\right) & \text{else}.\end{cases}$$ 
\end{definition}

For simplicity, we will not explicitly indicate the $\xi$-dependence of the weight $m$ in the notation. For the proof of Theorem~\ref{thm}, we will choose some fixed $\xi$ within a suitable range.

The operators $P_k^\varphi$ project onto states with $k$ particles outside the condensate described by $\varphi$. Consequently, $\llr{\psi,\hat{m}^\varphi\psi}$ is a weighted measure of the relative number of such particles in the state $\psi$.
Note that the weight function $m$ is increasing and $m(0)\approx0$, hence only the parts of $\psi$ outside the condensate contribute significantly to $\llr{\psi,\hat{m}^\varphi\psi}$.
For a sequence $(\psi^N)_{N\in\mathbb{N}}$ of $N$-body wavefunctions,~\cite[Lemma~3.2]{NLS}\footnote{Lemma~3.2 in \cite{NLS} collects different statements somewhat scattered in the literature. The respective proofs can be found e.g.~in \cite{keler2016,knowles2010,pickl2011,pickl2015,rodnianski2009}.}
implies that $\llr{\psi^N,\hat{m}^\varphi\psi^N}\rightarrow0$ as $N\to\infty$ is equivalent to the convergence of the one-particle reduced density matrix of $\psi^N$ to $|\varphi\rangle\langle\varphi|$ in trace norm or in operator norm.
Further, convergence of the one-particle reduced density matrix implies convergence of
all $k$-particle reduced density matrices.
This is summarised in the following lemma:

\begin{lem}\label{lem:equivalence}
Let $t\in\R$, $k\in\mathbb{N}$, $\varphi=\Phi\chi\in H^1(\R)\times L^2(\R^2)$ with $\Phi$ and $\chi$ normalised. Let $(\psi^N)_{N\in\mathbb{N}}\subset L^2(\R^{3N})$ be a sequence of normalised $N$-body wavefunctions and denote by $\gamma^{(k)}_{\psi^N}$ the $k$-particle reduced density matrix of $\psi^N$. Then the following statements are equivalent:
\lemit{
	\item $\lim\limits_{N\to\infty}\alpha^<_\xi(t,\psi^N,\varphi)=0$ for some $\xi\in(0,\frac12)$,
	\item $\lim\limits_{N\to\infty}\alpha^<_\xi(t,\psi^N,\varphi)=0$ for any $\xi\in(0,\frac12)$,
	\item $\lim\limits_{N\to\infty}\Tr\Big|\gamma^{(k)}_{\psi^N}-|\varphi\rangle\langle\varphi|^{\otimes k}\Big|=0$ \quad and \;
$\lim\limits_{N\to\infty}\left|E^{\psi^N}(t)-\mathcal{E}^{\Phi}(t)\right|=0 $\quad for all $k\in\mathbb{N}$,
	\item $\lim\limits_{N\to\infty}\Tr\Big|\gamma^{(1)}_{\psi^N}-|\varphi\rangle\langle\varphi|\Big|=0$ \quad and \;
$\lim\limits_{N\to\infty}\left|E^{\psi^N}(t)-\mathcal{E}^{\Phi}(t)\right|=0.$
}
The relation between the rates of convergence of $\alpha_\xi^<(t,\psi^N,\varphi)$ and $\gamma^{(1)}_{\psi^N}$ is 
$$
\Tr\left|\gamma_{\psi^N}^{(1)}-|\varphi\rangle\langle\varphi|\right|\leq \sqrt{8\alpha_\xi^<(t,\psi^N,\varphi)},$$
$$
\alpha_\xi^<(t,\psi^N,\varphi)\leq \left|E^{\psi^N}(t)-\mathcal{E}^{\Phi}(t)\right|+\sqrt{\Tr\left|\gamma_{\psi^N}^{(1)}-|\varphi\rangle\langle\varphi|\right|}+\tfrac12 N^{-\xi}.
$$
\end{lem}

\begin{proof}
\cite{NLS}, Lemma~3.2 and Lemma~3.3.
\end{proof}

To prove Theorem~\ref{thm}, we evaluate the functional $\alpha_\xi^<$ on the solution  $\psi^{N,\varepsilon}(t)$ of \eqref{SE} with initial datum $\psi^{N,\varepsilon}_0$ given by assumption A4, the solution $\Phi(t)$ of the Gross--Pitaevskii equation \eqref{NLS} with initial datum $\Phi_0$ from A4, and the ground state $\chie$ of $-\Delta_y+\tfrac{1}{\varepsilon^2}V^\perp(\tfrac{y}{\varepsilon})$ from A2. 
For simplicity, we will abbreviate
$$\alpha_\xi^<(t):=\alpha_\xi^<\Big(t,\,\psi^{N,\varepsilon}(t),\,\phe(t) = \Phi(t)\chie \Big).
$$
Due to Lemma~\ref{lem:equivalence}, $\alpha^<_\xi(t)\rightarrow0$ is equivalent to \eqref{T1} and \eqref{T2}; conversely, \eqref{A4:1} and \eqref{A4:2} imply $\alpha^<_\xi(0)\rightarrow 0$. Hence, to prove Theorem~\ref{thm}, it suffices to show the convergence of $\alpha^<_\xi(t)\rightarrow0$ for all $t\in\R$.

In \cite{NLS}, the functional $\alpha_\xi^<(t)$ is used as counting measure for the interaction \eqref{eqn:small:beta} scaling with $\beta\in(0,1)$. For the proof in that case, one   first shows an estimate of the kind $|\frac{\d}{\d t}\alpha_\xi^<(t)|\ls\alpha_\xi^<(t)+\smallO{1}$ and subsequently applies Grönwall's inequality, using that $\alpha_\xi^<(0)\rightarrow0$.

For the Gross--Pitaevskii scaling of the interaction, we cannot simply estimate $\frac{\d}{\d t}\alpha_\xi^<(t)$ for $\beta=1$ because this derivative is not controllable with the methods used in \cite{NLS}.
To understand why this is the case, let us first give a heuristic argument why the NLS equation with coupling parameter $\bb=\norm{w}_{L^1(\R^3)}\int_{\R^2}|\chi(y)|^4\d y$ is the right effective description for $\beta\in(0,1)$ but not for $\beta=1$. To this end, we compute the renormalised energy per particle with respect to the trial state $\psi_\mathrm{prod}(t,z_1\mydots z_N)=\phe(t,z_1)\phe(t,z_2)\cdots\phe(t,z_N)$, i.e.~the state where all particles are condensed into the single-particle orbital $\phe(t)$. For simplicity, we will ignore the external potential $\Vp$ and drop the time-dependence of $\phe$ in the notation. Making use of the fact that $\left(-\Delta_y+\tfrac{1}{\varepsilon^2}V^\perp(\tfrac{y}{\varepsilon})-\tfrac{E_0}{\varepsilon^2}\right)\chie(y)=0$ and that $\phe$ is normalised, we obtain
\begin{eqnarray*}
&&\hspace{-0.7cm}\tfrac{1}{N}\big\llangle\psi_\mathrm{prod},H\psi_\mathrm{prod}\big\rrangle-\tfrac{E_0}{\varepsilon^2}\\
&&\hspace{-0.5cm}=\lr{\Phi(x_1),(-\partial_{x_1}^2)\Phi(x_1)}+\tfrac{N-1}{2N}\int\d z_1|\Phi(x_1)|^2|\chi(y_1)|^2\int\d z|\Phi(x_1-\mu^\beta x)|^2|\chi(y_1-\tfrac{\mu^\beta}{\varepsilon}y)|^2w(z)\\
&&\hspace{-0.5cm}\rightarrow \lr{\Phi(x_1),\left(-\partial_{x_1}^2+\tfrac12\left(\int|\chi(y_1)|^4\d y_1 \int w(z)\d z \right)|\Phi(x_1)|^2\right)\Phi(x_1)}\;=\;\mathcal{E}^\Phi_{\beta\in(0,1)}
\end{eqnarray*}
in the limit $(N,\varepsilon)\rightarrow(\infty,0)$, where we have chosen the limiting sequence in such a way that $\tfrac{\mu^\beta}{\varepsilon}\rightarrow0$.\footnote{This condition in \cite{NLS}, called \textit{moderate confinement}, ensures that the extension $\varepsilon$ is always large compared to the range $\mu^\beta=(\tfrac{N}{\varepsilon^2})^{-\beta}$ of the interaction $\wb$. As $\tfrac{\mu^\beta}{\varepsilon}=N^{-\beta}\varepsilon^{2\beta-1}$, this is a restriction only for $\beta<\frac12$; in particular, it is satisfied for $\beta=1$.} Here, $\mathcal{E}^\Phi_{\beta\in(0,1)}$ is the effective energy per particle for $\beta\in(0,1)$, i.e.~it equals \eqref{E^Phi} with $\Vp=0$ and   $b$ replaced by $\bb$.

For the Gross--Pitaevskii scaling $\beta=1$, this very argument yields the same one-particle energy $\mathcal{E}^\Phi_{\beta\in(0,1)}$, which differs from the correct expression \eqref{E^Phi} by an error of $\mathcal{O}(1)$ as $\bb\neq b$. 
The reason for this error is that for $\beta=1$, the scattering length $a_\mu$ of $w_\mu$ is of the same order as its range $\mu$, i.e.~the inter-particle correlations live on the scale of the interaction and thus decrease the energy per particle by an amount of $\mathcal{O}(1)$.

Hence, an initial state $\psi^{N,\varepsilon}_0$ that is a pure product state is excluded by assumption A4.
This reasoning suggests to include the pair correlations in our trial function. To do so, let us first recall the definition of the scattering length: the zero energy scattering equation for the interaction $  w_\mu = \mu^{-2} w(\cdot/\mu)$ is given by 
\begin{equation}\label{eqn:scat}
\begin{cases}
	\left(-\Delta+\tfrac12 w_\mu(\z)\right)j_\mu(\z)=0 	& \text{ for } |\z|<\infty,\\
	\;j_\mu(\z)\rightarrow 1														& \text{ as } |\z|\rightarrow\infty.
\end{cases}
\end{equation}
By \cite[Theorems C.1 and C.2]{LSSY}, the unique solution $j_\mu\in\mathcal{C}^1(\R^3)$ of \eqref{eqn:scat} is spherically symmetric, non-negative, non-decreasing in $|\z|$ and
\begin{equation}\label{eqn:j}
\begin{cases}
	j_\mu(\z)=1-\frac{a_\mu}{|\z|} & \text{ for }|\z|>\mu,\vphantom{\bigg(}\\
	j_\mu(\z)\geq 1-\frac{a_\mu}{|\z|} & \text{ else. }
\end{cases}
\end{equation} 
The number  $a_\mu\in\R$ is by definition  the scattering length of $w_\mu$. Equivalently,
\begin{equation}\label{eqn:integral:scat}
8\pi a_\mu=\int\limits_{\R^3}w_\mu(\z)j_\mu(\z)\d\z.
\end{equation}
By the scaling behaviour of \eqref{eqn:scat}, we obtain
\begin{equation*}
\mu^{-2}\left(-\Delta+\tfrac12 w(\z)\right)j_\mu(\mu\z)=0
\end{equation*}
 for $|\z|<\infty$, hence $j_\mu(z) = j_1(  z/\mu)$ and 
\begin{equation}\label{eqn:a^N,eps}
a_\mu=\mu a,
\end{equation}
where $a$ denotes the scattering length of the unscaled interaction $w= w_1$. 
From \eqref{eqn:j} and \eqref{eqn:a^N,eps}, one immediately concludes that $j_\mu$ differs from one by an error of $\mathcal{O}(1)$ on $\supp w_\mu$. Hence, \eqref{eqn:integral:scat} implies that the first order Born approximation $\tfrac{1}{8\pi}\int w_\mu(z)\d z$ is no valid approximation to the scattering length $a_\mu$ in the Gross--Pitaevskii regime, whereas this approximation was justified for interactions $w_\beta$ as in \eqref{eqn:small:beta} with $\beta\in(0,1)$.

For practical reasons, we will in the following consider a function $\fb$ which asymptotically coincides with $j_\mu$ on $\supp w_\mu$ but is defined in such a way that $\fb(z)=1$ for $|z|$ sufficiently large. 
This is achieved by constructing a potential $U_\bo$ in such a way that the scattering length of $w_\mu-U_\bo$ equals zero; $\fb$ is then defined as the scattering solution of $w_\mu-U_\bo$. The advantage of using $\fb$ instead of $j_\mu$ is that $\nabla\fb$ and $1-\fb$ have compact support, which is not true for $j_\mu$. 

\begin{definition} \label{def:U}
Let $\bo\in(\tfrac13,1)$. Define
$$U_\bo(\z):=\begin{cases}
	\mu^{1-3\bo}a 	& \text{ for } \mu^\bo<|\z|<R_\bo,\\
	0				& \text{ else,}\end{cases}$$
where $R_\bo$ is the minimal value in $(\mu^\bo,\infty]$ such that the scattering length of $w_\mu-U_\bo$ equals zero. 
\end{definition}
In Section~\ref{subsec:mictrostructure}, we show by explicit construction that a suitable $R_\bo$ exists and that it is of order $\mu^\bo$. 
We will abbreviate
$$U_\bo^{(ij)}:=U_\bo(\z_i-\z_j)\quad \text{ and }\quad w_\mu^{(ij)}:=w_\mu(z_i-z_j).$$

\begin{definition}\label{def:scat}
Let $\fb\in\mathcal{C}^1(\R^3)$ be the solution of
\begin{equation}\label{eqn:scat:f}
\begin{cases}
\Big(-\Delta+\frac{1}{2}\left(w_\mu(\z)-U_{\bo}(\z)\right)\Big)\fb(\z)=0 \;& \text{for }|\z|< R_{\bo},\vphantom{\Big)}\\
\;\fb(\z)=1 & \text{for }|\z|\geq R_{\bo}.
\end{cases}
\end{equation}
Further, define
$$\gb:=1-\fb.$$

\end{definition}
We will in the sequel abbreviate
$$
\gbij:=\gb(\z_i-\z_j) \quad \text{ and } \quad \fbij:=\fb(\z_i-\z_j).
$$
Definitions \ref{def:U} and \ref{def:scat} imply in particular that 
\begin{equation}\label{eqn:scat(w-U)=0}
\int\limits_{\R^3}\left(w_\mu(\z)-U_\bo(\z)\right)\fb(\z)\d\z=0.
\end{equation}

We now repeat the above heuristic estimate for the renormalised energy per particle with the trial function\footnote{Note that this trial function is not normalised. However, a reasoning similar to Lemma~\ref{lem:g} leads to the estimate  $0\leq 1-\norm{\psi_\mathrm{cor}}^2\ls N\mu^{2\bo}$. As $\bo>\frac13$, the normalisation error is thus irrelevant for our heuristic argument.} $\psi_\mathrm{cor}(z_1\mydots z_N):=\prod_{k=1}^N\phe(z_k)\prod_{1\leq l<m\leq N}\fb(z_l-z_m)$, where the product state is overlaid with a microscopic structure characterised by $\fb$. For $\Vp=0$, this yields 
\begin{align*}
\tfrac{1}{N}&\llr{\psi_\mathrm{cor},H\psi_\mathrm{cor}}-\tfrac{E_0}{\varepsilon^2}\\
=&\;\llr{\prod_{k\geq1}\phe(z_k)\prod_{l<m}\fblm,(-\partial_{x_1}^2\phe(z_1))\prod_{k'\geq2}\phe(z_k')\prod_{l'<m'}\fblms}\\
&+(N-1)\llr{\prod_{k\geq1}\phe(z_k)\prod_{l<m}\fblm,\left(-\Delta_1\fbot+\tfrac12 w_\mu^{(12)}\fbot\right)\prod_{k'\geq1}\phe(z_k')\prod_{\substack{l'<m'\\(l',m')\neq (1,2)}}\fblms}\\
&+2(N-1)\llr{\prod_{k\geq1}\phe(z_k)\prod_{l<m}\fblm,\left(\nabla_1\phe(z_1)\cdot\nabla_1\fbot\right)\prod_{k'\geq2}\phe(z_k')\prod_{\substack{l'<m'\\(l',m')\neq (1,2)}}\fblms}\\
&+(N-1)(N-2)\llr{\prod_{k\geq1}\phe(z_k)\prod_{l<m}\fblm,\left(\nabla_1\fbot\cdot\nabla_1\fboth\right)\prod_{k'\geq1}\phe(z_k')\prod_{\substack{l'<m'\\(l',m')\notin\\\{(1,2),(1,3)\}}}\fblms}.
\end{align*}
Very roughly speaking, we may substitute $\fb\approx1$ unless we integrate against $w_\mu$, which is peaked on the set where $\fb\not=1$, or apply the Laplacian to $\fb$.
For the last line, also note that $\supp\nabla\fb\subseteq B_{R_\bo}(0)$ with $R_\bo=\mathcal{O}(\mu^\bo)$ (Lemma~\ref{lem:scat}), which is for $\bo>\tfrac13$ negligible compared to the mean inter-particle distance $\mu^\frac13$. Thus, the measure of the set 
$\supp\nabla_1\fb(\cdot-z_2)\cap\supp\nabla_1\fb(\cdot-z_3)$ vanishes sufficiently fast in the limit 
 $(N,\varepsilon)\rightarrow(\infty,0)$. 
For the second line, note that \eqref{eqn:scat:f} implies $-\Delta_1\fbot+\tfrac12 w_\mu^{(12)}\fbot=\tfrac12 U_\bo^{(12)}\fbot$. Besides, $1\geq \fb\geq 1- a\mu^{1-\bo}$ on the support of $U^\bo$ and $\fb\approx j_\mu$ on the support of $w_\mu$ (Lemma~\ref{lem:scat}). Hence $\norm{\fb U_\bo\fb}_{L^1(\R^3)}\approx\norm{U_\bo\fb}_{L^1(\R^3)}\approx\int_{\R^3} w_\mu(z)j_\mu(z)\d z=8\pi\mu a$ according to \eqref{eqn:scat(w-U)=0} and \eqref{eqn:a^N,eps}. Thus,  the second line gives to leading order
$$ \tfrac{N-1}{2}\int\d z_1|\phe(z_1)|^2\int\d z|\phe(z_1-z)|^2 U_\bo(z)\fb(z)
\rightarrow 4\pi a\int\d x_1|\Phi(x_1)|^4 \int\d y_1\chi(y_1)|^4\,,$$
and  the renormalised energy per particle is consequently given by the correct expression
$$\tfrac{1}{N}\llr{\psi_\mathrm{cor},H\psi_\mathrm{cor}}-\tfrac{E_0}{\varepsilon^2}
\rightarrow\lr{\Phi(x_1),\left(-\partial_{x_1}^2+\tfrac12\left(8\pi a\int|\chi(y)|^4\d y \right)|\Phi(x_1)|^2\right)\Phi(x_1)}\,.$$

This heuristic argument indicates that the state of the system is asymptotically close to $\psi_\mathrm{cor}$. We will therefore modify the counting functional such that 
$p_1^\varphi p_2^\varphi\cdots p_N^\varphi=|\psi_\mathrm{prod}\rangle\langle\psi_\mathrm{prod}|$
is replaced by $|\psi_\mathrm{cor}\rangle\langle\psi_\mathrm{cor}|,$
i.e.~$P_0$ is replaced by the projection onto the product state overlaid with a microscopic structure minimising the energy.
We substitute in the first term of $\alpha_\xi^<(t)$ 
\begin{equation}\label{alpha:replaced}
\llr{\psi,\hat{m}^{\phe}\psi}
\mapsto\llr{\psi,\prod\limits_{k<l}\fblk\,\hat{m}^{\phe}\prod\limits_{r<s}\fbrs\psi}
\approx\llr{\psi,\hat{m}^{\phe}\psi}-N(N-1)\Re\llr{\psi,\gbot\hat{m}^{\phe}\psi}, 
\end{equation}
where we have used the symmetry of $\psi^{N,\varepsilon}(t)\equiv\psi$ and expanded the products by writing $\fb=1-\gb$ and keeping only the terms which are at most linear in $\gb$. 

This correction in the functional effectively leads to the replacement of $w_\mu$ by $U_\bo\fb$ in the time derivative of the new functional. 
The underlying physical idea is that the low energy scattering is essentially described by the $s$-wave scattering length, hence the scattering at $w_\mu$ is to leading order equivalent to the scattering at $U_\bo\fb$. The terms containing $U_\bo\fb$ can be controlled by the result from \cite{NLS}; the remainders from this substitution must be estimated additionally.
To understand how the substitution works, let us for simplicity consider the case $N=2$ with $\Vp=0$. The full argument is given in Section~\ref{subsec:prop:dt_alpha:GP}. Abbreviating $Z^{(12)}:=w_\mu^{(12)}-b(|\Phi(x_1)|^2+|\Phi(x_2)|^2)$, we obtain
\begin{eqnarray*}
\tfrac{\d}{\d t}\llr{\psi,\hat{m}^{\phe}\psi}
&=&\i\llr{\psi,[Z^{(12)},\hat{m}^{\phe}]\psi}=-2\Im\llr{\psi, Z^{(12)}\hat{m}^{\phe}\psi},\\
-2\tfrac{\d}{\d t}\Re\llr{\psi,\gbot\hat{m}^{\phe}\psi}
&=&2\Im\Big\llangle\psi,\Big(\gbot[Z^{(12)},\hat{m}^{\phe}]+(w_\mu^{(12)}-U_\bo^{(12)})\fbot\hat{m}^{\phe}\\
&&\hspace{6cm}+4\nabla_1\fbot\cdot\nabla_1\hat{m}^{\phe}\Big)\psi\Big\rrangle.
\end{eqnarray*}
Adding these expressions and using that $\gb=1-\fb$, we observe that the term $\llr{\psi,Z^{(12)}\hat{m}^{\phe}\psi}$ cancels. It remains, among other contributions, 
$$-2\Im\llr{\psi,\left(U_\bo^{(12)}\fbot-b_\bo(|\Phi(x_1)|^2+|\Phi(x_2)|^2)\right)\hat{m}^{\phe}\psi},$$
where $w_\mu$ is replaced by $U_\bo\fb$.

\begin{remark}\label{rem:notation}
To simplify the notation, we will in the following drop the index $\varphi$ in all projections and (weighted) many-body operators from Definitions~\ref{def:p} and~\ref{def:hat}. 
From now on, $p=\pp\pc$ always projects onto $\phe(t)=\Phi(t)\chie$, where $\Phi(t)$ is the solution of the Gross--Pitaevskii equation \eqref{NLS} with initial datum $\Phi_0$ from A4, and $\chie$ is the ground state of $-\Delta_y+\tfrac{1}{\varepsilon^2}V^\perp(\tfrac{y}{\varepsilon})$ from A2.
\end{remark}

In our proof, we will use a slightly modified variant of the correction term in \eqref{alpha:replaced}. The reason for the modification is that Lemma~\ref{lem:equivalence} establishes the equivalence of \eqref{T1} and \eqref{T2} with $\alpha_\xi^<(t)\rightarrow0$, hence we must ensure that the correction term converges to zero in the limit $(N,\varepsilon)\rightarrow(\infty,0)$. To make the correction term in \eqref{alpha:replaced} controllable, we replace $\hat{m}$ by the weighted many-body operator $\hat{r}$, which is defined as follows:

\begin{definition}
Define the weight functions
\begin{equation*}\begin{array}{ll}
m^a(k)\,:=\,m(k)-m(k+1), &\quad  m^b(k)\,:=\,m(k)-m(k+2),\\
m^c(k)\,:=\,m^a(k)-m^a(k+1), \qquad &\quad m^d(k)\,:=\,m^a(k)-m^a(k+2),\\
m^e(k)\,:=\,m^b(k)-m^b(k+1), \qquad &\quad m^f(k)\,:=\,m^b(k)-m^b(k+2).\\
\end{array}
\end{equation*}
The corresponding weighted many-body operators are denoted by $\hat{m}^\sharp$, $\sharp\in\{a,b,c,d,e,f\}$. Further, define
$$\hat{r}:=\hat{m}^b p_1 p_2+\hat{m}^a(p_1 q_2+q_1p_2).$$
\end{definition}

Note that the weight functions $m^\sharp$ correspond to discrete derivatives of $m$, which appear in the computations when taking commutators with two-body operators such as $[Z^{(12)},\hat{m}]$.

When replacing $\hat{m}$ by $\hat{r}$ in \eqref{alpha:replaced}, we gain an additional projection $p_1$, which allows us to estimate $ {\gbot p_1}$ instead of $ {\gbot }$ (Lemma~\ref{lem:g:2}).
This change does not affect the replacement of $w_\mu$ by $U_\bo$ because $[Z^{(12)},\hat{m}]=[Z^{(12)},\hat{r}\,]$ by Lemma~\ref{lem:commutators:5}. 
The modified functional is now defined as follows:

\begin{definition}
$$\alpha_\xi(t):=\alpha_\xi^<(t)-N(N-1)\Re\llr{\psi^{N,\varepsilon}(t),\gbot\,\hat{r}\,\psi^{N,\varepsilon}(t)}.$$
\end{definition}

In Proposition~\ref{prop:dt_alpha:GP}, the time derivative of the new functional $\alpha_\xi(t)$ is explicitly calculated, following essentially the steps sketched for $N=2$.  

\begin{prop}\label{prop:dt_alpha:GP}

Under assumptions A1 -- A4, 
$$\left|\tfrac{\d}{\d t}\alpha_\xi(t)\right|\leq \big|\gamma^<(t)\big|+\big|\gamma_a(t)\big|+|\gamma_b(t)|+|\gamma_c(t)|+|\gamma_d(t)|+|\gamma_e(t)|+|\gamma_f(t)|$$
for almost every $t\in\R$,
where
\begin{eqnarray}
\gamma^<(t)&:=&\left|\llr{\psi^{N,\varepsilon}(t),\dot{\Vp}(t,\z_1)\psi^{N,\varepsilon}(t)}-\lr{\Phi(t),\dot{\Vp}(t,(x,0))\Phi(t)}_{L^2(\R)}\right| \label{gamma:GP:a<:1}\\
&&-2N\Im\llr{\psi^{N,\varepsilon}(t),q_1\hat{m}^a_{-1}\big(\Vp(t,\z_1)-\Vp(t,(x_1,0))\big)p_1\psi^{N,\varepsilon}(t)} \label{gamma:GP:a<:2}\\
&&-N(N-1)\Im\Big\llangle\psi^{N,\varepsilon}(t),\tilde{Z}^{(12)}\hat{m}\psi^{N,\varepsilon}(t)\Big\rrangle,\label{eqn:gamma:GP:b:1:3}\\\nonumber\\
\gamma_a(t)&:=&N^2(N-1)\Im\llr{\psi^{N,\varepsilon}(t),\gbot\left[\Vp(t,\z_1)-\Vp(t,(x_1,0)),\hat{r}\right]\psi^{N,\varepsilon}(t)},\label{gamma:GP:a}\\\nonumber\\
\gamma_b(t)&:=&-N\Im\llr{\psi,b(|\Phi(x_1)|^2+|\Phi(x_2)|^2)\gbot\,\hat{r}\,\psi}\label{gamma:GP:b:1:1}\\
&&-N\Im\Big\llangle\psi^{N,\varepsilon}(t),(b_\bo-b)(|\Phi(x_1)|^2+|\Phi(x_2)|^2)\,\hat{r}\,\psi^{N,\varepsilon}(t)\Big\rrangle\label{gamma:GP:b:1:2}\\
&&-N(N-1)\Im\llr{\psi^{N,\varepsilon}(t),\gbot\,\hat{r}\,Z^{(12)}\psi^{N,\varepsilon}(t)},\label{gamma:GP:b:2}\\
\nonumber\\
\gamma_c(t)&:=&-4N(N-1)\Im\llr{\psi^{N,\varepsilon}(t),(\na_1\gbot)\cdot\na_1\hat{r}\,\psi^{N,\varepsilon}(t)}\label{gamma:GP:c},\\\nonumber\\
\gamma_d(t)&:=&-N(N-1)(N-2)\Im\llr{\psi^{N,\varepsilon}(t),\gbot\left[b|\Phi(x_3)|^2,\hat{r}\,\right]\psi^{N,\varepsilon}(t)}\label{gamma:GP:d:1}\\
&&+2N(N-1)(N-2)\Im\llr{\psi^{N,\varepsilon}(t),\gbot\big[w_\mu^{(13)},\hat{r}\,\big]\psi^{N,\varepsilon}(t)}\label{gamma:GP:d:2},\\\nonumber\\
\gamma_e(t)&:=&\tfrac12 N(N-1)(N-2)(N-3)\Im\llr{\psi^{N,\varepsilon}(t),\gbot\big[w_\mu^{(34)},\hat{r}\,\big]\psi^{N,\varepsilon}(t)}\label{gamma:GP:e},\\\nonumber\\
\gamma_f(t)&:=&-2N(N-2)\Im\llr{\psi^{N,\varepsilon}(t),\gbot\left[b|\Phi(x_1)|^2,\hat{r}\,\right]\psi^{N,\varepsilon}(t)}\label{gamma:GP:f}.
\end{eqnarray}
Here, we have used the abbreviations
\begin{eqnarray*}
Z^{(ij)}&:=&w_\mu^{(ij)}-\tfrac{b}{N-1}\left(|\Phi(x_i)|^2+|\Phi(x_j)|^2\right),\\
\tilde{Z}^{(ij)}&:=&U_\bo^{(ij)}\fbij-\tfrac{b_\bo}{N-1}(|\Phi(x_i)|^2+|\Phi(x_j)|^2),
\end{eqnarray*}
where $$
b_\bo\,:=\lim\limits_{(N,\varepsilon)\rightarrow(\infty,0)}\mu^{-1}\int\limits_{\R^3}U_\bo(z)\fb(z)\d z\int\limits_{\R^2}|\chi(y)|^4\d y.$$
\end{prop}

The first expression $\gamma^<$ equals $|\tfrac{\d}{\d t}\alpha_\xi^<(t)|$ with $w_\mu$ replaced by the interaction $U_\bo\fb$. 
The terms $\gamma_a$ to $\gamma_f$ collect all remainders resulting from this replacement. Whereas $\gamma_a$ arises from the strong confinement, 
$\gamma_b$ to $\gamma_f$ are comparable to the corresponding terms from the problem without strong confinement in \cite{pickl2015}. 

\begin{prop}\label{prop:gamma:GP}
Let $\mu$ be sufficiently small and let assumptions A1 -- A4 be satisfied. Then there exist $\frac56<d<\bo<\frac{2}{2+\delta}$ and $0<\xi<\min\{1-\bo,\frac{\bo}{6}\}$ such that for any $t\in\R$
\begin{eqnarray*}
\big|\gamma^<(t)\big|&\ls &\mathfrak{e}(t)\exp\left\{\mathfrak{e}^2(t)+\int_0^t\mathfrak{e}^2(s)\d s\right\}\Big(\alpha_\xi^<(t)+(N\varepsilon^\delta)^{1-\bo}+N^{-1+\bo+\xi}+\mu^{d-\frac13-\frac{\bo}{2}}\Big), \\
\big|\gamma_a(t)\big|&\ls& \mathfrak{e}^3(t)\,\varepsilon^2,\\
\big|\gamma_b(t)\big|&\ls& \mathfrak{e}^3(t)\left(\varepsilon^{1+\bo}+N^{-1+\bo+\xi}\right),\\
\big|\gamma_c(t)\big|&\ls& \mathfrak{e}^2(t)\,N^{-1+\bo+\xi}, \\
\big|\gamma_d(t)\big|&\ls& \mathfrak{e}^3(t)\left(\varepsilon^{1+\bo}+(N\varepsilon^\delta)^{1-\bo+\xi}\right),\\
\big|\gamma_e(t)\big|&\ls&\mathfrak{e}^3(t)\,\varepsilon^{1+\bo}, \\
\big|\gamma_f(t)\big|&\ls&\mathfrak{e}^2(t)\,\varepsilon.
\end{eqnarray*}
\end{prop}

To control $\gamma^<$, we first prove that the interaction $U_\bo\fb$ is of the kind considered in \cite{NLS} and subsequently apply \cite[Proposition~3.5]{NLS}. 
This provides a bound of $|\gamma^<(t)|$ in terms of 
$$\llr{\psi^{N,\varepsilon}(t),\hat{m}\psi^{N,\varepsilon}(t)}+\big|E^{\psi^{N,\varepsilon}(t)}_{U_\bo\fb}(t)-\mathcal{E}^{\Phi(t)}_{U_\bo\fb}(t)\big|,$$ 
where $E^{\psi}_{U_\bo\fb}(t)$ and $\mathcal{E}^{\Phi}_{U_\bo\fb}(t)$ denote the quantities corresponding to \eqref{E^psi} and \eqref{E^Phi}, respectively, but where $w_\mu$ is replaced by $U_\bo\fb$ and $b$ by $\lim_{(N,\varepsilon)\to(\infty,0)}\mu^{-1}\norm{U_\bo\fb}_{L^1(\R^3)}\int|\chi(y)|^4\d y$.
The potential $U_\bo$ is chosen in such a way that $\lim_{(N,\varepsilon)\to(\infty,0)}\norm{U_\bo\fb}_{L^1(\R^3)}=8\pi a$, hence $\mathcal{E}^{\Phi}(t)=\mathcal{E}^{\Phi}_{U_\bo\fb}(t)$ but 
\begin{equation}\label{eqn:heuristics:2}
\Big|E^{\psi^{N,\varepsilon}_0}_{U_\bo\fb}(0)-E^{\psi^{N,\varepsilon}_0}(0)\Big|\sim\mathcal{O}(1).
\end{equation}
To explain why one expects the energy difference \eqref{eqn:heuristics:2} to be of order one,
let us again consider the trial funciton $\psi_\mathrm{cor}$. Following the same heuristic reasoning as before (i.e.~$\fb\approx 1$ unless we integrate against $w_\mu$ or apply the Laplacian, $\fb\approx j_\mu$ on $\supp w_\mu$, and $\fb\approx 1$ on $\supp U_\bo$), this difference is to leading order given by
\begin{eqnarray*}
\tfrac{N-1}{2}\left|\llr{\psi_\mathrm{cor}, \left(w_\mu^{(12)}-(U_\bo\fb)^{(12)}\right)\psi_\mathrm{cor}}\right|
&\sim &N\left|\int\d z_1|\phe(z_1)|^4\int\d z\,\fb(z)^2(w_\mu(z)-U_\bo(z))\right|\\
&\overset{\text{\eqref{eqn:scat(w-U)=0}}}{\sim} &\mu^{-1}\int\d z\,\gb(z)w_\mu(z)\fb(z)\\
&\geq& \mu^{-1}\gb(\mu)\int\d z\, w_\mu(z)\fb(z)\,\overset{\text{\eqref{eqn:integral:scat}}}{\sim} \,8\pi a^2
\;\sim\;\mathcal{O}(1).
\end{eqnarray*}
In the first line, we have substituted $z_2\mapsto z:=z_2-z_1$, approximated $\phe(z_1+z)\approx\phe(z_1)$ for $z\in\supp(w_\mu-U_\bo)$ and used the estimate $\norm{\phe}_{L^\infty(\R^3)}^2\ls\varepsilon^{-2}$ (Lemma~\ref{lem:Phi}).
Further, we have decomposed $\fb=1-\gb$ and used that $\gb$ is decreasing in $|z|$, $\gb(\mu)\sim a$ and $\gb\approx0$ on $\supp U_\bo$. 
Note that by \eqref{eqn:scat:f}, this difference between the potential energies equals exactly the part of the kinetic energy $\llr{\psi_\mathrm{cor},(-\Delta_1)\psi_\mathrm{cor}}$ that is due to the correlations.

As a consequence of \eqref{eqn:heuristics:2}, \cite[Proposition~3.5]{NLS} does not immediately provide a bound of $|\gamma^<(t)|$ in terms of $\alpha_\xi^<(t)$. However, the energy difference enters merely in the single term in this proposition%
\footnote{It enters in (24) in \cite{NLS}, which is a part of $\gamma_b^{(3)}$ in Proposition~3.4. The estimate is given in \cite[Section~4.4.4]{NLS}.}
whose control requires a bound of the kinetic energy $\norm{\partial_{x_1}\qp_1\psi^{N,\varepsilon}(t)}$. 
For interactions $\wb$ scaling with $\beta\in(0,1)$, one shows that (neglecting some terms that vanish in the limit)
\begin{eqnarray}
|E^\psi_{\wb}(t)-\mathcal{E}_{\wb}^\Phi(t)|
&\gs & \llr{\psi,(-\Delta_1+\tfrac{1}{\varepsilon^{2}}(V^\perp(\tfrac{y_1}{\varepsilon})-\tfrac{E_0}{\varepsilon^2})\psi}-\norm{\Phi'}^2\nonumber\\
&\gs&\norm{\partial_{x_1}\qp_1\psi}^2+(\norm{\partial_{x_1}\pp_1\psi}^2-\norm{\Phi'}^2)\nonumber\\
&\geq&\norm{\partial_{x_1}\qp_1\psi}^2-\norm{\Phi'}^2\llr{\psi,\hat{n}\psi}.\label{eqn:heuristics:1}
\end{eqnarray}
Hence, essentially $\norm{\partial_{x_1}\qp_1\psi^{N,\varepsilon}(t)}^2\ls\alpha_\xi^<(t)$ \cite[Lemma~4.17]{NLS}, which is why the energy difference enters the estimate of $|\gamma^<(t)|$.

Turning back to the Gross--Pitaevskii regime, let us apply \eqref{eqn:heuristics:1} to the interaction $U_\bo\fb$. Making use of the fact that $\mathcal{E}^{\Phi}(t)=\mathcal{E}^{\Phi}_{U_\bo\fb}(t)$, we obtain
$$|E^\psi_{U_\bo\fb}(t)-E^\psi(t)|+|E^\psi(t)-\mathcal{E}^\Phi(t)|
\;\geq\;|E^\psi_{U_\bo\fb}(t)-\mathcal{E}^\Phi(t)|
\;\gs\; \norm{\partial_{x_1}\qp_1\psi}^2-\norm{\Phi'}^2\llr{\psi,\hat{n}\psi}.
$$
Since $|E^\psi_{U_\bo\fb}(t)-E^\psi(t)|\sim\mathcal{O}(1)$
already at time zero by \eqref{eqn:heuristics:2}
and $|E^\psi(t)-\mathcal{E}^\Phi(t)|\leq \alpha_\xi^<(t)$, we expect
$$\norm{\partial_{x_1}\qp_1\psi}^2\; \ls\; \alpha_\xi^<(t)+\mathcal{O}(1)$$ 
for the Gross--Pitaevskii scaling of the interaction.
The additional $\mathcal{O}(1)$-contribution arises because one of the terms\footnote{This is the term $\llr{\psi,((N-1)w_\mu^{(12)}-b|\Phi(x_1)|^2)\psi}$. In the proof of Lemma~\ref{lem:E_kin:GP}, we cope with this term essentially by adding and subtracting the potential $U_\bo$. The term containing the difference $w_\mu-U_\bo$ together with the part of the kinetic energy close around the scattering centers is non-negative (Lemma~\ref{lem:scat:peter}). The terms containing $U_\bo$ can be shown to vanish in the limit as in \eqref{eqn:heuristics:1}.} we have neglegted in \eqref{eqn:heuristics:1} is not small for $\beta=1$. 

The part of the kinetic energy orthogonal to the condensate 
$\norm{\partial_{x_1}\qp_1\psi}$ is not small since the microscopic structure does not vanish in the limit but carries a kinetic energy of order $\mathcal{O}(1)$.
This energy is the reason for the factor $8\pi a$ in the effective equation, which is $\mathcal{O}(1)$ different from the factor $\norm{w}_{L^1(\R^3)}$ for scalings $\beta\in(0,1)$ with negligible microscopic structure.

To estimate the one problematic term in $\gamma^<(t)$, one notes that the predominant part of the kinetic energy is localised around the scattering centers, where the microscopic structure is non-trivial. Therefore, we define the set $\Ao$ (Definition~\ref{def:cutoffs}) as $\R^{3N}$ where sufficiently large balls around the scattering centers are cut out,
and show that $\norm{\charAo\partial_{x_1}\qp_1\psi^{N,\varepsilon}(t)}^2\ls |E^\psi(t)-\mathcal{E}^\Phi(t)|+\llr{\psi^{N,\varepsilon}(t),\hat{n}\psi^{N,\varepsilon}(t)}$ plus some terms vanishing in the limit (Lemma~\ref{lem:E_kin:GP}). 
Subsequently, we adapt the estimate from \cite[Proposition~3.5]{NLS} to this new energy lemma, making use of the fact that the complement of $\Ao$ is very small.

The remainder of the proof consists of estimating the terms $\gamma_a$ to $\gamma_f$ arising from the effective replacement of $w_\mu$ by $U_\bo\fb$. 
The key tool for this is our knowledge of the microscopic structure (Lemma~\ref{lem:scat} and Lemma~\ref{lem:g}).

\begin{remark}\label{rem:differences}
In principle, we adjust the method from \cite{pickl2015} to the situation with strong confinement and to the associated more singular scaling of the interaction. We give a new proof for Lemma~\ref{lem:scat:1}-c (concerning the microscopic structure) by exploiting the spherical symmetry of the scattering problem to reduce it to an ODE and explicitly construct its solution. 

The proof of Lemma~\ref{lem:E_kin:GP} (providing an estimate for the kinetic energy) becomes more involved due to the confinement, since one must show that the positive expression $\norm{\nabla_{y_1}\psi^{N,\varepsilon}(t)}^2$ compensates not only for a sufficient share of the negative part of $\llangle{\psi^{N,\varepsilon}(t),(w_\mu-U_\bo)\psi^{N,\varepsilon}(t)}\rrangle$ as in \cite{pickl2015} but also for the large negative part of $\tfrac{1}{\varepsilon^2}\llr{\psi^{N,\varepsilon}(t),(V^\perp(\tfrac{y_1}{\varepsilon})-E_0)\psi^{N,\varepsilon}(t)}$.

For the control of $\gamma^d$, we follow \cite{jeblick2016}. 
The estimate of $\gamma_c$ is different from the problem without confinement because each $\na$ contributes a factor $\varepsilon^{-1}$. To handle this, we prove a new Lemma~\ref{lem:nabla:g} which provides estimates for $\na\gb$, and combine this with the new estimate in Lemma~\ref{lem:g:5}.
\end{remark}

The last proposition ensures that the correction term converges to zero as $(N,\varepsilon)\rightarrow(\infty,0)$, which is required for the Grönwall argument. 

\begin{prop}\label{prop:correction}
Under assumptions A1 -- A4, the correction term in $\alpha_\xi(t)$ is for all $t\in\R$ bounded as
$$\left|N(N-1)\Re\llr{\psi^{N,\varepsilon}(t),\gbot\,\hat{r}\,\psi^{N,\varepsilon}(t)}\right|\;\ls\; \varepsilon^{1+\bo}N^{\xi-\frac{\bo}{2}}.$$
\end{prop}

\noindent\emph{Proof of Theorem~\ref{thm}.}
From Propositions \ref{prop:dt_alpha:GP} and \ref{prop:gamma:GP}, we gather that for sufficiently small $\mu$, there exist suitable $\bo$, $\xi$ and $d$ such that 
$$\left|\tfrac{\d}{\d t}\alpha_\xi(t)\right|\;\ls\; 
C(t)\left(\alpha_\xi^<(t)+(N\varepsilon^\delta)^{1-\bo+\xi}+N^{-1+\bo+\xi}+\mu^{d-\frac13-\frac{\bo}{2}}\right)$$
for almost every $t\in\R$. We have simplified the expression by noting that 
$\varepsilon^{1+\bo}<\varepsilon<(N\varepsilon^\delta)^{1+\xi-\bo}$ because $\delta(1+\xi-\bo)<\delta(1+\xi)<1$ as $\delta<\frac25$ and $\xi<1-\bo<\frac16$. 
Besides, we have used the abbreviation
\begin{equation}\label{C}
C(t)\;:=\;\mathfrak{e}(t)\exp\left\{\mathfrak{e}^2(t)+\int_0^t\mathfrak{e}^2(s)\d s\right\}.
\end{equation} 
Recall that $\mathfrak{e}(t)$ is for each $t\in\R$ bounded uniformly in $N$ and $\varepsilon$ by assumption A4. 
Let us introduce the abbreviations
\begin{eqnarray*}
R(t)&:=&-N(N-1)\Re\llr{\psi^{N,\varepsilon}(t),\gbot\,\hat{r}\,\psi^{N,\varepsilon}(t)},\\
B&:=&(N\varepsilon^\delta)^{1-\bo+\xi}+N^{-1+\bo+\xi}+\mu^{d-\frac13-\frac{\bo}{2}}.
\end{eqnarray*}
By Proposition~\ref{prop:correction}, $|R(t)|< B$ uniformly in $t$. $\alpha_\xi(t)+B$ is thus non-negative and 
\begin{eqnarray*}
\alpha_\xi^<(t)&=&\alpha_\xi(t)-R(t)\;\leq\;\alpha_\xi(t)+|R(t)|\;\ls\;\alpha_\xi(t)+B,\\
\alpha_\xi(t)+B&=&\alpha_\xi^<(t)+R(t)+B\;\ls\;\alpha_\xi^<(t)+B,
\end{eqnarray*}
hence
$$\tfrac{\d}{\d t}(\alpha_\xi(t)+B)\; \ls\;  C(t)\left(\alpha_\xi(t)+B\right)$$
for almost every $t\in\R$.
By the differential form of Grönwall's inequality,
$$0\; \leq\; \alpha_\xi^<(t)\ls \alpha_\xi(t)+B \ls\left(\alpha_\xi^<(0)+B\right)\exp\left\{2\int_0^tC(s)\d s\right\}$$
for all $t\in\R$.
The sequence $(N,\varepsilon)$ is admissible and $\xi<1-\bo$, hence $\lim_{(N,\varepsilon)\rightarrow(\infty,0)}B=0$ and \eqref{A4:1} and \eqref{A4:2} imply by Lemma~\ref{lem:equivalence} that
$$0\; \leq\; \lim\limits_{(N,\varepsilon)\rightarrow(\infty,0)}\left(\alpha_\xi(0)+B\right)
\; \ls\; \lim\limits_{(N,\varepsilon)\rightarrow(\infty,0)}\left(\alpha^<_\xi(0)+B\right)\,\overset{\text{\ref{lem:equivalence}}}{=}\,0,$$
 which by Lemma~\ref{lem:equivalence} concludes the proof.\qed

\begin{cor}\label{cor:rates}
Let $t\in\R$. Then for any $\rho\in(0,\frac{1}{12})$,
$$
\Tr\left|\gamma^{(1)}_{\psi^{N,\varepsilon}(t)}-|\phe(t)\rangle\langle\phe(t)|\right|\; \ls\; 
\left(A(0)+N^{-\frac{1}{12}+\rho}+\left(N\varepsilon^\delta\right)^{\frac{3}{12}-3\rho} \right)^\frac12
\exp\left\{\int_0^tC(s)\d s\right\},
$$
with $C(t)$ as in \eqref{C} and where $$A(0)\; :=\; \left|E^{\psi^{N,\varepsilon}_0}(0)-\mathcal{E}^{\Phi_0}(0)\right|+\sqrt{\Tr\Big|\gamma^{(1)}_{\psi_0^{N,\varepsilon}}-|\phe_0\rangle\langle\phe_0|\Big|}.$$
\end{cor}
\begin{proof}
Follows from Lemma~\ref{lem:equivalence} after optimisation over $\xi$, $\bo$ and $d$.
\end{proof}

\begin{remark}
For $\Vp=0$, one obtains $\norm{\Phi(t)}_{H^2(\R)}\ls C(\norm{\Phi_0}_{H^2(\R)})$ uniformly in $t$, where $C(\norm{\Phi_0}_{H^2(\R)})$ denotes some expression depending only on $\norm{\Phi_0}_{H^2(\R)}$ \cite[Exercise 3.36]{tao}\footnote{To prove this, one observes that the quantity 
$E_2(\Phi):=\int_\R\left(|\partial_x^2\Phi|^2+c_1|\partial_x\Phi|^2|\Phi|^2+c_2\Re((\overline{\Phi}\partial_x\Phi)^2)+c_3|\Phi|^6\right)\d x$ 
is conserved for solutions of \eqref{NLS} with $\Vp=0$, where $c_1$, $c_2$ and $c_3$ denote some absolute constants.}. Defining $\tilde{\mathfrak{e}}:=1+|E^{\psi_0^{N,\varepsilon}}(0)|+|\mathcal{E}^{\Phi_0}(0)|+(C(\norm{\Phi_0}_{H^2(\R)}))^2$ in analogy to \eqref{def:e}, we obtain the rate
$$\Tr\left|\gamma^{(1)}_{\psi^{N,\varepsilon}(t)}-|\phe(t)\rangle\langle\phe(t)|\right|\; \ls\; 
\left(A(0)+N^{-\frac{1}{12}+\rho}+\left(N\varepsilon^\delta\right)^{\frac{3}{12}-3\rho} \right)^\frac12
\exp\left\{\,\tilde{\mathfrak{e}}\,t\right\},
$$
where the growth in time is exponential instead of doubly exponential.
\end{remark}

\section{Proofs of the propositions}
\subsection{Preliminaries}
In this section, we collect some useful lemmata, which are for the most part taken from \cite{NLS} and we refer to this work for the proofs. Lemma~\ref{lem:a_priori:w12} contains additional statements following \cite[Proposition~A.2]{pickl2015}.
We will from now on always assume that assumptions A1 -- A4 are satisfied.

\begin{lem}\label{lem:l}
Let $f:\mathbb{N}_0\rightarrow\mathbb{R}_0^+$, $d\in\mathbb{Z}$, $\rho\in\{a,b\}$ and $\nu\in\{c,d,e,f\}$. Then
\lemit{
 	\item $\onorm{\hat{f}}=\onorm{\hat{f}_d}=\onorm{\hat{f}^\frac{1}{2}}^2=\sup\limits_{0\leq k\leq N}f(k)$,\label{lem:l:1} 
	\item $\onorm{\hat{m}^\rho}\leq N^{-1+\xi}$, \; $\onorm{\hat{m}^\nu}\ls N^{-2+3\xi}$ \, and \, $\onorm{\hat{r}}\ls N^{-1+\xi},$ \label{lem:l:2}
	\item $\norm{\hat{m}^\rho q_1\psi^{N,\varepsilon}(t)}\ls N^{-1},$\label{lem:l:3}
	\item	\label{lem:fqq:2}
			$\norm{\hat{f}q_1q_2\psi^{N,\varepsilon}(t)}^2\ls\norm{\hat{f}\,\hat{n}^2\psi^{N,\varepsilon}(t)}^2.$
}
\end{lem}

\begin{proof} 
Assertions (a), (c) and (d) are proven in \cite{NLS}, Lemma~4.1 and 4.4. For part (b), note that
\begin{equation*}
m'(k)= 
\begin{cases}
      \frac{1}{2\sqrt{kN}}& \text{for}\;\; k\geq N^{1-2\xi}, \vspace{2pt}\\
      \frac12N^{-1+\xi} & \text{else}\\
\end{cases} 
\quad\text{and}\quad
m''(k)=
\begin{cases}
      -\frac{1}{4\sqrt{k^3N}}& \text{for}\;\; k\geq N^{1-2\xi},\vspace{2pt} \\
      0 & \text{else},
\end{cases} 
\end{equation*}
where $'\equiv\tfrac{\d}{\d k}$. Hence $|m'(k)|\leq \tfrac12 N^{-1+\xi}$ and $|m''(k)|\leq\tfrac14 N^{-2+3\xi}$ for any $k\geq0$. By the mean value theorem, this implies e.g.~$|m^a(k)|\ls N^{-1+\xi}$ and $|m^c(k)|\ls N^{-2+3\xi}$. The other expressions work analogously.
\end{proof}

\begin{lem}\label{lem:commutators}
Let $f,g:\mathbb{N}_0\rightarrow\mathbb{R}_0^+$ be any weights and $i,j\in\{1,\dots,N\}$. 
\lemit{
	\item	\label{lem:commutators:1} For $k\in\{0,\dots,N\}$,
			$$\hat{f}\,\hat{g}=\hat{fg}=\hat{g}\hat{f},\qquad\hat{f}p_j=p_j\hat{f},
			\qquad\hat{f}q_j=q_j\hat{f}, \qquad\hat{f}P_k=P_k\hat{f}.$$
	\item	\label{lem:commutators:2}
			Define $Q_0:=p_j$, $Q_1:=q_j$, $\tilde{Q}_0:=p_ip_j$, $\tilde{Q}_1\in\{p_iq_j,q_ip_j\}$ and $\tilde{Q}_2:=q_iq_j$. 
			Let $S_j$ be an operator acting only on factor $j$ in the tensor product and $T_{ij}$ acting only on $i$ and $j$.
			Then for $\mu,\nu\in\{0,1,2\}$	
			$$ Q_\mu\hat{f}S_jQ_\nu=Q_\mu S_j\hat{f}_{\mu-\nu}Q_\nu \quad \text{ and } \quad
			\tilde{Q}_\mu\hat{f}T_{ij}\tilde{Q}_\nu=\tilde{Q}_\mu T_{ij}\hat{f}_{\mu-\nu}\tilde{Q}_\nu.$$
	\item	\label{lem:commutators:5}
			\begin{equation*}
			[T_{ij},\hat{f}]=[T_{ij},p_ip_j(\hat{f}-\hat{f}_2)+(p_iq_j+q_ip_j)(\hat{f}-\hat{f}_1)].
			\end{equation*}
}
\end{lem}

\begin{proof}
\cite{NLS}, Lemma~4.2.
\end{proof}

\begin{lem}\label{lem:derivative_m}
Let $f:\mathbb{N}_0\rightarrow\R^+_0$. Then
\lemit{
	\item	\label{lem:derivative_m:1}
			$P_k,\;\hat{f}\in \mathcal{C}^1\big(\R,\mathcal{L}\left(L^2(\R^{3N})\right)\big)$ for $0\leq k\leq N$,
	\item	\label{lem:derivative_m:2}
			$\left[-\Delta_{y_j}+\tfrac{1}{\varepsilon^2}V^\perp(\tfrac{y_j}{\varepsilon}),\hat{f}\right]=0$ for $1\leq j \leq N$,
	\item	\label{lem:derivative_m:3}
			$\tfrac{\d}{\d t}\hat{f}=i\Big[\hat{f},\sum\limits_{j=1}^N h_j(t)\Big],$\\
			where $h_j(t)$ denotes the one-particle operator corresponding to $h(t)$ from \eqref{NLS} acting on the $j$\textsuperscript{th} factor in $L^2(\R^{3N})$.
}
\end{lem}
\begin{proof}
\cite{NLS}, Lemma~4.3.
\end{proof}

\begin{lem}\label{lem:Gamma:Lambda}
Let $\Gamma,\Lambda\in L^2(\R^{3N})$ be symmetric in the coordinates $\{z_2\mydots z_N\}$, let $r_2$ and $s_2$ denote operators acting only on the second factor of the tensor product, and let $F:\R^3\times\R^3\to\R^d$ for $d\in\mathbb{N}$. Then
$$\left|\llr{\Gamma,r_2F(z_1,z_2)s_2\Lambda}\right|\leq\norm{\Gamma}\left(\norm{s_2F(z_1,z_2)r_2\Lambda}^2+\tfrac{1}{N-1}\norm{r_2F(z_1,z_2)s_2\Lambda}^2\right)^\frac12.$$
\end{lem}
\begin{proof}
\cite{NLS}, Lemma~4.7.
\end{proof}

\begin{lem}\label{lem:Phi}
The nonlinear equation \eqref{NLS} is well-posed and $H^2(\R)$ solutions exist globally, i.e.~for any initial datum $\Phi_0\in H^2(\R)$, it follows that $\Phi(t)\in H^2(\R)$ for any $t\in\R$. Besides, for sufficiently small $\varepsilon$,
\lemit{
	\item 	\label{lem:Phi:1}
			\begin{tabular}{p{4.5cm}p{5cm}}
			$\norm{\Phi(t)}_{L^2(\R)}=1,$ &
			$\norm{\Phi(t)}_{H^1(\R)}\leq\mathfrak{e}(t),$
			\end{tabular}\\
			\begin{tabular}{p{4.5cm}p{10cm}}
			$\norm{\Phi(t)}_{L^\infty(\R)}\leq\mathfrak{e}(t),$ & 
			$\norm{\Phi'}_{L^\infty(\R)}\leq\norm{\Phi(t)}_{H^2(\R)}\ls \exp\left\{\mathfrak{e}^2(t)+\int_0^t\mathfrak{e}^2(s)\d s\right\},$
			\end{tabular}
	\item	\label{cor:varphi:1} 
			\begin{tabular}{p{4.5cm}p{5cm}}
			$\norm{\phe(t)}_{L^\infty(\R^3)}\ls\mathfrak{e}(t)\varepsilon^{-1},$ &
			$\norm{\na\phe(t)}_{L^\infty(\R^3)}\ls\mathfrak{e}(t)\varepsilon^{-2}.$
			\end{tabular}
}
\end{lem}
\begin{proof}
\cite{NLS}, Lemma~4.8.
\end{proof}

\begin{lem}\label{lem:pfp}
Let $t\in\R$ be fixed and let $j,k\in\{1\mydots N\}$. Let $g:\R^3\times\R^3\rightarrow\R$ and $h:\R\times\R\to\R$ 
be measurable functions such that $|g(z_j,z_k)|\leq G(z_k-z_j)$ and $|h(x_j,x_k)|\leq H(x_k-x_j)$ 
almost everywhere for some $G:\R^3\rightarrow\R$, $H:\R\to\R$. 
 Then
\lemit{
	\item	\label{lem:pfp:1} 	
			$\onorm{p_jg(z_j,z_k)p_j}\ls\mathfrak{e}^2(t)\varepsilon^{-2}\norm{G}_{L^1(\R^3)}$\; for $G\in L^1(\R^3)$, 
	\item	\label{lem:pfp:2}
			$\onorm{g(z_j,z_k)p_j}=\onorm{p_jg(z_j,z_k)}\ls \mathfrak{e}(t)\varepsilon^{-1}\norm{G}_{L^2(\R^3)}$\; for $G\in L^2\cap L^\infty(\R^3)$, 
	\item	\label{lem:pfp:3} 
			$\onorm{g(z_j,z_k)\na_jp_j}\ls\mathfrak{e}(t)\varepsilon^{-2}\norm{G}_{L^2(\R^3)}$\; for $G\in L^2(\R^3)$, 
	\item 	\label{lem:pfp:4}	
			$\onorm{ h(x_j,x_k)\pp_j}=\onorm{\pp_j h(x_j,x_k)}\leq\mathfrak{e}(t)\norm{H}_{L^2(\R)}$ for $H\in L^2\cap L^\infty(\R)$.  
}
\end{lem}

\begin{proof}
\cite{NLS}, Lemma~4.9.
\end{proof}

\begin{lem}\label{lem:a_priori:w12}
Let $\varepsilon$ be sufficiently small and $t\in\R$ be fixed. Then
\lemit{
	\item	\label{lem:a_priori:4}
			\hspace{-0.2cm}\begin{tabular}{p{3.6cm}p{3.5cm}p{5cm}}
			$\onorm{\partial_{x_1}\pp_1}\leq\mathfrak{e}(t),$ & $\onorm{\na_{y_1}\pc_1}\ls\varepsilon^{-1}$, 
			& $\onorm{\partial^2_{x_1}p_1}\leq\norm{\Phi(t)}_{H^2(\R)}$,
 			\end{tabular}
 		\item[]
 			\hspace{-0.2cm}\begin{tabular}{p{3.6cm}p{3.5cm}p{4cm}}
 			$\norm{\qc_1\psi^{N,\varepsilon}(t)}\leq \mathfrak{e}(t)\varepsilon$, &
 			$\norm{\partial_{x_1}\qp_1\psi}\ls\mathfrak{e}(t)$,&
			$\norm{\na_{y_1}\qc_1\psi^{N,\varepsilon}(t)}\ls\mathfrak{e}(t),$ 
			\end{tabular}	
 		\item[]
 			\hspace{-0.2cm}\begin{tabular}{p{3.6cm}p{3.5cm}p{4cm}}
 			$\norm{\partial_{x_1}\psi^{N,\varepsilon}(t)}\leq\mathfrak{e}(t)$, &
			$\norm{\nabla_{y_1}\psi^{N,\varepsilon}(t)}\ls\varepsilon^{-1}$, &
			$\norm{\na_1\psi^{N,\varepsilon}(t)}\ls\varepsilon^{-1},$
			\end{tabular}
	\item $\left\lVert{\sqrt{w_\mu^{(12)}}\psi^{N,\varepsilon}(t)}\right\rVert\ls \mathfrak{e}(t)N^{-\frac12},\qquad $\label{lem:w12:1}
	\item $\norm{w_\mu^{(12)}\psi^{N,\varepsilon}(t)}\ls \mathfrak{e}(t)N^\frac12\varepsilon^{-2},$\label{lem:w12:2}
	\item $\onorm{p_1\mathbbm{1}_{\supp{w_\mu}}(\z_1-\z_2)}=\onorm{\mathbbm{1}_{\supp{w_\mu}}(\z_1-\z_2)p_1}
	\ls \mathfrak{e}(t)N^{-\frac32}\varepsilon^2,$\label{lem:w12:3}
	\item $\norm{p_1w_\mu^{(12)}\psi^{N,\varepsilon}(t)}\ls \mathfrak{e}^2(t)N^{-1},$\label{lem:w12:4}
	\item 
	$\norm{\left(\Vp(t,z_1)-\Vp(t,(x_1,0))\right)\psi^{N,\varepsilon}(t)}\ls\mathfrak{e}^3(t)\varepsilon.$\label{lem:taylor}
}
\end{lem}

\begin{proof}
Part (a) is proven in \cite[Lemma~4.10.]{NLS}. $\tfrac{E_0}{\varepsilon^2}$ is the smallest eigenvalue of $-\Delta_y+\tfrac{1}{\varepsilon^2}V^\perp(\tfrac{y}{\varepsilon})$, hence $\llr{\psi^{N,\varepsilon}(t),(-\Delta_{y_1}+\tfrac{1}{\varepsilon^2}V^\perp(\tfrac{y_1}{\varepsilon})-\tfrac{E_0}{\varepsilon^2})\psi^{N,\varepsilon}(t)}\geq 0$. 
This implies (b) as
$$\mathfrak{e}^2(t)\geq|E^{\psi^{N,\varepsilon}(t)}(t)|\geq \tfrac{N-1}{2}\Big\lVert{\sqrt{w_\mu^{(12)}}\psi^{N,\varepsilon}(t)}\Big\rVert^2-\norm{\Vp(t)}_{L^\infty(\R^3)}\gs N\norm{\sqrt{w_\mu^{(12)}}\psi^{N,\varepsilon}(t)}^2-\mathfrak{e}^2(t).$$
For part (c), observe that
$$\norm{w_\mu^{(12)}\psi^{N,\varepsilon}(t)}\leq \norm{w_\mu}_{L^\infty(\R^3)}^\frac12\Big\lVert{\sqrt{w_\mu^{(12)}}\psi^{N,\varepsilon}(t)}\Big\rVert \ls \mu^{-1}\mathfrak{e}(t)N^{-\frac12}.$$
Assertion (d) follows from Lemma~\ref{lem:pfp:2} because $\norm{\mathbbm{1}_{\supp w_\mu}}^2_{L^2(\R^3)}\ls\mu^3$.
Part (e) is a consequence of
$$\norm{p_1w_\mu^{(12)}\psi^{N,\varepsilon}(t)}=\norm{p_1\mathbbm{1}_{\supp{w_\mu}}(\z_1-\z_2)w_\mu^{(12)}\psi^{N,\varepsilon}(t)}\leq \onorm{p_1\mathbbm{1}_{\supp{w_\mu}}(\z_1-\z_2)}\norm{w_\mu^{(12)}\psi^{N,\varepsilon}(t)}.$$
Finally, (f) is proven in \cite[Lemma~4.11]{NLS}.
\end{proof}

\begin{lem}\label{lem:psi-Phi}
Let $\psi\in L_+^2(\R^{3N})$ be normalised and $f\in L^\infty(\R)$. Then
\begin{equation*}
\left|\llr{\psi,f(x_1)\psi}-\lr{\Phi(t),f\Phi(t)}_{L^2(\R)}\right|\ls\norm{f}_{L^\infty(\R)}\llr{\psi,\hat{n}\psi}.
\end{equation*}
\end{lem}
\begin{proof}
\cite{NLS}, Lemma~4.6.
\end{proof}

\subsection{Microscopic structure}\label{subsec:mictrostructure}
In this section, we prove some important properties of the solution $\fb$ of the zero-energy scattering equation \eqref{eqn:scat:f} and of its complement $\gb$. 

\begin{lem}\label{lem:scat} 
Let $\fb$ as in Definition~\ref{def:scat}, $j_\mu$ as in \eqref{eqn:scat} and $R_\bo$ as in Definition~\ref{def:U}. Then
\lemit{
	\item $\fb$ is a non-negative, non-decreasing function of $|\z|$,\label{lem:scat:1}
	\item $\fb(\z)\geq j_\mu(\z)$ for all $\z\in\R^3$ and there exists $\kappa_\bo\in \big(1,\frac{\mu^\bo}{\mu^\bo-\mu a}\big)$ such that for $|\z|\leq\mu^\bo$, $\fb(\z)=\kappa_\bo j_\mu(\z)$,\label{lem:scat:2}
	\item $R_\bo\ls\mu^\bo$.\label{lem:scat:5}
	\item $\norm{\mathbbm{1}_{|z_1-z_2|<R_\bo}\nabla_1\psi}^2+\tfrac12\llr{\psi, (w_\mu^{(12)}-U_\bo^{(12)})\psi}\geq 0$ for any $\psi\in \mathcal{D}(\nabla_1)$.\label{lem:scat:peter}
}
\end{lem}

\begin{proof}
We prove this Lemma~by explicitly constructing a spherically symmetric, continuously differentiable solution $\fb$ of~\eqref{eqn:scat:f}. This solution is unique by \cite[Chapter 2.2, Theorem~16]{evans}.
Consider $\tilde{f}:\R^+_0\rightarrow \R$ with
\begin{equation}\label{eqn:f:tilde}
\tilde{f}(r):=r\fb(r),
\end{equation}
where $r:=|\z|$. $\fb\in \mathcal{C}^1(\R^3)$ solves \eqref{eqn:scat:f} precisely if $\tilde{f}$ solves the corresponding ODE
\begin{equation}\label{eqn:scat:f^tilde}
\begin{cases}
	\tilde{f}''(r)=\frac12\left(w_\mu(r)-U_\bo(r)\right)\tilde{f}(r) 	& \text{for } 0<r<R_\bo,\\
	\tilde{f}(r)=r											& \text{for } r\geq R_\bo,\\
	\tilde{f}(r)=0											& \text{for } r=0,							
\end{cases}
\end{equation}
where $'\equiv \frac{\d}{\d r}$. Analogously, \eqref{eqn:scat} is equivalent to
\begin{equation}\label{eqn:scat:j^tilde}
\begin{cases}
	\tilde{j}''(r)=\frac12w_\mu(r)\tilde{j}(r) 	& \text{for } 0<r<\mu,\\
	\tilde{j}(r)=r-\mu a									& \text{for } r\geq \mu,\\
	\tilde{j}(r)=0										& \text{for } r=0,							
\end{cases}
\end{equation}
where $\tilde{j}:\R^+_0\rightarrow\R^+_0$ is defined as $\tilde{j}(r):=rj_\mu(r)$
and depicted in Figure \ref{fig}.

\begin{figure}[t]
	\begin{center}
	\includegraphics[scale=0.6]{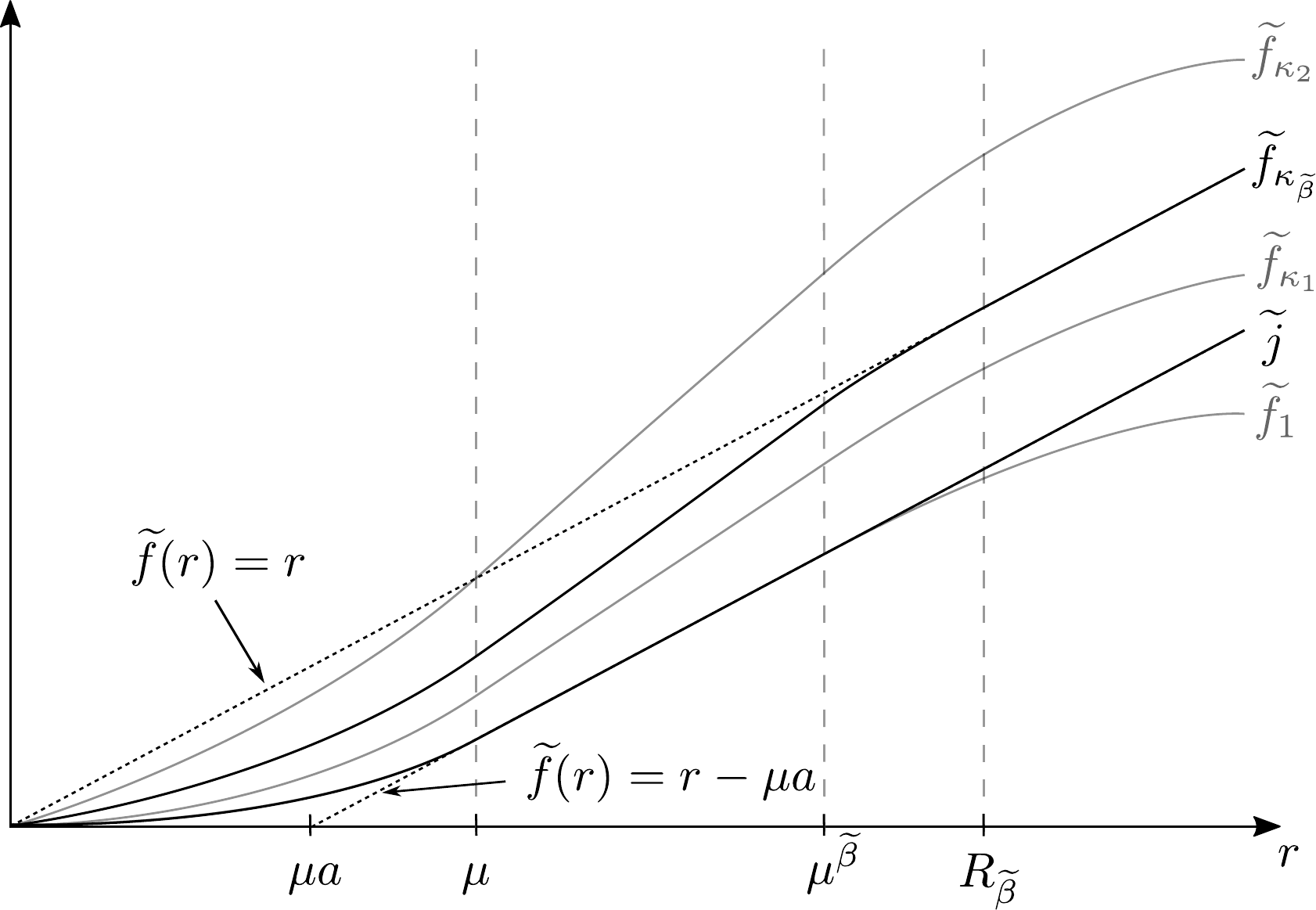}
	\end{center}
	\caption{Construction of the solution $\tilde{f}_{\kappa_\bo}$ of \eqref{eqn:scat:f^tilde}. The lower black curve represents the solution $\tilde{j}$ of \eqref{eqn:scat:j^tilde}, the dashed graphs mark the straight lines $r$ and $r-\mu a$. The functions $\tilde{f}_1$ and $\tilde{f}_{\kappa_2}$ drawn in grey are exemplary members of the one-parameter family $\{\tilde{f}_{\kappa}\}_{\kappa\geq 1}$ with $1<\kappa_1<\kappa_2$. For $0<r<\mu^\bo$, $\tilde{f}_{\kappa}(r)=\kappa \tilde{j}(r)$ is a multiple of $\tilde{j}(r)$. This implies in particular that $\tilde{f}_\kappa(r)$ is a straight line with slope $\kappa$ for $\mu<r<\mu^\bo$. In the region $r>\mu^\bo$, $\tilde{f}_\kappa$ is concave. 
	The solution to \eqref{eqn:scat:f^tilde} must become tangential to the straight line $r$ at some point $r>\mu^\bo$, which will be called $R_\bo$. It is clear that $\tilde{f}_1$ and $\tilde{f}_{\kappa_1}$ will not touch the straight line $r$ (at least not before they decrease and increase again). Contrarily, $\tilde{f}_{\kappa_2}$ already intersects $r$ at $\mu$ and is therefore ruled out as well.
	As the family is strictly increasing in $\kappa$, there must be a curve in between $\tilde{f}_1$ and $\tilde{f}_{\kappa_2}$ that is tangential to $r$ at some point. This is the solution $\tilde{f}_{\kappa_\bo}$ of \eqref{eqn:scat:f^tilde}, drawn in black.
	} 
	\label{fig}
\end{figure}

For $0\leq r\leq \mu^\bo$, $\tilde{f}''(r)=\frac12 w_\mu(r)\tilde{f}(r)$ and $\tilde{f}(0)=0$. Clearly, both conditions are fulfilled by the choice $\tilde{f}_\kappa(r)=\kappa\, \tilde{j}(r)$ for some $\kappa\geq 1$. Consequently, 
\begin{equation}\label{eqn:AB}
\tilde{f}_\kappa(\mu^\bo)=\kappa(\mu^\bo-\mu a)\quad\text{and}\quad\tilde{f}'_\kappa(\mu^\bo)=\kappa.
\end{equation} 
For $\mu^\bo<r<R_\bo$, $\tilde{f}_\kappa$ solves $\tilde{f}''_\kappa(r)=-\frac12 U_\bo(r)\tilde{f}_\kappa(r)$ and is subject to the boundary conditions~\eqref{eqn:AB}. As $U_\bo$ is constant over this region, the solution for $\mu^\bo< r< R_\bo$ is
$$
\tilde{f}_\kappa(r)=\kappa\Big[A\sin(ur)+B\cos(ur)\Big],
$$
where $u:=\sqrt{\tfrac12a\mu^{1-3\bo}}$ and
\begin{eqnarray*}
A&:=&\left((\mu^\bo-\mu a)\sin (\mu^\bo u)+ u^{-1}\cos (\mu^\bo u)\right),\\
B&:=&\left((\mu^\bo-\mu a)\cos (\mu^\bo u)- u^{-1}\sin (\mu^\bo u)\right),
\end{eqnarray*}
i.e.~$A$ and $B$ depend on the quantities $\mu$, $a$ and $\mu^\bo$ but are independent of $\kappa$. The two parameters $\kappa$ and $R_\bo$ must be chosen such that
\begin{equation}\label{eqn:AB:2}
\tilde{f}_\kappa(R_\bo)=R_\bo \quad \text{and} \quad \tilde{f}'_\kappa(R_\bo)=1.
\end{equation}
Denote the position of the first maximum of $\tilde{f}_\kappa$ by $r_\mathrm{max}$. Clearly, $r_\mathrm{max}$ is independent of $\kappa$.
$R_\bo$ is defined as the minimal value where the scattering length of $w_\mu-U_\bo$ equals zero. This means
$$R_\bo:=\min\{r\in(\mu^\bo,r_\mathrm{max}]: \tilde{f}_\kappa(r)=r \;\text{ and }\;\tilde{f}'_\kappa(r)=1\},$$ 
i.e.~$R_\bo$ is defined as the first value of $r$ where $\tilde{f}_\kappa$ is tangential to the straight line $\tilde{f}(r)=r$. This implies in particular that $\tilde{f}_\kappa$ is increasing.
Clearly, $R_\bo$ depends on $\kappa$, hence it remains to prove that suitable $\kappa$, $R_\bo$ exist.
To this end, consider the one-parameter family $\{\tilde{f}_{\kappa}\}_{\kappa\geq 1}$.
\begin{itemize}
\item For $\kappa=1$, we have $\tilde{f}_1(r)=\tilde{j}(r)\leq \tilde{j}(\mu^\bo)=\mu^\bo-\mu a$ for $r\leq \mu^\bo$. As $\tilde{f}_1$ is concave for $\mu^\bo<r<R_\bo$, this implies $\tilde{f}_1(r)<r$ for all $r\in(\mu^\bo, r_\mathrm{max}]$. Consequently, the choice $\kappa=1$ cannot be a solution of~\eqref{eqn:scat:f^tilde}.
\item On the other hand, $\kappa=\frac{\mu^\bo}{\mu^\bo-\mu a}>1$ can neither yield a solution because in this case, $\tilde{f}_{\kappa}(\mu^\bo)=\mu^\bo$ and $\tilde{f}'_{\kappa}(\mu^\bo)>0$, hence $\tilde{f}_{\kappa}>r$ for all $r\in(\mu^\bo, r_\mathrm{max}]$. 
\item Since $\tilde{f}_\kappa(r)=\kappa\tilde{f}_1(r)$, the one-parameter family is strictly increasing in $\kappa$. Together with $\tilde{f}_\kappa(r)<r$ for $\kappa=1$ and $\tilde{f}_\kappa(r)>r$ for $\kappa\geq\frac{\mu^\bo}{\mu^\bo-\mu a}$, this implies that there must be a unique $\kappa_\bo\in(1,\frac{\mu^\bo}{\mu^\bo-\mu a})$ such that $\tilde{f}_{\kappa_\bo}$ satisfies \eqref{eqn:AB:2}. 
\end{itemize}
To obtain an upper bound for $R_\bo$, recall that $\tilde{f}_{\kappa_\bo}$ is increasing and, by construction, $\mathcal{C}^2$ in $[\mu^\bo,R_\bo]$, hence
\begin{eqnarray*}
\kappa_\bo-1&=&\tilde{f}'_{\kappa_\bo}(\mu^\bo)-\tilde{f}'_{\kappa_\bo}(R_\bo)=-\int\limits_{\mu^\bo}^{R_\bo}\tilde{f}_{\kappa_\bo}''(r)\d r
=\tfrac12 a\mu^{1-3\bo}\int\limits_{\mu^\bo}^{R_\bo}\tilde{f}_{\kappa_\bo}(r)\d r\\
&\geq &\tfrac12 a\mu^{1-3\bo}\tilde{f}_{\kappa_\bo}(\mu^\bo)(R_\bo-\mu^\bo)
\gs \kappa_\bo\mu^{1-3\bo}(\mu^\bo-\mu a)(R_\bo-\mu^\bo).
\end{eqnarray*}
With $\frac{\kappa_\bo-1}{\kappa_\bo}<\frac{\mu a}{\mu^\bo}\ls\mu^{1-\bo}$, this yields
$$R_\bo-\mu^\bo\ls\frac{\kappa_\bo-1}{\kappa_\bo(\mu^\bo-\mu a)}\mu^{-1+3\bo}\ls\frac{\mu^{2\bo}}{\mu^\bo-\mu a}=\frac{\mu^\bo}{1-\frac{\mu a}{\mu^\bo}}\ls\mu^\bo$$
for sufficiently small $\mu$.
Due to the respective properties of $\tilde{f}_{\kappa_\bo}$, it is immediately clear that $\fb$ is non-negative, that $\fb\geq j_\mu$ and that $\fb(\z)=\kappa_\bo\, j_\mu(\z)$ for $|\z|\leq\mu^\bo$. To see that $\fb$ is non-decreasing, observe that for $\mu^\bo\leq r\leq R_\bo$, $\tilde{f}'_{\kappa_\bo}(R_\bo)\leq \tilde{f}'_{\kappa_\bo}(r)$ as $\tilde{f}_{\kappa_\bo}$ is concave, hence
$$1=\tilde{f}'_{\kappa_\bo}(R_\bo) \leq \tilde{f}_{\kappa_\bo}'(r)=r(\fb)'(r)+\fb(r)\leq r(\fb)'(r)+1$$
for $\mu^\bo\leq r\leq R_\bo$ as $\fb(r)=r^{-1}\tilde{f}_{\kappa_\bo}(r)\leq 1$. Thus $(\fb)'(r)\geq0$ for all $r\geq0$.

Finally, for the proof of part (d), we refer to \cite[Lemma~5.1(3)]{pickl2015} and the analogous two-dimensional statement in \cite[Lemma~7.10]{jeblick2016}. The idea of the proof is the following: one shows first that the one-particle operator $H^{Z_n}:=-\Delta+\tfrac12\sum_{z_k\in Z_n}(w_\mu-U_\bo)(\cdot-z_k)$ is for each $n\in\mathbb{N}$ a positive operator, where $Z_n$ is an $n$-elemental subset of $\R^3$ such that $B_{R_\bo}(z_k)$ are pairwise disjoint for any two $z_k\in Z_n$. This first assertion follows from the definition of $\fb$ and from the fact that if the ground state energy of $H^{Z_n}$ was negative, the ground state would be strictly positive.
The next step is to prove that the quadratic form $Q(\psi):=\norm{\mathbbm{1}_{|\cdot|\leq R_\bo}\nabla\psi}^2+\tfrac12\lr{\psi,(w_\mu-U_\bo)\psi}$ for $\psi\in H^1(\R^3)$ is nonnegative. Assuming that there exists a $\tilde{\psi}$ such that $Q(\tilde{\psi})<0$, one constructs a set $Z_n$ and a function $\chi_R\in H^1(\R^3)$ such that $\lr{\chi_R,H^{Z_n}\chi_R}<0$ for some $n$, contradicting the positivity of $H^{Z_n}$ which holds for all $n\in\mathbb{N}$. The function $\chi_R$ is constructed in such a way that the part of $\lr{\chi_R,H^{Z_n}\chi_R}$ inside a ball with radius $R$ containing a sufficiently large neighbourhood of $Z_n$ equals $n Q(\tilde{\psi})<0$. The decay of $\chi_R$ outside the ball is chosen such that its positive kinetic energy  is not large enough to cancel this negative term for sufficiently large $n$.
\end{proof}

The next two lemmata provide estimates for expressions containing $\gb$ or $\na\gb$.

\begin{lem}\label{lem:g}
For $\gb$ as in Definition~\ref{def:scat} and sufficiently small $\varepsilon$,
\lemit{
	\item $|\gb(\z)|\ls\frac{\mu}{|\z|}$ ,\label{lem:g:1}
	\item $\norm{\gb}_{L^2(\R^3)}\ls \varepsilon^{2+\bo} 	N^{-1-\frac{\bo}{2}},$ \quad
		 $\onorm{p_1\gbot}=\onorm{\gbot p_1}\ls \mathfrak{e}(t)\varepsilon^{1+\bo}N^{-1-\frac{\bo}{2}},$\label{lem:g:2}
	\item $\norm{\gbot\psi^{N,\varepsilon}(t)}
	\ls\varepsilon N^{-1}$,\label{lem:g:3}
	\item $\onorm{p_1\mathbbm{1}_{\supp{\gb}}(\z_1-\z_2)}=\onorm{\mathbbm{1}_{\supp{\gb}}(\z_1-\z_2)p_1}\ls\mathfrak{e}(t)\varepsilon^{-1+3\bo}N^{-\frac{3}{2}\bo}$,\label{lem:g:4}
	\item $\norm{\mathbbm{1}_{\supp{\gb}}(\z_1-\z_2)\psi^{N,\varepsilon}(t)}\ls\mathfrak{e}(t)\varepsilon^{2\bo-\frac{2}{3}}N^{-\bo}$.\label{lem:g:5}
}
\end{lem}

\begin{proof}
By Lemma~\ref{lem:scat:2}, $\fb(\z)\geq j_\mu(\z)$, hence
\begin{equation*}
\gb(\z)=1-\fb(\z)\leq 1-j_\mu(\z)\leq\tfrac{\mu a}{|\z|}
\end{equation*}
and, since $\supp{\gb}\subseteq\{z\in\R^3:|\z|\leq R_\bo\ls\mu^\bo\}$,
\begin{equation*}
\norm{\gb}^2_{L^2(\R^3)}=\int\limits_{|\z|\leq R_\bo}|\gb(\z)|^2\d\z\ls\mu^2\int\limits_{|\z|\ls \mu^\bo}\tfrac{1}{|\z|^2}\d\z\ls\mu^{2+\bo}.
\end{equation*}
The second part of (b) then follows immediately from Lemma~\ref{lem:pfp:2}.
For part (c), observe that
$\norm{\gbot\psi} \ls \mu\norm{\tfrac{1}{|\z_1-\z_2|}\psi}$
and 
\begin{eqnarray*}
\left\lVert{\tfrac{1}{|\z_1-\z_2|}\psi}\right\rVert^2
&=&\int\limits_{\R^{3(N-1)}} \d z_N\mycdots \d z_2\int\limits_{\R^3} \d\z_1 \overline{\psi(z_1\mydots z_N)}\left(\na_1\cdot\tfrac{\z_1-\z_2}{|\z_1-\z_2|^2}\right) \psi(z_1 \mydots z_N)\\
&=&-2\Re\llr{\na_1\psi,\tfrac{\z_1-\z_2}{|\z_1-\z_2|^2}\psi}
\leq 2\norm{\na_1\psi} \left\lVert\tfrac{1}{|\z_1-\z_2|}\psi\right\rVert.
\end{eqnarray*}
Consequently,
$$\norm{\gbot\psi^{N,\varepsilon}(t)}\ls\mu \norm{\na_1\psi^{N,\varepsilon}(t)}\overset{\text{\ref{lem:a_priori:4}}}{\ls}\mu\varepsilon^{-1}.$$
The proof of (d) works analogously to the proof of Lemma~\ref{lem:w12:3}.
Finally, using Hölder's inequality with $p=3$, $q=\frac{3}{2}$ in the $\d z_1$-integration, we obtain for (e)
\begin{eqnarray*}
\norm{\mathbbm{1}_{\supp{\gb}}(\z_1-\z_2)\psi}^2
&=&\int \d z_N\mycdots \d z_2\int \d z_1\mathbbm{1}_{\supp{\gb}}(\z_1-\z_2)|\psi(z_1\mydots z_N)|^2\\
&\leq& \int \d z_N\mycdots \d z_2 \left(\int \d z_1\mathbbm{1}_{\supp{\gb}}(\z_1-\z_2)\right)^\frac23 \left(\int \d z_1 |\psi(z_1\mydots z_N)|^6\right)^\frac26\\
&\ls& \mu^{2\bo}\int\d z_N\mycdots \d z_2\left(\int\d z_1|\psi(z_1\mydots z_N)|^6\right)^\frac26.
\end{eqnarray*}
Substituting $z_1\mapsto \tilde{z}_1=(x_1,\tfrac{y_1}{\varepsilon})$ and using Sobolev's inequality in the $\d \tilde{z}_1$-integral, we obtain
\begin{eqnarray*}
\left(\int\d z_1|\psi(z_1\mydots z_N)|^6\right)^\frac26
&=&\left(\varepsilon^2\int\d \tilde{z}_1|\psi((x_1,\varepsilon\tilde{y}_1),z_2\mydots z_N)|^6\right)^\frac26\\
&\ls&\varepsilon^\frac23\int\d\tilde{z}_1|\nabla_{\tilde{z}_1}\psi((x_1,\varepsilon\tilde{y}_1),z_2\mydots z_N)|^2\\
&=&\varepsilon^{-\frac43}\int\d z_1\left(|\partial_{x_1}\psi(z_1\mydots z_N)|^2+\varepsilon^2|\nabla_{y_1}\psi(z_1\mydots z_N)|^2\right)
\end{eqnarray*} 
as $\nabla_{\tilde{z}_1}=(\partial_{x_1},\varepsilon\nabla_{y_1})$ and $\d z_1=\varepsilon^2\d\tilde{z}_1$. Hence by Lemma~\ref{lem:a_priori:4},
$$\norm{\mathbbm{1}_{\supp{\gb}}(\z_1-\z_2)\psi^{N,\varepsilon}(t)}^2
\; \ls\; \mu^{2\bo}\varepsilon^{-\frac43}\left(\norm{\partial_{x_1}\psi^{N,\varepsilon}(t)}^2+\varepsilon^2\norm{\nabla_{y_1}\psi^{N,\varepsilon}(t)}^2\right)
\; \ls\; \mu^{2\bo}\varepsilon^{-\frac43}\mathfrak{e}^2(t).$$
\end{proof}

\begin{lem}\label{lem:nabla:g}
For $\gb$ as in Definition~\ref{def:scat}, it holds that
\lemit{
	\item $\norm{\na\gb}_{L^2(\R^3)}\ls N^{-\frac{1}{2}}\varepsilon,$\label{lem:nabla:g:1}
	\item $\onorm{(\na_1\gbot)p_1}\ls\mathfrak{e}(t)N^{-\frac12},$
	\item $\onorm{(\na_1\gbot)\cdot\na_1p_1}\ls\mathfrak{e}(t)N^{-\frac12}\varepsilon^{-1}$.
}
\end{lem}

\begin{proof}
Denote $r\equiv|z|$ and $'\equiv\tfrac{\d}{\d r}$.
As $\gb$ is spherically symmetric, we define $\tilde{g}(r):=r\gb(r)$. Consequently, 
$$|\na\gb(r)|=|{\gb}'(r)|=\tfrac{|\tilde{g}'(r)-\gb(r)|}{r}\leq \tfrac{|\tilde{g}'(r)|}{r}+\tfrac{|\gb(r)|}{r},$$
$\tilde{g}'(r)=1-\tilde{f}'(r)$ and $\tilde{g}''(r)=-\tilde{f}''(r)$ with $\tilde{f}$ from \eqref{eqn:f:tilde}. Hence $\tilde{g}'(R_\bo)=0$ by \eqref{eqn:AB:2} and 
\begin{equation}\label{eqn:lem:nabla:g:1}
|\tilde{g}'(r)|
=|\tilde{g}'(r)-\tilde{g}'(R_\bo)|=\left|\int\limits_r^{R_\bo}\tilde{f}''(\rho)\d\rho\right|
\leq \tfrac12\int\limits_r^{R_\bo}w_\mu(\rho)\rho\d\rho+\tfrac12\int\limits_r^{R_\bo}U_\bo(\rho)\rho\d\rho
\end{equation}
by \eqref{eqn:scat:f^tilde} and as $\tilde{f}(\rho)\leq \rho$. For $0\leq r\leq \mu$, 
$$|\tilde{g}'(r)|\leq\tfrac12\norm{w_\mu}_{L^\infty(\R^3)} \int\limits_0^\mu \rho\d\rho
+\tfrac12 a\mu^{1-3\bo}\int\limits_{\mu^\bo}^{R_\bo}\rho\d\rho
\overset{\text{\ref{lem:scat:5}}}{\ls} 1+\mu^{1-\bo}\ls 1
$$
and $|\gb(r)|\leq 1$, hence $|{\gb}'(r)|\ls \tfrac{1}{r}$. For $\mu\leq r\leq R_\bo$, the first term in \eqref{eqn:lem:nabla:g:1} equals zero, hence $|{\gb}'(r)|\ls \tfrac{\mu^{1-\bo}}{r}+\tfrac{\mu}{r^2}$ by Lemma~\ref{lem:g:1}. Thus
\begin{equation*}
\norm{\nabla\gb}^2_{L^2(\R^3)}=\int\limits_0^\mu|{\gb}'(r)|^2r^2\d r+\int\limits_\mu^{R_\bo}|{\gb}'(r)|^2r^2\d r\ls\mu+\mu^{2-\bo}\ls\mu.
\end{equation*}
The two remaining inequalities follow by Lemma~\ref{lem:pfp}.
\end{proof}

\subsection{Estimate of the kinetic energy}\label{subsec:E_kin}

In this section, we provide a bound for the kinetic energy of $\qp_1\psi^{N,\varepsilon}(t)$. 
The main part of the kinetic energy results from the microscopic structure, which is localised around the scattering centres (on the sets $\overline{\mathcal{C}}_j$ in Definition~\ref{def:cutoffs} below).
We show that the kinetic energy in regions where sufficiently large neighbourhoods around these centres (the sets $\Abar_j\supset\Cbar_j$) are cut out is of lower order. 
To prove this, we will also need the sets $\B_j$, which consist of all $N$-particle configurations where at most two particles interact (one of which is particle $j$).

\begin{definition}\label{def:cutoffs}
Let $d\in(\tfrac56,\bo)$, $j,k\in\{1\mydots N\}$ and define
\begin{eqnarray*}
a_{j,k}&: = &\left\lbrace (z_1\mydots z_N): |z_j-z_k|<\mu^d\right\rbrace,\\
c_{j,k}&: = &\left\lbrace (z_1\mydots z_N): |z_j-z_k|<R_\bo\right\rbrace,\\
a_{j,k}^x&: = &\left\lbrace (z_1\mydots z_N): |x_j-x_k|<\mu^d\right\rbrace.
\end{eqnarray*}
Then the subsets $\Abar_j$,  $\Bbar_j$, $\overline{\mathcal{C}}_j$ and $\Abar^x_j$
of $\R^{3N}$ are defined as
$$\Abar_j\; :=\; \bigcup\limits_{k\neq j}a_{j,k}, \qquad
\Bbar_j\; :=\; \bigcup\limits_{k,l\neq j}a_{k,l}, \qquad
\overline{\mathcal{C}}_j\; :=\; \bigcup\limits_{k\neq j}c_{j,k}, \qquad
\Abar^x_j\; :=\; \bigcup\limits_{k\neq j}a^x_{j,k}
$$
and their complements are denoted by $\A_j$, $\B_j$, $\mathcal{C}_j$ and $\A^x_j$, i.e.~$\A_j:=\R^{3N}\setminus\Abar_j$ etc.
\end{definition}

Note that the characteristic functions $\charBo$ and $\charBbaro$ do not depend on $z_1$, and $\charAbarox$ and $\charAox$ do not depend on any $y$-coordinate. Hence, $\charBo$ and $\charBbaro$ commute with all operators acting exclusively on the first slot of the tensor product, and $\charAbarox$ and $\charAox$ commute with all operators acting only the $y$-coordinates.
The main result of this section is given by the following lemma:

\begin{lem}\label{lem:E_kin:GP}
\begin{equation*}
\norm{\charAo\partial_{x_1}\qp_1\psi^{N,\varepsilon}(t)}\ls \exp\left\{\mathfrak{e}^2(t)+\int_0^t\mathfrak{e}^2(s)\d s\right\}
\left(\alpha_\xi^<(t)+(N\varepsilon^\delta)^{1-\bo}+N^{-1+\bo}\right)^\frac12.
\end{equation*}
\end{lem}

To prove Lemma~\ref{lem:E_kin:GP}, we need several estimates on the cutoff functions $\charAbaro$, $\charAbarox$ and $\charBbaro$.

\begin{lem}\label{lem:cutoffs}
Let $\Abaro$, $\Abar^x_1$ and $\Bbaro$ as in Definition~\ref{def:cutoffs}. Then
\lemit{
	\item $\onorm{\charAbaro p_1}\ls\mathfrak{e}(t)\mu^{\frac{3d}{2}-\frac12}$, \qquad  
	$\onorm{\charAbaro \partial_{x_1}p_1}\ls\norm{\Phi(t)}_{H^2(\R)}\,\mu^{\frac{3d}{2}-\frac12}$,\label{lem:cutoffs:1}
	\item $\norm{\charAbaro\psi}\ls\mu^{d-\frac13}\left(\norm{\partial_{x_1}\psi}+\varepsilon\norm{\nabla_{y_1}\psi}\right)$\; for any $\psi\in L^2(\R^{3N})$,\label{lem:cutoffs:2}
	\item $\norm{\charAbaro\nabla_{y_1}\pc_1\psi^{N,\varepsilon}(t)}\ls\mathfrak{e}(t)N^{-\frac12}
	,$\label{lem:cutoffs:5}
	\item $\norm{\charBbaro\psi}\ls\mu^{d-\frac13}\left(\sum\limits_{k=2}^N(\norm{\partial_{x_k}\psi}^2+\varepsilon^2\norm{\nabla_{y_k}\psi}^2)\right)^\frac12$\; for any $\psi\in L^2(\R^{3N})$,
	\item $\norm{\charBbaro\psi^{N,\varepsilon}(t)}\ls\varepsilon\mathfrak{e}(t)$,
	\label{lem:cutoffs:3}
	\item $\norm{\charAbarox\qc_1\psi^{N,\varepsilon}(t)}^2\ls\mathfrak{e}^2(t)\varepsilon^2(N\varepsilon^\delta)^{1-\bo}$.
	\label{lem:cutoffs:4}
}
\end{lem}

\begin{proof}
In the sense of operators, $\charAbaro=\mathbbm{1}_{\bigcup\limits_{k\geq2}a_{1,k}}\leq \sum\limits_{k=2}^N\mathbbm{1}_{a_{1,k}}$. Hence, for any $\psi\in L^2(\R^{3N})$
\begin{eqnarray*}
\norm{\charAbaro p_1\psi}^2&\leq&\sum\limits_{k=2}^N\Big\llangle{\psi\ket{\phe(z_1)}
\left(\int_{\R^3}\d z_1|\phe(z_1)|^2\mathbbm{1}_{a_{1,k}}(z_1,z_k)\right)
\bra{\phe(z_1)}\psi}\Big\rrangle\\
&\overset{\text{\ref{cor:varphi:1}}}{\ls}&\mathfrak{e}^2(t)N\varepsilon^{-2}\mu^{3d}\norm{p_1\psi}^2
\end{eqnarray*}
and the second part of assertion (a) follows analogously with Lemma~\ref{lem:Phi:1}. Part (b) is proven analogously to Lemma~\ref{lem:g:5}, noting that 
$\left(\int_{\R^3}\d z_1\charAbaro(z_1\mydots z_N)\right)^\frac23\ls N^\frac23\mu^{2d}$. Part (c) follows from this with $\tfrac13-d<-\tfrac12$ and $2d-\tfrac53>0$ 
and as $\norm{\nabla_{y_1}\pc_1\partial_{x_1}\psi}^2\ls\mathfrak{e}^2(t)\varepsilon^{-2}$ by Lemma~\ref{lem:a_priori:4} and
$$\norm{\nabla_{y_1}\partial_{y_1^{(1)}}\pc_1\psi}^2+\norm{\nabla_{y_1}\partial_{y_1^{(2)}}\pc_1\psi}^2
\ls\varepsilon^{-4},$$
where we have put $y_1=(y_1^{(1)},y_1^{(2)})$.
For assertion (d), note that $\charBbaro\leq \sum_{k=2}^N\mathbbm{1}_{\Abar_k}$, hence $\norm{\charBbaro\psi}^2\leq\sum_{k=2}^N\norm{\mathbbm{1}_{\Abar_k}\psi}^2$, and (e) follows from Lemma~\ref{lem:a_priori:4} and since $d>\frac56$. Finally, 
$$\norm{\charAbarox\qc_1\psi}^2
\leq\int\limits_{\R^{3N-1}}\hspace{-0.2cm}\d z_N\mycdots \d y_1
\left(\int\limits_{\R}\d x_1\charAbarox(x_1\mydots x_N)\right)
\left(\sup\limits_{x_1\in\R}|\qc_1\psi(z_1\mydots z_N)|^{2}\right).$$
Note that $\int_{\R}\d x_1\charAbarox(x_1\mydots x_N)\ls N\mu^d$ analogously to above. For the second factor in the integral, the one-dimensional Gagliardo-Nirenberg-Sobolev inequality \cite[Theorem~8.5]{lieb_loss},
$$\sup\limits_{x\in\R}|f(x)|^2\leq \norm{f'}_{L^2(\R)}\norm{f}_{L^2(\R)}\quad \text{ for } f\in H^1(\R),
$$
implies
$$\sup\limits_{x_1\in\R}|\qc_1\psi(z_1\mydots z_N)|^{2}\leq\norm{\qc_1\partial_{x_1}\psi(\cdot,y_1\mydots z_N)}_{L^{2}(\R)}\norm{\qc_1\psi(\cdot,y_1\mydots z_N)}_{L^{2}(\R)}.
$$
Using Cauchy-Schwarz in the $\d y_1\mycdots\d z_N$-integration, we obtain
\begin{eqnarray*}
\norm{\charAbarox\qc_1\psi^{N,\varepsilon}(t)}^2
&\ls &N\mu^d\norm{\qc_1\partial_{x_1}\psi^{N,\varepsilon}(t)}\norm{\qc_1\psi^{N,\varepsilon}(t)}\\
&\overset{\text{\ref{lem:a_priori:4}}}{\ls}&
\mathfrak{e}^2(t)N^{1-d}\varepsilon^{2d+1}
\; =\; \mathfrak{e}^2(t)(N\varepsilon^\delta)^{1-d}\varepsilon^2\varepsilon^{2d-1+\delta(d-1)}.
\end{eqnarray*}
Assertion (f) follows from this because $d<\bo$ and since the last exponent is positive as $0<\delta<\frac25$ and $d>\frac56$.
\end{proof}

We will use some techniques and intermediate results from \cite{NLS}, which are listed in Lemma~\ref{lem:NLS} below. In \cite{NLS}, one considers a class of interaction potentials $\mathcal{W}_{\bo,\eta}$ (\cite[Definition~2.2]{NLS}), which, recalling that $\mu(N,\varepsilon)=\tfrac{\varepsilon^2}{N}$, can be characterised in the following way:

\begin{definition}\label{def:W}
Let $\eta>0$. The set $\mathcal{W}_{\bo,\eta}$ is defined as the set containing all families of interaction potentials
$$v:(0,1)\to L^\infty(\R^3,\R), \quad \mu\mapsto v(\mu),$$ 
such that it holds for all $\mu\in(0,1)$ that $\norm{v(\mu)}_{L^\infty(\R^3)}\ls \mu^{1-3\bo}$, $v(\mu)$ is non-negative and spherically symmetric, $\supp{v(\mu)} \subseteq \left\{\z\in\R^3:|\z|\ls \mu^\bo\right\}$ and 
$	\lim\limits_{\mu\to0}\mu^{-\eta}	\left|b(\mu,v)-b(v)\right|	=0,$
where
$$b(\mu,v):= \mu^{-1}\int_{\R^3}v(\mu,z)\d\z\int_{\R^2}|\chi(y)|^4\d y \quad\text{ and }\quad b(v):=\lim_{\mu\to0}b(\mu,v).$$
\end{definition}

\begin{lem}\label{lem:NLS}
Let $v\in\mathcal{W}_{\bo,\eta}$ for some $\eta>0$.
\lemit{
\item Let $\he:\{z\in\R^3:|z|\leq\varepsilon\}\rightarrow\R$ be the unique solution of $\Delta\he=v(\mu)$ with boundary condition $\he\big|_{|z|=\varepsilon}=0$ and denote $\heij:=\he(\z_i-\z_j).$
Then $$\onorm{p_1(\na_1\heot)}\ls \mathfrak{e}(t) N^{-1}\mu^{-\frac{\bo}{2}}\varepsilon.$$ \label{lem:NLS:2}
\item Let $R\ls \mu^\bo$ such that $\supp v(\mu)\subseteq \{z\in\R^3: |z|\leq R\}$.
Let $\te\in\mathcal{C}^\infty(\R^3, [0,1])$ be spherically symmetric such that 
$\te(z)=1$ for $|z|\leq R$, $\te(z)=0$ for $|z|\geq \varepsilon$, and $\te$ is decreasing for $R<|z|<\varepsilon$. Denote $\teij:=\te(\z_i-\z_j).$
Then $$\norm{\na\te}_{L^\infty(\R^3)}\ls \varepsilon^{-1}.$$\label{lem:NLS:3}
\item Let $\beta_1\in[0,\bo]$. Define
$$\overline{v(\mu,x)}:=\int_{\R^2}\d y_1|\chie(y_1)|^2\int_{\R^2}\d y_2|\chie(y_2)|^2 v(\mu,(x,y_1-y_2))$$
and let $\hbo:[-N^{-\beta_1},N^{-\beta_1}]\rightarrow\R$ be the unique solution of $\tfrac{\d^2}{\d x^2}\hbo=\overline{v(\mu)}$ with boundary condition  $\hbo(\pm N^{-\beta_1})=0$. Then $$\onorm{\pp_1(\tfrac{\d}{\d x_1}\hboot)}\ls\mathfrak{e}(t)N^{-1-\frac{\beta_1}{2}}.$$\label{lem:NLS:4}
\item Let $R\ls \mu^\bo$ such that $\supp v(\mu)\subseteq \{z\in\R^3: |z|\leq R\}$.
For $\beta_1\in[0,\bo]$, let $\tb\in\mathcal{C}^\infty(\R,[0,1])$ be an even function such that $\tb(x)=1$ for $|x|\leq R$, $\tb(x)=0$ for $|x|\geq N^{-\beta_1}$ and $\tb$ is decreasing for $R<|x|<N^{-\beta_1}$. Denote $\tbij:=\tb(x_i-x_j)$. Then 
$$\norm{\tfrac{\d}{\d x}\tb}_{L^\infty(\R)}\ls N^{\beta_1}, \qquad \onorm{\pp_1\left(\tfrac{\d}{\d x_1}\tbot\right)}\ls\mathfrak{e}(t)N^{\frac{\beta_1}{2}}.$$\label{lem:NLS:5}
\item Let $\psi\in L^2(\R^{3N})$ be symmetric in $\{z_1\mydots z_N\}$. Then
\begin{align*}
&\left|\llr{\psi,p_1p_2\left((N-1)v(\mu,z_1-z_2)\right)p_1p_2\psi}
-\llr{\psi,b(v)|\Phi(x_1)|^2\psi}\right|\\
&\hspace{7cm}\ls\mathfrak{e}^2(t)\big(\tfrac{\mu^\bo}{\varepsilon}+N^{-1}+\mu^\eta+\llr{\psi,\hat{n}\psi}\big). 
\end{align*}\label{lem:NLS:1}
\item Let $\psi,\tilde{\psi}\in L^2(\R^{3N})$ and $t_2\in\{q_2,\qp_2\pc_2\}$. Then
$$N\left|\llr{\psi, \qc_1t_2v(\mu,z_1-z_2)p_1p_2\tilde{\psi}}\right|\ls \mathfrak{e}(t)\mu^{-\frac{\bo}{2}}(\norm{\qc_1\psi}+\varepsilon\norm{\nabla_1\qc_1\psi})\norm{\tilde{\psi}}.$$ \label{lem:NLS:6}
\item Let $\psi\in L^2(\R^{3N})$ be symmetric in $\{z_1\mydots z_N\}$. Then
$$N\left|\llr{\psi,\pc_1\pc_2\qp_1\qp_2(\tfrac{\d}{\d x_1}\tbot)(\tfrac{\d}{\d x_1}\hbot)p_1p_2\psi}\right|\ls\mathfrak{e}^2(t)\llr{\psi,\hat{n}\psi}.$$ \label{lem:NLS:7}

}
\end{lem}
\begin{proof}
Parts (a) and (b) follow from Lemma~4.12, Lemma~4.13 and Corollary~4.14 in \cite{NLS} and assertions (c) and (d) are taken from Lemma~4.15 and Corollary~4.16 in \cite{NLS}.
Parts (e) and (f) are (69)-(71) and (74) in \cite{NLS}, and (g) follows from the estimate of (75) in \cite{NLS}.
\end{proof}

\begin{lem}\label{lem:U:in:W}
Let $\eta>0$. Then the family $U_\bo$ is contained in $\mathcal{W}_{\bo,\eta}$.
\end{lem}
\begin{proof}
Note that $\mu^{-1}\int_{\R^3}U_\bo(z)\d z=\tfrac{4\pi}{3}a(R_\bo^3\mu^{-3\bo}-1)=\tfrac{4\pi}{3}ac$ for some $c>0$ by Lemma~\ref{lem:scat:5}, hence $b(\mu,U_\bo)=b(U_\bo)$. The remaining requirements are easily verified.
\end{proof}
\begin{lem}\label{lem:Uf:in:W}
Let $0<\eta<1-\bo$. Then the family $U_\bo\fb$ is contained in $\mathcal{W}_{\bo,\eta}$.
\end{lem}
\begin{proof}
We drop the $\mu$-dependence of the family members and write $U_\bo\fb$ instead of $(U_\bo\fb)(\mu)$. 
By Lemma~\ref{lem:scat}, $\fb$ is spherically symmetric, $0\leq\fb(\z)\leq 1$ and $R_\bo\ls\mu^\bo$, hence 
$\norm{U_\bo\fb}_{L^\infty(\R^3)}\ls \mu^{1-3\bo}$ and $\supp{U_\bo\fb}\subseteq \{\z\in\R^3:|\z|\ls \mu^{\bo}\}$ by Definition~\ref{def:U} of $U_\bo$. Further,
$$
\mu^{-1}\int\limits_{\R^3}U_\bo(\z)\fb(\z)\d\z
\overset{\text{\eqref{eqn:scat(w-U)=0}}}{=}\mu^{-1}\int\limits_{B_\mu(0)} w_\mu(\z)\fb(\z)
\overset{\text{\ref{lem:scat:2}}}{=}\mu^{-1}\kappa_\bo\int\limits_{B_\mu(0)}w_\mu(\z)j_\mu(\z)
\overset{\text{\eqref{eqn:integral:scat}}}{=}\kappa_\bo 8\pi a,
$$
which yields $b(\mu,U_\bo\fb)=\kappa_\bo 8\pi a\int_{\R^2}|\chi(y)|^4\d y$ and consequently
\begin{equation}\label{b=b}
b(U_\bo\fb)=\lim\limits_{\mu\to0}b(\mu,U_\bo\fb)=8\pi a\int\limits_{\R^2}|\chi(y)|^4\d y=b
\end{equation}
by Lemma~\ref{lem:scat:2}. This implies
$$|b(\mu,U_\bo\fb)-b(U_\bo\fb)|=8\pi a (\kappa_\bo-1)\int\limits_{\R^2}|\chi(y)|^4\d y\ls \frac{\mu a}{\mu^\bo-\mu a}\overset{\text{\ref{lem:scat:2}}}{\ls} \mu^{1-\bo}.$$
\end{proof}

\noindent\emph{Proof of Lemma~\ref{lem:E_kin:GP}.}
We will in the following abbreviate $\psi^{N,\varepsilon}(t)\equiv \psi$ and $\Phi(t)\equiv\Phi$. 
\begin{eqnarray}
&\hspace{-1cm}E^\psi&\hspace{-0.7cm}(t)-\mathcal{E}^\Phi(t)\nonumber\\
&=&\norm{\charAo\partial_{x_1}q_1\psi}^2+\norm{\charAo\partial_{x_1}p_1\psi}^2+2\Re\llr{\partial_{x_1}p_1\psi,\charAo\partial_{x_1}q_1\psi}
+\norm{\charAbaro\charBbaro\partial_{x_1}\psi}^2\nonumber\\
&&+\norm{\charAbaro\charBo\partial_{x_1}\psi}^2+\llr{\psi,(-\Delta_{y_1}+\tfrac{1}{\varepsilon^2}V^\perp(\tfrac{y_1}{\varepsilon})-\tfrac{E_0}{\varepsilon^2})\psi}+\tfrac{N-1}{2}\Big\|{\charBbaro\sqrt{w_\mu^{(12)}}\psi}\Big\|^2\nonumber\\
&&+\tfrac{N-1}{2}\llr{\psi,\charBo \left(w_\mu^{(12)}-U_\bo^{(12)}\right)\psi}
+\tfrac{N-1}{2}\llr{\psi,\charBo p_1p_2U_\bo^{(12)}p_1p_2\charBo \psi}\nonumber\\
&&+\tfrac{N-1}{2}\llr{\psi,\charBo (1-p_1p_2)U_\bo^{(12)}(1-p_1p_2)\charBo \psi}\nonumber\\
&&+(N-1)\Re\llr{\psi,\charBo p_1p_2 U_\bo^{(12)}(1-p_1p_2)\charBo \psi}
+\llr{\psi,\Vp(t,z_1)\psi}\nonumber\\
&&-\norm{\Phi'}^2_{L^2(\R)}-\lr{\Phi,\tfrac{b}{2}|\Phi|^2\Phi}-\lr{\Phi,\Vp(t,(x,0))\Phi}\nonumber\\
&\geq&\norm{\charAo\partial_{x_1}q_1\psi}^2\nonumber\\
&&+\norm{\charAbaro\charBo\partial_{x_1}\psi}^2+\llr{\psi,(-\Delta_{y_1}+\tfrac{1}{\varepsilon^2}V^\perp(\tfrac{y_1}{\varepsilon})-\tfrac{E_0}{\varepsilon^2})\psi}\nonumber\\
&&\hspace{7cm}+\tfrac{N-1}{2}\llr{\psi,\charBo \left(w_\mu^{(12)}-U_\bo^{(12)}\right)\psi}\label{eqn:E_kin:GP:1}\\
&&+2\Re\llr{\partial_{x_1}p_1\psi,\charAo\partial_{x_1}q_1\psi}\label{eqn:E_kin:GP:2}\\
&&+\norm{\charAo\partial_{x_1}p_1\psi}^2-\norm{\Phi'}^2_{L^2(\R)}\label{eqn:E_kin:GP:3}\\
&&+\tfrac{b}{2}\left(\llr{\psi,|\Phi(x_1)|^2\psi}-\lr{\Phi,|\Phi|^2\Phi}\right)
+\llr{\psi,\Vp(t,z_1)\psi}-\lr{\Phi,\Vp(t,(x,0))\Phi}\label{eqn:E_kin:GP:4}\\
&&+\tfrac{N-1}{2}\llr{\psi,\charBo p_1p_2U_\bo^{(12)}p_1p_2\charBo\psi}-\tfrac{b}{2}\llr{\psi,|\Phi(x_1)|^2\psi}\label{eqn:E_kin:GP:5}\\
&&+(N-1)\Re\llr{\psi,\charBo(p_1q_2+q_1p_2)U_\bo^{(12)}p_1p_2\charBo\psi}\label{eqn:E_kin:GP:6}\\
&&+(N-1)\Re\llr{\psi,\charBo q_1q_2U_\bo^{(12)}p_1p_2\charBo\psi}.\label{eqn:E_kin:GP:7}
\end{eqnarray}
We will now estimate these expressions separately. For \eqref{eqn:E_kin:GP:1}, recall that $\chie$ is the ground state of 
$-\Delta_{y}+\tfrac{1}{\varepsilon^2}V^\perp(\tfrac{y}{\varepsilon})$ with eigenvalue $\tfrac{E_0}{\varepsilon^2}$, 
hence $(-\Delta_{y_1}+\tfrac{1}{\varepsilon^2}V^\perp(\tfrac{y_1}{\varepsilon})-\tfrac{E_0}{\varepsilon^2})\pc_1=0$
and $-\Delta_{y}+\tfrac{1}{\varepsilon^2}V^\perp(\tfrac{y}{\varepsilon})-\tfrac{E_0}{\varepsilon^2}\geq0$ as operator. Using further that $\charAbarox=(\charAbarox)^2$, $\charBbaro=(\charBbaro)^2$ and their complements commute with $-\Delta_{y_1}+\tfrac{1}{\varepsilon^2}V^\perp(\tfrac{y_1}{\varepsilon})-\tfrac{E_0}{\varepsilon^2}$ and with $\qc_1$, we conclude
\begin{eqnarray*}
\llangle\psi,(-\Delta_{y_1}+\tfrac{1}{\varepsilon^2}V^\perp(\tfrac{y_1}{\varepsilon})-\tfrac{E_0}{\varepsilon^2})\psi\rrangle
&\geq& \llr{\charAbarox\charBo\qc_1\psi,(-\Delta_{y_1}+\tfrac{1}{\varepsilon^2}V^\perp(\tfrac{y_1}{\varepsilon})-\tfrac{E_0}{\varepsilon^2})\charAbarox\charBo\qc_1\psi}\\
&\geq& \norm{\charAbarox\charBo\nabla_{y_1}\qc_1\psi}^2\\
&&-\tfrac{1}{\varepsilon^2}\norm{(V^\perp-E_0)_-}_{L^\infty(\R^2)}\norm{\charAbarox\qc_1\psi}^2\\
&\overset{\text{\ref{lem:cutoffs:4}}}{\gs}&\norm{\charAbaro\charBo\nabla_{y_1}\qc_1\psi}^2-\mathfrak{e}^2(t)(N\varepsilon^\delta)^{1-\bo}
\end{eqnarray*}
because $\charAbarox\geq\charAbaro$ in the sense of operators since $\Abar_1^x\supset\Abaro$. Further,
\begin{eqnarray*}
\norm{\charAbaro\charBo\nabla_{y_1}\psi}^2
&\leq&\norm{\charAbaro\charBo\nabla_{y_1}\pc_1\psi}^2
+\norm{\charAbaro\charBo\nabla_{y_1}\qc_1\psi}^2\\
&&+2\norm{\charAbaro\charBo\nabla_{y_1}\pc_1\psi}\norm{\nabla_{y_1}\qc_1\psi}\\
&\ls&\norm{\charAbaro\charBo\nabla_{y_1}\qc_1\psi}^2+\mathfrak{e}^2(t)N^{-\frac12}
\end{eqnarray*}
by Lemma~\ref{lem:a_priori:4} and Lemma~\ref{lem:cutoffs:5}. Together, this implies
$$
\eqref{eqn:E_kin:GP:1}\gs\norm{\charAbaro\charBo\nabla_1\psi}^2+\tfrac{N-1}{2}\llr{\psi,\charBo \left(w_\mu^{(12)}-U_\bo^{(12)}\right)\psi}-\mathfrak{e}^2(t)\left(N^{-\frac12}+(N\varepsilon^\delta)^{1-\bo}\right).
$$
As $d<\bo$, it follows that $R_\bo<2R_\bo<\mu^d$ for sufficiently small $\mu$, and consequently
$\Cbaro\subset\Abaro$ and $(c_{1,k}\cap\Bo)\cap(c_{1,l}\cap\Bo)=\emptyset$ for $k,l\neq 1$, $l\neq k$. Hence,
$$\charAbaro\charBo\geq \charCbaro\charBo=\mathbbm{1}_{\bigcup\limits_{k\geq2}c_{1,k}\cap\Bo}=\sum\limits_{k=2}^N\mathbbm{1}_{c_{1,k}\cap\Bo}
=\charBo\sum\limits_{k=2}^N\mathbbm{1}_{c_{1,k}}
$$
in the sense of operators, which implies
\begin{eqnarray*}
\eqref{eqn:E_kin:GP:1}
&\gs&(N-1)\norm{\mathbbm{1}_{c_{1,2}}\nabla_1\charBo\psi}^2+\tfrac{N-1}{2}\llr{\charBo\psi, \left(w_\mu^{(12)}-U_\bo^{(12)}\right)\charBo\psi}\\
&&-\mathfrak{e}^2(t)\left(N^{-\frac12}+(N\varepsilon^\delta)^{1-\bo}\right)\\
&\gs&-\mathfrak{e}^2(t)\left(N^{-\frac12}+(N\varepsilon^\delta)^{1-\bo}\right)
\end{eqnarray*}
by Lemma~\ref{lem:scat:peter} 
because $\charBo\psi\in \mathcal{D}(\nabla_1)$ and as $\mathbbm{1}_{c_{1,2}}=\mathbbm{1}_{|z_1-z_2|<R_\bo}$.
Next, observe that
\begin{eqnarray*}
|\eqref{eqn:E_kin:GP:2}|&
\leq& \left|\llr{\partial_{x_1}q_1\psi,\partial_{x_1}p_1\psi}\right|+\left|\llr{\partial_{x_1}q_1\psi,\charAbaro\partial_{x_1}p_1\psi}\right|\\
&\overset{\text{\ref{lem:commutators:2}}}{\leq}& \left|\llr{\hat{n}^{-\frac12}q_1\psi,\partial_{x_1}^2p_1\hat{n}^\frac12_1\psi}\right|+\onorm{\charAbaro\partial_{x_1}p_1}\norm{\partial_{x_1}q_1\psi}\\
&\ls&\norm{\Phi}_{H^2(\R)}\left(\llr{\psi,\hat{n}\psi}+\mathfrak{e}(t)\mu^{-\frac12+\frac{3d}{2}}\right)
\end{eqnarray*}
by Lemma~\ref{lem:a_priori:4} and Lemma~\ref{lem:cutoffs:1}. Due to the identity $\norm{\partial_{x_1}p_1\psi}^2=\norm{\Phi'}_{L^2(\R)}^2\norm{p_1\psi}^2$, 
$$
|\eqref{eqn:E_kin:GP:3}|=\left|-\norm{\charAbaro\partial_{x_1}p_1\psi}^2+\norm{\partial_{x_1}p_1\psi}^2-\norm{\Phi'}^2_{L^2(\R)}\right|
\ls \norm{\Phi}_{H^2(\R)}^2\mu^{-1+3d}+\mathfrak{e}^2(t)\norm{q_1\psi}^2.
$$
Applying Lemma~\ref{lem:psi-Phi} and Lemma~\ref{lem:taylor} to \eqref{eqn:E_kin:GP:4} yields $|\eqref{eqn:E_kin:GP:4}|\ls\mathfrak{e}^2(t)\llr{\psi,\hat{n}\psi}+\mathfrak{e}^3(t)\varepsilon$. 
Using the identity $\fb+\gb=1$ and decomposing $\charBo=\mathbbm{1}-\charBbaro$, we estimate \eqref{eqn:E_kin:GP:5} as
\begin{eqnarray*}
|\eqref{eqn:E_kin:GP:5}|
&\leq&\tfrac12\left|\llr{\psi,p_1p_2\left((N-1)(U_\bo\fb)^{(12)}\right)p_1p_2\psi}
-\llr{\psi,b|\Phi(x_1)|^2\psi}\right|\\
&&+\tfrac{N-1}{2}\left|\llr{\charBo\psi,p_1p_2(U_\bo\gb)^{(12)}p_1p_2\charBo\psi}\right|\\
&&+\tfrac{N-1}{2}\left|\llr{\psi,\charBbaro p_1p_2(U_\bo\fb)^{(12)}p_1p_2\charBbaro\psi}\right|\\
&&+(N-1)\left|\llr{\psi,\charBbaro p_1p_2(U_\bo\fb)^{(12)}p_1p_2\psi}\right|\\
&\overset{\text{\ref{lem:NLS:1}}}{\ls}&\mathfrak{e}^2(t)\left(\tfrac{\mu^\bo}{\varepsilon}+N^{-1}+\llr{\psi,\hat{n}\psi}\right)
+N\norm{\charBbaro\psi}\onorm{p_1(U_\bo\fb)^{(12)}p_1}\\
&&+N\onorm{p_1(U_\bo\gb)^{(12)}p_1}\\
&\ls&\mathfrak{e}^2(t)\left(\tfrac{\mu^\bo}{\varepsilon}+\llr{\psi,\hat{n}\psi}+\mathfrak{e}(t)\varepsilon+\mu^{1-\bo}+\mu^\eta
\right)
\end{eqnarray*}
for any $\eta<1-\bo$ by Lemma~\ref{lem:cutoffs:3} and Lemma~\ref{lem:pfp:1}. Here, we have used that $U_\bo\fb\in\mathcal{W}_{\bo,\eta}$ for $\eta<1-\bo$ by Lemma~\ref{lem:Uf:in:W}, 
$\norm{U_\bo\fb}_{L^1(\R^3)}\ls \mu$ and
$$\norm{U_\bo\gb}_{L^1(\R^3)}=a\mu^{1-3\bo}\int_{\supp{U_\bo}}\d z |\gb(z)|\ls\mu^{2-\bo}
$$
because $|\gb(z)|\leq \gb(\mu^\bo)\leq \kappa_\bo a\mu^{1-\bo}$ on $\supp U_\bo$ by Lemma~\ref{lem:scat:2} and \eqref{eqn:j}. Decomposing $\charBo$ as before and abbreviating $Q_0:=p_1p_2$ and $Q_1:=p_1q_2+q_1p_2$, we find
\begin{eqnarray*}
|\eqref{eqn:E_kin:GP:6}|&\ls& N\left|\llr{\charBbaro\psi,Q_1U_\bo^{(12)}Q_0\psi}\right|
+N\left|\llr{\psi,Q_1U_\bo^{(12)}Q_0\charBbaro\psi}\right|\\
&&+N\left|\llr{\charBbaro\psi,Q_1U_\bo^{(12)}Q_0\charBbaro\psi}\right|
+N\left|\llr{\psi,Q_1U_\bo^{(12)}Q_0\psi}\right|\\
&\overset{\text{\ref{lem:commutators:2}}}{\ls}&N\norm{\charBbaro\psi}\onorm{p_1U_\bo^{(12)}p_1}+ N\left|\llr{\hat{n}^{-\frac12}q_2\psi,p_1U_\bo^{(12)}p_1p_2\hat{n}_1^{\frac12}\psi}\right| \\
&&\ls\mathfrak{e}^2(t)\left(\mathfrak{e}(t)\varepsilon+\llr{\psi,\hat{n}\psi}\right)
\end{eqnarray*}
by Lemma~\ref{lem:fqq:2} and Lemma~\ref{lem:cutoffs:3}. For the last term, we decompose $q=\qc+\pc\qp$, hence
\begin{eqnarray}
|\eqref{eqn:E_kin:GP:7}|&\ls&
N\left|\llr{\charBo\psi, \qc_1q_2U_\bo^{(12)}p_1p_2\charBo\psi}\right|
+N\left|\llr{\psi,\qc_1\qp_2\pc_2U_\bo^{(12)}p_1p_2\mathbbm{1}_{\B_2}\psi}\right|\label{eqn:7:1}\\
&&+N\left|\llr{\charBbaro\psi,\qc_2\qp_1\pc_1U_\bo^{(12)}p_1p_2\charBo\psi}\right|\label{eqn:7:2}\\
&&+N\left|\llr{\charBo\psi,\qp_1\qp_2\pc_1\pc_2U_\bo^{(12)}p_1p_2\charBo\psi}\right|,\label{eqn:7:3}
\end{eqnarray}
where we have exchanged $1\leftrightarrow2$ in the second term of \eqref{eqn:7:1}.
As $\charBo$ and $\charBbaro$ are functions of $(z_2\mydots z_N)$ but not of $z_1$, $\norm{\nabla_1\qc_1\charBo\psi}=\norm{\charBo\nabla_1\qc_1\psi}\leq\norm{\nabla_1\qc_1\psi}\ls\mathfrak{e}(t)$ and analogously $\norm{\qc_1\charBo\psi}\ls\mathfrak{e}(t)\varepsilon$ by Lemma~\ref{lem:a_priori:4}, hence Lemma~\ref{lem:NLS:6} implies $\eqref{eqn:7:1}\ls \mathfrak{e}^2(t)\left(\tfrac{\varepsilon^2}{\mu^\bo}\right)^\frac12$.
By Lemma~\ref{lem:NLS:2}, $U_\bo^{(12)}=\teot\Delta_1\heot$. Integrating by parts in $z_1$ yields
\begin{eqnarray*}
\eqref{eqn:7:2}&\leq&
N\left|\llr{\charBbaro\nabla_1\pc_1\qp_1\psi,\qc_2\teot(\nabla_1\heot)p_1p_2\charBo\psi}\right|\\
&&+N\left|\llr{\charBbaro\psi,\qc_2\qp_1\pc_1(\nabla_1\teot)\cdot(\nabla_1\heot)p_1p_2\charBo\psi}\right|\\
&&+N\left|\llr{\charBbaro\psi,\qc_2\qp_1\pc_1\teot(\nabla_1\heot)p_2\cdot\nabla_1p_1\charBo\psi}\right|\\
&\ls&N\onorm{(\nabla_1\heot)p_1}\left(\norm{\charBbaro\psi}\left(\norm{\nabla\te}_{L^\infty(\R^3)}+\onorm{\nabla_1p_1}\right)
+\norm{\charBbaro\nabla_1\pc_1\qp_1\psi}\right)\\
&\ls&\mathfrak{e}^2(t)\left(\tfrac{\varepsilon^2}{\mu^\bo}\right)^\frac12,
\end{eqnarray*}
where we have used Lemmas~\ref{lem:cutoffs:3},~\ref{lem:a_priori:4},~\ref{lem:NLS:3} and~\ref{lem:NLS:4} and the fact that 
\begin{eqnarray*}
\norm{\charBbaro\nabla_1\pc_1\qp_1\psi}^2&=&
\norm{\charBbaro\pc_1\partial_{x_1}\qp_1\psi}^2+\norm{\qp_1\nabla_{y_1}\pc_1\charBbaro\psi}^2\\
&\leq&\norm{\partial_{x_1}\qp_1\psi}^2+\onorm{\nabla_{y_1}\pc_1}^2\norm{\charBbaro\psi}^2
\ls\mathfrak{e}^2(t).
\end{eqnarray*}
Finally, choosing $\beta_1=\bo$ such that $\pc_1\pc_2U_\bo^{(12)}\pc_1\pc_2=\Theta_{\bo}^{(12)}(\tfrac{\d^2}{\d x_1^2}\hbot)\pc_1\pc_2$ by Lemma~\ref{lem:NLS:4}, we find with the abbreviations $Q_0:=p_1p_2$ and $Q_2:=\qp_1\qp_2\pc_1\pc_2$
\begin{eqnarray*}
\eqref{eqn:7:3}
&\leq&N\left|\llr{\charBo\partial_{x_1}Q_2\psi,\Theta_{\bo}^{(12)}(\tfrac{\d}{\d x_1}\hbot)Q_0\charBo\psi}\right|\\
&&+N\left|\llr{\charBo\psi,Q_2\Theta_{\bo}^{(12)}(\tfrac{\d}{\d x_1}\hbot)\partial_{x_1}Q_0\charBo\psi}\right|\\
&&+N\left|\llr{\charBbaro\psi,Q_2(\tfrac{\d}{\d x_1}\Theta_{\bo}^{(12)})(\tfrac{\d}{\d x_1}\hbot)Q_0\charBo\psi}\right|\\
&&+N\left|\llr{\psi,Q_2(\tfrac{\d}{\d x_1}\Theta_{\bo}^{(12)})(\tfrac{\d}{\d x_1}\hbot)Q_0\charBbaro\psi}\right|\\
&&+N\left|\llr{\psi,Q_2(\tfrac{\d}{\d x_1}\Theta_{\bo}^{(12)})(\tfrac{\d}{\d x_1}\hbot)Q_0\psi}\right|\\
&\overset{\text{\ref{lem:NLS:7}}}{\leq}&N\onorm{(\tfrac{\d}{\d x_1}\hbot)\pp_1}\left(\norm{\partial_{x_1}\qp_1\psi}+\onorm{\partial_{x_1}\pp_1}+\norm{\charBbaro\psi}\norm{\tfrac{\d}{\d x}\Theta_{\bo}}_{L^\infty(\R)}
\right)\\
&&+\mathfrak{e}^2(t)\llr{\psi,\hat{n}\psi}\\
&\ls&\mathfrak{e}^2(t)\left(N^{-\frac{\bo}{2}}+\varepsilon^\bo\left(\tfrac{\varepsilon^2}{\mu^\bo}\right)^\frac12+\llr{\psi,\hat{n}\psi}\right).
\end{eqnarray*}
Thus, $|\eqref{eqn:E_kin:GP:7}|\ls \mathfrak{e}^2(t)\big(N^{-\frac{\bo}{2}}+\left(\tfrac{\varepsilon^2}{\mu^\bo}\right)^\frac12+\llr{\psi,\hat{n}\psi}\big)$. The estimates for \eqref{eqn:E_kin:GP:1} to \eqref{eqn:E_kin:GP:7} imply
$$
\left|E^\psi(t)-\mathcal{E}^\Phi(t)\right|\gs\norm{\charAo\partial_{x_1}q_1\psi}^2-\norm{\Phi}^2_{H^2(\R)}
\left(\llr{\psi,\hat{n}\psi}+(N\varepsilon^\delta)^{1-\bo}+N^{-1+\bo}\right)
$$
because $\mu^{\bo}\varepsilon^{-1}<N^{-\bo}, \varepsilon\mu^{-\frac{\bo}{2}}<(N\varepsilon^\delta)^{\frac{\bo}{2}} $, $\frac{\bo}{2}>1-\bo$ and $\mu^\eta<N^{-1+\bo}$ for sufficiently large $\eta<1-\bo$. 
As $\norm{\charAo\partial_{x_1}\qp_1\psi}\leq \norm{\charAo\partial_{x_1}q_1\psi}+\onorm{\partial_{x_1}\pp_1}\norm{\qc_1\psi}\ls \norm{\charAo\partial_{x_1}q_1\psi}+\mathfrak{e}^2(t)\varepsilon$ by Lemma~\ref{lem:a_priori:4}, this proves the claim with Lemma~\ref{lem:Phi:1}.
\qed

\subsection{Proof of Proposition~\ref{prop:dt_alpha:GP}}\label{subsec:prop:dt_alpha:GP}
Also in this proof, we will abbreviate $\psi^{N,\varepsilon}\equiv\psi$ and $\Phi(t)\equiv\Phi$.
We need to estimate 
\begin{equation}\label{eqn:dt_alpha:GP:1}
\tfrac{\d}{\d t}\alpha_\xi(t)=\tfrac{\d}{\d t}\alpha_\xi^<(t)-N(N-1)\Re\left(\tfrac{\d}{\d t}\llr{\psi,\gbot\hat{r}\psi}\right).
\end{equation}
Proposition~3.4 in \cite{NLS} provides a bound for $|\tfrac{\d}{\d t}\alpha_\xi^<(t)|$  for almost every $t\in\R$. This bound implies
$$\left|\tfrac{\d}{\d t} \alpha_\xi(t)\right|\leq |\gamma_a^<(t)|+\left|\gamma_b^<(t)-N(N-1)\Re\left(\tfrac{\d}{\d t}\llr{\psi,\gbot\hat{r}\psi}\right)\right|$$
for almost every $t$, where we have added the superscript $^<$ to the notation to avoid confusion. The two first terms are given by
\begin{align}\begin{split}
\gamma_a^<(t)\; :=\;\;&\Big|\llr{\psi,\dot{\Vp}(t,\z_1)\psi}-\lr{\Phi,\dot{\Vp}(t,(x,0))\Phi}_{L^2(\R)}\Big|\\
&-2N\Im\llr{\psi,q_1\hat{m}^a_{-1}\big(\Vp(t,\z_1)-\Vp(t,(x_1,0))\big)p_1\psi},\label{eqn:gamma_a^<}
\end{split}\\
\gamma_b^<(t)\;\;:=\;\;&-N(N-1)\Im\llr{\psi,Z^{(12)}\hat{m}\psi}=-N(N-1)\Im\llr{\psi,Z^{(12)}\hat{r}\psi}.\vphantom{\bigg(}\label{eqn:gamma_b^<}
\end{align}
The last equality in \eqref{eqn:gamma_b^<} follows by Lemma~\ref{lem:commutators:5} as
\begin{equation}\label{eqn:dt_alpha:GP:2}
\left[Z^{(12)},\hat{m}\right]=\left[Z^{(12)},p_1p_2(\hat{m}-\hat{m}_2)+(p_1q_2+q_1p_2)(\hat{m}-\hat{m}_1)\right]
=\left[Z^{(12)},\hat{r}\right]
\end{equation}
since $p_1p_2P_{N-1}=p_1p_2P_N=(p_1q_2+q_1p_2)P_N=0$.
For the second term in~\eqref{eqn:dt_alpha:GP:1}, we compute with the aid of Lemma~\ref{lem:derivative_m:3}
\begin{eqnarray}
-N(N-1)\Re\left(\tfrac{\d}{\d t}\llr{\psi,\gbot\hat{r}\psi}\right)
&=&N(N-1)\Im\llr{\psi,\gbot\Big[H(t)-\sum\limits_{j=1}^N h_j(t),\hat{r}\Big]\psi}\label{eqn:dt_alpha:GP:5}\\
&&+N(N-1)\Im\llr{\psi,\left[H(t),\gbot\right]\hat{r}\psi}.\label{eqn:dt_alpha:GP:6}
\end{eqnarray}
We expand the pair interaction in~\eqref{eqn:dt_alpha:GP:5} as
$$\sum_{i<j}w_\mu^{(ij)}=w_\mu^{(12)}+\sum\limits_{j=3}^N\left(w_\mu^{(1j)}+w_\mu^{(2j)}\right)+\sum\limits_{3\leq i<j\leq N}w_\mu^{(ij)}$$
and use 
$$w_\mu^{(12)}-b(|\Phi(x_1)|^2+|\Phi(x_2)|^2)=Z^{(12)}-\tfrac{N-2}{N-1}b(|\Phi(x_1)|^2+|\Phi(x_2)|^2),$$
hence by Lemma~\ref{lem:derivative_m:2} and the symmetry of $\psi$,
\begin{eqnarray}
\eqref{eqn:dt_alpha:GP:5}
&=&N^2(N-1)\Im\llr{\psi,\gbot\left[\Vp(t,\z_1)-\Vp(t,(x_1,0)),\hat{r}\right]\psi}\label{eqn:dt_alpha:GP:7}\\
&&+N(N-1)\Im\llr{\psi,\gbot\left[Z^{(12)},\hat{r}\right]\psi}\label{eqn:dt_alpha:GP:8}\\
&&-2N(N-2)\Im\llr{\psi,\gbot\left[b|\Phi(x_1)|^2,\hat{r}\right]\psi}\label{eqn:dt_alpha:GP:9}\\
&&+2N(N-1)(N-2)\Im\llr{\psi,\gbot\left[w_\mu^{(13)},\hat{r}\right]\psi}\label{eqn:dt_alpha:GP:10}\\
&&+\tfrac12 N(N-1)(N-2)(N-3))\Im\llr{\psi,\gbot\left[w_\mu^{(34)},\hat{r}\right]\psi}\label{eqn:dt_alpha:GP:11}\\
&&-N(N-1)(N-2)\Im\llr{\psi,\gbot\left[b|\Phi(x_3)|^2,\hat{r}\right]\psi}.\label{eqn:dt_alpha:GP:12}
\end{eqnarray}
For~\eqref{eqn:dt_alpha:GP:6}, note that
\begin{eqnarray*}
\left[H(t),\gbot\right]\hat{r}\psi
&=&-\left[H(t),\fbot\right]\hat{r}\psi\\
&=&(\Delta_1\fbot+\Delta_2\fbot)\hat{r}\psi+2(\na_1\fbot)\cdot\na_1\hat{r}\psi+2(\na_2\fbot)\cdot\na_2\hat{r}\psi\\
&=&\left(w_\mu^{(12)}- U_\bo^{(12)}\right)\fbot\hat{r}\psi-2(\na_1\gbot)\cdot\na_1\hat{r}\psi-2(\na_2\gbot)\cdot\na_2\hat{r}\psi,
\end{eqnarray*}
hence
\begin{eqnarray}
\eqref{eqn:dt_alpha:GP:6}&=&-4N(N-1)\Im\llr{\psi,(\na_1\gbot)\cdot\na_1\hat{r}\psi}\label{eqn:dt_alpha:GP:13}\\
&&+N(N-1)\Im\llr{\psi,\left(w_\mu^{(12)}- U_\bo^{(12)}\right)\fbot\hat{r}\psi}\label{eqn:dt_alpha:GP:14}.
\end{eqnarray}
We now identify some of the terms in $|\tfrac{\d}{\d t}\alpha_\xi(t)|$ with the expressions in Proposition~\ref{prop:dt_alpha:GP}: $\eqref{eqn:dt_alpha:GP:7}=\gamma_a(t)$, $\eqref{eqn:dt_alpha:GP:13}=\gamma_c(t)$, $\eqref{eqn:dt_alpha:GP:10}+\eqref{eqn:dt_alpha:GP:12}=\gamma_d(t)$, $\eqref{eqn:dt_alpha:GP:11}=\gamma_e(t)$ and $\eqref{eqn:dt_alpha:GP:9}=\gamma_f(t)$.
The remaining terms are $\gamma_a^<(t)$, $\gamma_b^<(t)$, \eqref{eqn:dt_alpha:GP:8} and \eqref{eqn:dt_alpha:GP:14}. The latter yield
\begin{eqnarray*}
\gamma_b^<(t)&&\hspace{-0.9cm}+\,\eqref{eqn:dt_alpha:GP:8}+\eqref{eqn:dt_alpha:GP:14}\\
&=&N(N-1)\Im\bigg(-\llr{\psi,Z^{(12)}\hat{r}\psi}+\llr{\psi,(1-\fbot)\left[Z^{(12)},\hat{r}\right]\psi}\\
&&\hspace{5.05cm}+\llr{\psi,(w_\mu^{(12)}- U_\bo^{(12)})\fbot\hat{r}\psi}\bigg)\\
&=&N(N-1)\Im\bigg(-\llr{\psi,\gbot\hat{r}Z^{(12)}\psi}-\llr{Z^{(12)}\fbot\psi,\hat{r}\psi}\\
&&\hspace{5cm}+\llr{(w_\mu^{(12)}- U_\bo^{(12)})\fbot\psi,\hat{r}\psi}\bigg).
\end{eqnarray*}
Observing that
$$Z^{(12)}\fbot=
\left(w_\mu^{(12)}- U_\bo^{(12)}\right)\fbot+ U_\bo^{(12)}\fbot-\tfrac{b}{N-1}\left(|\Phi(x_1)|^2+|\Phi(x_2)|^2\right)\fbot,$$
we conclude
\begin{eqnarray}
\gamma_b^<(t)&&\hspace{-0.9cm}+\,\eqref{eqn:dt_alpha:GP:8}+\eqref{eqn:dt_alpha:GP:14}\\
&=&-N(N-1)\Im\llr{\psi,\gbot\hat{r}Z^{(12)}\psi}\nonumber\\
&&-N(N-1)\Im\llr{\psi,\left( U_\bo^{(12)}-\tfrac{b}{N-1}\left(|\Phi(x_1)|^2+|\Phi(x_2)|^2\right)\right)(1-\gbot)\hat{r}\psi}\nonumber\\
&=&-N(N-1)\Im\llr{\psi,\gbot\hat{r}Z^{(12)}\psi}\label{eqn:dt_alpha:GP:15}\\
&&-N\Im\llr{\psi,b(|\Phi(x_1)|^2+|\Phi(x_2)|^2)\gbot\hat{r}\psi}\label{eqn:dt_alpha:GP:16}\\
&&-N\Im\Big\llangle\psi,(b_\bo-b)
(|\Phi(x_1)|^2+|\Phi(x_2)|^2)\hat{r}\psi\Big\rrangle\label{eqn:dt_alpha:GP:17}\\
&&-N(N-1)\Im\llr{\psi,\tilde{Z}^{(12)}\hat{m}\psi},\label{eqn:dt_alpha:GP:18}
\end{eqnarray}
where we have used the fact that $\Im\llr{\psi,\tilde{Z}^{(12)}\hat{r}\psi}=\Im\llr{\psi,\tilde{Z}^{(12)}\hat{m}\psi}$ as in \eqref{eqn:dt_alpha:GP:2}. Hence $\eqref{eqn:dt_alpha:GP:15}+\eqref{eqn:dt_alpha:GP:16}+\eqref{eqn:dt_alpha:GP:17}=\gamma_b(t)$ and $\gamma_a^<(t)+\eqref{eqn:dt_alpha:GP:18}=\gamma^<(t)$.
\qed

\subsection{Proof of Proposition~\ref{prop:gamma:GP}}
\subsubsection{Proof of the bound for $\gamma^<(t)$}\label{subsec:gamma_<}
The main tool for the estimate of $\gamma^<(t)$ is Proposition~3.5 from \cite{NLS}, which we apply to the interaction potential $U_\bo\fb$ (which, given $w$, is completely determined by a choice for $\mu$ and $\bo$, cf.\ Definitions~\ref{def:U} and \ref{def:scat}). Let us therefore first verify that the assumptions of this proposition are fulfilled, i.e.~that 
\begin{enumerate}[label={(\alph*)}]
\item ${\mu^\bo}/{\varepsilon}\rightarrow 0$, ${\varepsilon^2}/{\mu^\bo}\rightarrow0$ and $\xi\leq\tfrac{\bo}{4}$ (for $\xi$ from Definition~\ref{def:alpha}),
\item the family $U_\bo\fb$ is contained in $\mathcal{W}_{\bo,\eta}$ for some $\eta>0$.
\end{enumerate} 
We will in the sequel drop the $\mu$-dependence of the family members and simply write $U_\bo\fb$ instead of $(U_\bo\fb)(\mu)$. 
Part (a) is satisfied since $\mu^\bo/\varepsilon\to 0$ because $\bo>\tfrac56>\frac{1}{2}$. Further,
${\varepsilon^2}/{\mu^\bo}=\left(N\varepsilon^\delta\right)^\bo \varepsilon^{2-\bo(2+\delta)}<\left(N\varepsilon^\delta\right)^\bo\rightarrow0$
because $\bo\leq \tfrac{2}{2+\delta}$, and finally $\xi<\tfrac{\bo}{6}$ by assumption. Part (b) is proven in Lemma~\ref{lem:Uf:in:W}.

Proposition~3.5 in \cite{NLS} implies that for any $\beta_1\in(0,\bo]$, $\gamma^<_a(t)$ and $\gamma^<_b(t)$ are bounded by
\begin{equation}\begin{split}\label{eqn:gamma<}
|\gamma_a^<(t)|+|\gamma_b^<(t)|
\ls \mathfrak{e}(t)\exp\bigg\{\mathfrak{e}^2(t)+\int_0^t\mathfrak{e}^2&(s)\d s\bigg\}
\Big(\big|E^\psi_{U_\bo\fb}(t)-\mathcal{E}_{U_\bo\fb}^\Phi(t)\big|+\llr{\psi,\hat{n}\psi}\\
&+\tfrac{\mu^\bo}{\varepsilon}+\left(\tfrac{\varepsilon^2}{\mu^\bo}\right)^\frac12+N^{-\frac{\beta_1}{2}}+N^{-1+\beta_1+\xi}+\mu^{\eta}\Big),
\end{split}\end{equation}
where $E^\psi_{U_\bo\fb}(t)$ and $\mathcal{E}_{U_\bo\fb}^\Phi(t)$ denote the respective quantities corresponding to \eqref{E^psi} and \eqref{E^Phi} but with $w_\mu$ replaced by $U_\bo\fb$ and $b$ by $b(U_\bo\fb)$. 
Note that the energy difference $\big|E^\psi_{U_\bo\fb}(t)-\mathcal{E}_{U_\bo\fb}^\Phi(t)\big|$ enters only in the estimate of $\gamma_b^<(t)$, exclusively in the term (24) in \cite[Proposition~3.4]{NLS}, which is given by
\begin{equation}\label{eqn:24}
-2N(N-1)\Im\llr{\psi^{N,\varepsilon}(t),\qp_1\qp_2\hat{m}^a_{-1}\pc_1\pc_2(U_\bo\fb)^{(12)} p_2\pc_1\qp_1\psi^{N,\varepsilon}(t)}.
\end{equation}
To obtain a bound in terms of $|E^\psi(t)-\mathcal{E}^\Phi(t)|$ instead of $|E^\psi_{U_\bo\fb}(t)-\mathcal{E}_{U_\bo\fb}^\Phi(t)|$, we need a new estimate of \eqref{eqn:24} by means of Lemma~\ref{lem:E_kin:GP}. 

Define $\hat{l}:=N\hat{m}^a_{-1}$. 
We apply Lemma~\ref{lem:NLS:4} and~\ref{lem:NLS:5} with the choice $\beta_1=0$, i.e.
$\tz\tfrac{\d^2}{\d x^2}\hz=\overline{U_\bo\fb}$, where  
$\pc_1\pc_2(U_\bo\fb)^{(12)}\pc_1\pc_2=\pc_1\pc_2\overline{U_\bo\fb}(x_1-x_2).$
Integrating by parts and subsequently inserting the identity $\charAo+\charAbaro$ before $\partial_{x_1}\qp_1\psi$ yields
\begin{eqnarray}
\eqref{eqn:24}&\ls&N\left|\llr{\hat{l}\qp_1\qp_2\psi,\tzot(\tfrac{\d^2}{\d x_1^2}\hzot)p_2\pc_1\qp_1\psi}\right|\nonumber\\
&\leq& N\left|\llr{\charAo\partial_{x_1}\qp_1\psi,\qp_2\tzot(\tfrac{\d}{\d x_1}\hzot)p_2\pc_1 \hat{l}_1\qp_1\psi}\right|\label{eqn:24:1}\\
&&+N\left|\llr{\hat{l}\qp_1\qp_2\psi,\tzot(\tfrac{\d}{\d x_1}\hzot)p_2\pc_1 \charAo\partial_{x_1}\qp_1\psi}\right|\label{eqn:24:2}\\
&&+N\left|\llr{\partial_{x_1}\qp_1\psi,\charAbaro\qp_2(\tfrac{\d}{\d x_1}\hzot)\tzot\pc_1\pc_2 \pp_2\hat{l}_1\qp_1\psi}\right|\label{eqn:24:3}\\
&&+N\left|\llr{\charAbaro\pp_2(\tfrac{\d}{\d x_1}\hzot)\tzot\pc_1\pc_2\qp_2\hat{l}\qp_1\psi, \partial_{x_1}\qp_1\psi}\right|\label{eqn:24:4}\\
&&+N\left|\llr{\hat{l}\qp_1\qp_2\psi,(\tfrac{\d}{\d x_1}\tzot)(\tfrac{\d}{\d x_1}\hzot)p_2\pc_1\qp_1\psi}\right|.\label{eqn:24:5}
\end{eqnarray}
To estimate \eqref{eqn:24:1}, note that $\charAo\partial_{x_1}\qp_1\psi$ and $\hat{l}_1\pc_1\qp_1\psi$ are symmetric in $\{z_2\mydots z_N\}$, hence Lemma~\ref{lem:Gamma:Lambda} implies
\begin{eqnarray*}
\eqref{eqn:24:1}
&\ls& N\norm{\charAo\partial_{x_1}\qp_1\psi}\onorm{\pp_2(\tfrac{\d}{\d x_1}\hzot)}\left(\norm{\hat{l}_1\qp_1\qp_2\psi}+N^{-\frac12}\norm{\hat{l}_1\qp_1\psi}\right)\\
&\overset{\text{\ref{lem:NLS:4}}}{\ls}&\mathfrak{e}(t)\left(\norm{\charAo\partial_{x_1}\qp_1\psi}^2+\llr{\psi,\hat{n}\psi}+N^{-\frac12}\norm{\charAo\partial_{x_1}\qp_1\psi}\right)
\end{eqnarray*}
by Lemma~\ref{lem:NLS:4} because $\norm{\hat{l}_1\qp_1\psi}\ls1$ by Lemma~\ref{lem:l:3} and
$\norm{\hat{l}\qp_1\qp_2\psi}\ls \norm{\hat{n}\psi}$ by Lemma~\ref{lem:fqq:2}. \eqref{eqn:24:2} is immediately controlled by 
$$
\eqref{eqn:24:2}\ls\mathfrak{e}(t)\norm{\charAo\partial_{x_1}\qp_1\psi}\llr{\psi,\hat{n}\psi}^\frac12
\ls \mathfrak{e}(t)\left(\norm{\charAo\partial_{x_1}\qp_1\psi}^2+\llr{\psi,\hat{n}\psi}\right). 
$$
Similarly, $\eqref{eqn:24:5}\ls\mathfrak{e}(t)\llr{\psi,\hat{n}\psi}$. To estimate the two remaining terms, let $$(s^\Phi_2,t^\Phi_2)\in\{(\pp_2,\qp_2),(\qp_2,\pp_2)\}$$ and $\hat{l}_j\in\{\hat{l},\hat{l}_1\}$. By Lemma~\ref{lem:cutoffs:2} and Lemma~\ref{lem:a_priori:4},
\begin{eqnarray*}
\eqref{eqn:24:3}+\eqref{eqn:24:4}
&\ls& N\norm{\partial_{x_1}\qp_1\psi}
	\norm{\charAbaro s^\Phi_2(\tfrac{\d}{\d x_1}\hzot)\tzot t^\Phi_2\pc_2\pc_1\hat{l}_j\qp_1\psi}\\
&\ls& N\mathfrak{e}(t)\mu^{d-\frac13}\Big(
	\norm{s^\Phi_2(\tfrac{\d^2}{\d x_1^2}\hzot)\tzot t^\Phi_2\pc_2\pc_1\hat{l}_j\qp_1\psi}\\
&&\hphantom{ N\mathfrak{e}(t)\mu^{d-\frac13}\Big(}\;
	+\norm{s^\Phi_2(\tfrac{\d}{\d x_1}\hzot)(\tfrac{\d}{\d x_1}\tzot) t^\Phi_2\pc_2\pc_1\hat{l}_j\qp_1\psi}\\
&&\hphantom{N\mathfrak{e}(t)\mu^{d-\frac13}\Big(}\;
	+\norm{s^\Phi_2(\tfrac{\d}{\d x_1}\hzot)\tzot t^\Phi_2\pc_2\pc_1\partial_{x_1}\hat{l}_j\qp_1\psi}\\
&&\hphantom{ N\mathfrak{e}(t)\mu^{d-\frac13}\Big(}\;
	+\varepsilon\norm{s^\Phi_2(\tfrac{\d}{\d x_1}\hzot)\tzot t^\Phi_2\pc_2\nabla_{y_1}\pc_1\hat{l}_j\qp_1\psi}
	\Big)\\
&\ls&\mathfrak{e}(t) N\mu^{d-\frac13}\onorm{\pp_2\overline{U_\bo\fb}(x_1-x_2)}
\\
&&+\mathfrak{e}(t) N\mu^{d-\frac13}
\onorm{(\tfrac{\d}{\d x_1}\hzot)\pp_2}\left(\norm{\tfrac{\d}{\d x}\tz}_{L^\infty(\R)}
+N^\xi\mathfrak{e}^2(t)  +\varepsilon\onorm{\nabla_{y_1}\pc_1}\right)
\end{eqnarray*}
as $\norm{\partial_{x_1}\hat{l}_j\qp_1\psi}\ls\onorm{\hat{l}_j}\norm{\partial_{x_1}\qp_1\psi}
\ls N^\xi\mathfrak{e}(t)$ by Lemma~\ref{lem:commutators:2} and Lemma~\ref{lem:l}.
The last line is bounded by $\mathfrak{e}^3(t)\mu^{d-\frac13}N^\xi$ by Lemma~\ref{lem:NLS:4} and~\ref{lem:NLS:5} and Lemma~\ref{lem:a_priori:4}. 
Finally, note that $|x_1-x_2|<R_\bo\ls\mu^\bo$ for $(x_1-x_2)\in \supp \overline{U_\bo\fb}$, hence
\begin{eqnarray*}
\onorm{\pp_2\overline{U_\bo\fb}(x_1-x_2)}&=&\onorm{\pp_2\mathbbm{1}_{|\cdot|<R_\bo}(x_1-x_2)}\norm{\overline{U_\bo\fb}}_{L^\infty(\R)}\\
&\overset{\text{\ref{lem:pfp:4}}}{\ls}&\mathfrak{e}(t)\norm{\overline{U_\bo\fb}}_{L^\infty(\R)}\norm{\mathbbm{1}_{|\cdot|<R_\bo}}_{L^2(\R)}
\ls\mathfrak{e}(t)N^{-1}\mu^{-\frac{\bo}{2}}.
\end{eqnarray*}
The last bound follows since $\norm{\mathbbm{1}_{|x_1-x_2|<R_\bo}}_{L^2(\R)}\ls\mu^\frac{\bo}{2}$ and as
\begin{eqnarray*}
\left|\overline{U_\bo\fb}(x)\right|&=&\int\limits_{\R^2}\d y_1|\chie(y_1)|^2\int\limits_{\R^2}\d y_2|\chie(y_2)|^2 (U_\bo\fb)(x,y_1-y_2)\\
&\leq&\varepsilon^{-2}\int\limits_{\R^2}\d y_1|\chie(y)|^2\int\limits_{|y|<R_\bo}\d y\norm{U_\bo\fb}_{L^\infty(\R^3)}
\; \ls\; \varepsilon^{-2}\mu^{1-\bo},
\end{eqnarray*} 
where we have used that $|y|<R_\bo$ for $(x,y)\in\supp U_\bo\fb$ as above and that  $\chie$ is normalised and $\norm{U_\bo\fb}_{L^\infty(\R^3)}\ls\mu^{1-3\bo}$.
Hence,
\begin{equation}
\eqref{eqn:24}\ls\mathfrak{e}(t)\exp\left\{\mathfrak{e}^2(t)+\int_0^t\mathfrak{e}^2(s)\d s\right\}\left(\alpha_\xi^<(t)+(N\varepsilon^\delta)^{1-\bo}+N^{-1+\bo}+\mu^{d-\frac13-\frac{\bo}{2}}\right),
\end{equation}
where we have used Lemma~\ref{lem:E_kin:GP} and the fact that $\mu^{d-\frac13}N^\xi<\mu^{d-\frac13-\frac{\bo}{2}}$ and $N^{-\frac12}<\mu^{d-\frac13-\frac{\bo}{2}}$.\\

Combining this new bound for \eqref{eqn:24} with the remaining estimates of \cite[Proposition~3.5]{NLS}, we find
$$
|\gamma^<(t)|\ls \mathfrak{e}(t)\exp\left\{\mathfrak{e}^2(t)+\int_0^t\mathfrak{e}^2(s)\d s\right\}\left(\alpha_\xi^<(t)+\left(N\varepsilon^\delta\right)^{1-\bo}+N^{-1+\bo+\xi}+\mu^{d-\frac13-\frac{\bo}{2}}\right),$$
where we have chosen $\beta_1=\bo$ and used that $-1+\frac{3\bo}{2}+\xi>0$, $\mu^{1-\bo}<N^{-1+\bo+\xi}$, $\tfrac{\mu^\bo}{\varepsilon}<N^{-\frac{\bo}{2}}<N^{-1+\bo+\xi}$ and $\varepsilon\mu^{-\frac{\bo}{2}}<(N\varepsilon^\delta)^{\frac{\bo}{2}}<(N\varepsilon^\delta)^{1-\bo}$.
\qed

\subsubsection{Proof of the bound for $\gamma_a(t)$}\label{subsec:gamma_a}
By definition of $\hat{r}$ and with Lemma~\ref{lem:taylor}, Lemma~\ref{lem:g:2} and Lemma~\ref{lem:l:2},  we compute
\begin{eqnarray*}
|\eqref{gamma:GP:a}|
&\ls& N^3\left|\llr{\left(\Vp(t,z_1)-\Vp(t,(x_1,0))\right)\psi,\gbot\left(p_1p_2\hat{m}^b+(p_1q_2+q_1p_2)\hat{m}^a\right)\psi}\right|\\
&&+N^3\left|\llr{\psi,\gbot\left(p_1p_2\hat{m}^b+(p_1q_2+q_1p_2)\hat{m}^a\right)\left(\Vp(t,z_1)-\Vp(t,(x_1,0))\right)\psi}\right|\\
&\leq& 2N^3\norm{(\Vp(t,z_1)-\Vp(t,(x_1,0)))\psi}\onorm{\gbot p_1}\left(\onorm{\hat{m}^a}+\onorm{\hat{m}^b}\right)\\
&\ls&\mathfrak{e}^3(t)N^{1+\xi-\frac{\bo}{2}}\varepsilon^{2+\bo}
=\mathfrak{e}^3(t)\left(N\varepsilon^\delta\right)^{1+\xi-\frac{\bo}{2}}\varepsilon^{2+\bo-\delta(1+\xi-\frac{\bo}{2})}
<\mathfrak{e}^3(t)\varepsilon^2
\end{eqnarray*}
as $\bo-\delta(1+\xi-\frac{\bo}{2})>0$ and since $1+\xi-\frac{\bo}{2}>0$.
\qed

\subsubsection{Proof of the bound for $\gamma_b(t)$}
\emph{Estimate of \eqref{gamma:GP:b:1:1}.}
By Lemma~\ref{lem:g:2}, Lemma~\ref{lem:l:2} and Lemma~\ref{lem:Phi:1} and as $-1-\frac{\bo}{2}+\xi<0$,
$$
|\eqref{gamma:GP:b:1:1}|\ls N\norm{\Phi}^2_{L^\infty(\R)}\onorm{\gbot p_1}\left(\onorm{\hat{m}^a}+\onorm{\hat{m}^b}\right)
\ls\mathfrak{e}^3(t)N^{-1-\frac{\bo}{2}+\xi}\varepsilon^{1+\bo}
<\mathfrak{e}^3(t)\varepsilon^{1+\bo}.
$$

\noindent\emph{Estimate of \eqref{gamma:GP:b:1:2}.}
Note that $b_\bo=b(U_\bo\fb)=b$ by \eqref{b=b}, hence $\eqref{gamma:GP:b:1:2}=0$.\\

\noindent\emph{Estimate of \eqref{gamma:GP:b:2}.} By definition of $\hat{r}$ and due to the symmetry of $\psi$, 
\begin{eqnarray*}
|\eqref{gamma:GP:b:2}|
&\leq& N^2\left|\llr{\psi,\gbot p_1\hat{m}^bp_2Z^{(12)}\psi}+2\llr{\psi,\gbot p_1q_2\hat{m}^ap_1Z^{(12)}\psi}\right|\\
&\ls &N^2\onorm{p_1\gbot}\left(\onorm{\hat{m}^a}+\onorm{\hat{m}^b}\right)\norm{p_1\left(w_\mu^{(12)}-\tfrac{b}{N-1}(|\Phi(x_1)|^2+|\Phi(x_2)|^2)\right)\psi}\\
&\ls&\mathfrak{e}(t)N^{-\frac{\bo}{2}+\xi}\varepsilon^{1+\bo}\left(\norm{p_1w_\mu^{(12)}\psi}+N^{-1}\norm{\Phi}_{L^\infty(\R)}^2\right)\\
&\ls&\mathfrak{e}^3(t)N^{-1-\frac{\bo}{2}+\xi}\varepsilon^{1+\bo} 
\; <\; \mathfrak{e}^3(t)\varepsilon^{1+\bo}
\end{eqnarray*}
as a consequence of Lemma~\ref{lem:g:2}, Lemma~\ref{lem:l:2}, Lemma~\ref{lem:w12:4} and Lemma~\ref{lem:Phi:1}.

\subsubsection{Proof of the bound for $\gamma_c(t)$}

\begin{eqnarray*}
|\eqref{gamma:GP:c}|&\ls& N^2\left|\llr{\mathbbm{1}_{\supp{\gb}}(\z_1-\z_2)\psi,(\na_1\gbot)\cdot
\left(p_2\na_1(p_1\hat{m}^b+q_1\hat{m}^a)\psi+\na_1p_1q_2\hat{m}^a\psi\right)}\right|\\
&\leq &N^2\norm{\mathbbm{1}_{\supp{\gb}}(\z_1-\z_2)\psi}\Big(\onorm{(\na_1\gbot)p_2}\onorm{\na_1p_1}\onorm{\hat{m}^b}\\
&&\hspace{2.9cm}+\onorm{(\na_1\gbot)\na_1p_1}\onorm{\hat{m}^a}+\onorm{(\na_1\gbot)p_2}\norm{\na_1q_1\hat{m}^a\psi}\Big)\\
&&\ls\mathfrak{e}^2(t)\varepsilon^{2\bo-\frac{5}{3}}N^{\frac12+\xi-\bo}
\; <\; \mathfrak{e}^2(t)N^{\frac12+\xi-\bo}
\; <\; \mathfrak{e}^2(t)N^{-1+\xi+\bo}
\end{eqnarray*}
because $2\bo-\frac53>0$ and $\frac12-\bo<-1+\bo$ as $\bo>\frac56$.
In the third step, we have used Lemma~\ref{lem:g:5}, Lemma~\ref{lem:l:2}, Lemma~\ref{lem:a_priori:4}, Lemma~\ref{lem:nabla:g} and the fact that
\begin{eqnarray*}
\norm{\na_1q_1\hat{m}^a\psi}
&\overset{\text{\ref{lem:commutators:2}}}{\leq}& \norm{p_1\hat{m}^a_1\na_1(1-p_1)\psi}+\norm{q_1\hat{m}^a\na_1(1-p_1)\psi}\\
&\overset{\text{\ref{lem:l:1}}}{\ls}&\onorm{\hat{m}^a}\left(\norm{\na_1\psi}+\norm{\na_1 p_1\psi}\right)
\; \overset{\text{\ref{lem:a_priori:4}}}{\ls}\; N^{-1+\xi}\varepsilon^{-1}.
\end{eqnarray*}
\qed

\subsubsection{Proof of the bound for $\gamma_d(t)$}
\emph{Estimate of \eqref{gamma:GP:d:1}.}
With Lemma~\ref{lem:g:2}, Lemma~\ref{lem:l:2} and Lemma~\ref{lem:Phi:1}, 
\begin{eqnarray*}
|\eqref{gamma:GP:d:1}|
&\ls &N^3\left|\llr{\psi,\gbot p_1p_2b\left[|\Phi(x_3)|^2,\hat{m}^b\right]\psi}\right|\\
&&+N^3\left|\llr{\psi,\gbot\left(p_1q_2+q_1p_2\right)b\left[|\Phi(x_3)|^2,\hat{m}^a\right]\psi}\right|\\
&\ls& N^3\onorm{\gbot p_1}\norm{\Phi}^2_{L^\infty(\R)}\left(\onorm{\hat{m}^a}+\onorm{\hat{m}^b}\right)\\
&\ls& \mathfrak{e}^3(t)N^{1+\xi-\frac{\bo}{2}}\varepsilon^{1+\bo}
<\mathfrak{e}^3(t)\left(N\varepsilon^\delta\right)^{1+\xi-\frac{\bo}{2}}
\end{eqnarray*}
analogously to the estimate of $\gamma_a(t)$.\\

\noindent\emph{Estimate of \eqref{gamma:GP:d:2}.}
Observe first that
$$\hat{r}=\hat{m}^bp_1p_2+\hat{m}^a(p_1(1-p_2)+(1-p_1)p_2)=
\hat{m}^a(p_1+p_2)+(\hat{m}^b-2\hat{m}^a)p_1p_2.$$
As a consequence,
\begin{eqnarray}
|\eqref{gamma:GP:d:2}|&\ls& N^3\left|\llr{\psi,\gbot[w_\mu^{(13)},\hat{r}\,]\psi}\right|\nonumber\\
&\leq& N^3\left|\llr{\psi,\gbot p_2[w_\mu^{(13)},\hat{m}^a]\psi}\right|\label{eqn:gamma:GP:d:1}\\
&&+N^3\left|\llr{\psi,\gbot w_\mu^{(13)}p_1\hat{m}^a\psi}\right|\label{eqn:gamma:GP:d:2}\\
&&+N^3\left|\llr{\psi,\gbot p_1(\hat{m}^a+ p_2(\hat{m}^b-2\hat{m}^a))p_1w_\mu^{(13)}\psi}\right| \label{eqn:gamma:GP:d:3}\\
&&+N^3\left|\llr{\psi,\gbot w_\mu^{(13)}p_2p_1(\hat{m}^b-2\hat{m}^a)\psi}\right|.\label{eqn:gamma:GP:d:4}
\end{eqnarray}
We estimate~\eqref{eqn:gamma:GP:d:1} to~\eqref{eqn:gamma:GP:d:4} separately.
$$
\eqref{eqn:gamma:GP:d:1}=N^3\left|\llr{\psi, \gbot p_2\left[w_\mu^{(13)},p_1p_3(\hat{m}^a-\hat{m}^a_2)+(p_1q_3+q_1p_3)(\hat{m}^a-\hat{m}^a_1)\right]\psi}\right|.
$$
By definition of $\hat{m}^c$ and $\hat{m}^d$,
\begin{eqnarray*}
p_1p_3(\hat{m}^a-\hat{m}^a_2)&=&p_1p_3\hat{m}^d+p_1p_3\left(m^a(N+1)P_{N-1}+m(N+2)P_N\right)=p_1p_3\hat{m}^d,\\
(p_1q_3+q_1p_3)(\hat{m}^a-\hat{m}^a_1)&=&(p_1q_3+q_1p_3)\hat{m}^c.
\end{eqnarray*}
This leads to
\begin{eqnarray*}
\eqref{eqn:gamma:GP:d:1}
&\leq& N^3\left|\llr{w_\mu^{(13)}\psi,\gbot p_2\mathbbm{1}_{\supp{w_\mu}}(\z_1-\z_3)\Big(p_1p_3\hat{m}^d+(p_1q_3+q_1p_3)\hat{m}^c\Big)\psi}\right|\\
&&+N^3\left|\llr{\psi,\gbot p_2\Big(p_1p_3\hat{m}^d+(p_1q_3+q_1p_3)\hat{m}^c\Big)w_\mu^{(13)}\psi}\right|\\
&\ls& N^3\onorm{\gbot p_2}\left(\onorm{\hat{m}^d}+\onorm{\hat{m}^c}\right)\left(\norm{w_\mu^{(13)}\psi}\onorm{\mathbbm{1}_{\supp{w_\mu}}(\z_1-\z_3)p_1}+\norm{p_1w_\mu^{(13)}\psi}\right)
\\
&\ls& \mathfrak{e}^3(t) N^{-1+3\xi-\frac{\bo}{2}}\varepsilon^{1+\bo}
\; <\; \mathfrak{e}^3(t)\varepsilon^{1+\bo}
\end{eqnarray*}
by Lemma~\ref{lem:a_priori:w12}, Lemma~\ref{lem:g:2} and Lemma~\ref{lem:l:2}.
In order to estimate~\eqref{eqn:gamma:GP:d:2}, observe first that $\gbot w_\mu^{(13)}\neq0$ implies $|\z_2-\z_3|\ls R_\bo$. This can be seen as follows: $\gbot\neq 0$ implies $|\z_1-\z_2|\leq R_\bo$ and $w_\mu^{(13)}\neq 0$ implies $|\z_1-\z_3|\leq \mu$. Together, this yields
$$|\z_2-\z_3|\leq |\z_1-\z_2|+|\z_1-\z_3|\leq R_\bo+\mu\leq 2R_\bo.$$
Consequently,~\eqref{eqn:gamma:GP:d:2} can be written as
\begin{eqnarray*}
\eqref{eqn:gamma:GP:d:2}&=&N^3\left|\llr{\psi,\gbot w_\mu^{(13)}\mathbbm{1}_{B_{2R_\bo}(0)}(\z_2-\z_3)p_1\hat{m}^a\psi}\right|\\
&=&N^3\left|\llr{p_1\mathbbm{1}_{\supp{w_\mu}}(\z_1-\z_3)w_\mu^{(13)}\gbot\psi,\mathbbm{1}_{B_{2R_\bo}(0)}(\z_2-\z_3)\hat{m}^a\psi}\right|\\
&\leq& N^3\onorm{p_1\mathbbm{1}_{\supp{w_\mu}}(\z_1-\z_3)}\norm{\gb}_{L^\infty(\R^3)}\norm{w_\mu^{(13)}\psi}\norm{\mathbbm{1}_{B_{2R_\bo}(0)}(\z_2-\z_3)\hat{m}^a\psi}\\
&\ls &\mathfrak{e}^3(t)N^{1+\xi-\bo}\varepsilon^{2\bo-\frac23}
\; <\; \mathfrak{e}^3(t)\left(N\varepsilon^\delta\right)^{1+\xi-\bo}
\end{eqnarray*} 
by Lemma~\ref{lem:a_priori:w12} and as $2\bo-\tfrac23-\delta(1+\xi-\bo)>0$. We have used that as in the proof of Lemma~\ref{lem:g:5},
$$
\norm{\mathbbm{1}_{B_{2R_\bo}(0)}(\z_2-\z_3)\hat{m}^a\psi}^2\ls\varepsilon^{-\frac43}\mu^{2\bo}(\norm{\partial_{x_1}\hat{m}^a\psi}^2+\varepsilon^2\norm{\nabla_{y_1}\hat{m}^a\psi}^2)\ls N^{-2+2\xi-2\bo}\varepsilon^{4\bo-\frac43}\mathfrak{e}^2(t)
$$
because by Lemma~\ref{lem:l:2}, Lemma~\ref{lem:commutators:2} and Lemma~\ref{lem:a_priori:w12},
$$
\norm{\partial_{x_1}\hat{m}^a\psi}\ls \onorm{\hat{m}^a}\left(\norm{\partial_{x_1}p_1\psi}+\norm{\partial_{x_1}(1-p_1)\psi}\right)
\ls N^{-1+\xi}\mathfrak{e}(t)
$$
and analogously $\norm{\nabla_{y_1}\hat{m}^a\psi}\ls N^{-1+\xi}\varepsilon^{-1}$.
The remaining two terms~\eqref{eqn:gamma:GP:d:3} and~\eqref{eqn:gamma:GP:d:4} can be estimated as
\begin{eqnarray*}
\eqref{eqn:gamma:GP:d:3}&\ls &N^3\onorm{\gbot p_1}\left(\onorm{\hat{m}^a}+\onorm{\hat{m}^b}\right)\norm{p_1w_\mu^{(13)}\psi}\\
&\ls &\mathfrak{e}^3(t)N^{-\frac{\bo}{2}+\xi}\varepsilon^{1+\bo}
<\mathfrak{e}^3(t)\varepsilon^{1+\bo},\\
\eqref{eqn:gamma:GP:d:4}&=& N^3\left|\llr{w_\mu^{(13)}\psi,\gbot p_2\mathbbm{1}_{\supp{w_\mu}}(\z_1-\z_3)p_1(\hat{m}^b-2\hat{m}^a)\psi}\right|\\
&\leq &N^3\norm{w_\mu^{(13)}\psi}\onorm{\gbot p_2}\onorm{\mathbbm{1}_{\supp{w_\mu}}(\z_1-\z_3)p_1}\left(\onorm{\hat{m}^b}+2\onorm{\hat{m}^a}\right)\\
&\ls&\mathfrak{e}^3(t)N^{-\frac{\bo}{2}+\xi}\varepsilon^{1+\bo}
\; <\; \mathfrak{e}^3(t)\varepsilon^{1+\bo},
\end{eqnarray*}
where we have used that $\xi<\frac{\bo}{6}$ as well as Lemma~\ref{lem:a_priori:w12}, Lemma~\ref{lem:l:2} and Lemma~\ref{lem:g:2}.
\qed

\subsubsection{Proof of the bound for $\gamma_e(t)$}
Using again Lemma~\ref{lem:commutators:5}, $|\gamma_e(t)|$ can be written as
\begin{equation}\label{eqn:gamma:GP:d:7}
|\eqref{gamma:GP:e}|\ls N^4\left|\llr{\psi,\gbot\left[w_\mu^{(34)},p_3p_4(\hat{r}-\hat{r}_2)+(p_3q_4+q_3p_4)(\hat{r}-\hat{r}_1)\right]\psi}\right|.
\end{equation}
By definition of $\hat{r}$ and $\hat{m}^{c/d/e/f}$, we obtain
\begin{eqnarray*}
p_3p_4(\hat{r}-\hat{r}_2)+(p_3q_4+q_3p_4)(\hat{r}-\hat{r}_1)
&=&(p_1q_2+q_1p_2)(p_3q_4+q_3p_4)\hat{m}^c
+(p_1q_2+q_1p_2)p_3p_4\hat{m}^d\\
&&+p_1p_2(p_3q_4+q_3p_4)\hat{m}^e+
p_1p_2p_3p_4\hat{m}^f.
\end{eqnarray*}
Due to the symmetry of~\eqref{eqn:gamma:GP:d:7} under the exchanges $1\leftrightarrow2$ and $3\leftrightarrow4$, this yields
\begin{eqnarray}
|\eqref{gamma:GP:e}|
&\ls& N^4\left|\llr{\psi,\gbot p_1q_2\left[w_\mu^{(34)},p_3q_4\hat{m}^c+p_3p_4\hat{m}^d\right]\psi}\right|
\label{eqn:gamma:GP:d:8}\\
&&+N^4\left|\llr{\psi,\gbot p_1p_2\left[w_\mu^{(34)},p_3q_4\hat{m}^e+p_3p_4\hat{m}^f\right]\psi}\right|,
\label{eqn:gamma:GP:d:9}
\end{eqnarray}
where by Lemma~\ref{lem:w12:4}, Lemma~\ref{lem:g:2} and Lemma~\ref{lem:l:2},
\begin{eqnarray*}
\eqref{eqn:gamma:GP:d:8}&\leq&
N^4\left|\llr{\psi,w_\mu^{(34)}p_3\gbot p_1q_2(q_4\hat{m}^c+p_4\hat{m}^d)\psi}\right|\\
&&+N^4\left|\llr{\psi,\gbot p_1q_2(q_4\hat{m}^c+p_4\hat{m}^d)p_3w_\mu^{(34)}\psi}\right|\\
&\ls& N^4\norm{p_3w_\mu^{(34)}\psi}\onorm{\gbot p_1}\left(\onorm{\hat{m}^c}+\onorm{\hat{m}^d}\right)\\
&\ls&\mathfrak{e}^3(t)N^{-\frac{\bo}{2}+3\xi}\varepsilon^{1+\bo}
\; <\; \mathfrak{e}^3(t)\varepsilon^{1+\bo}
\end{eqnarray*}
as $\xi<\frac{\bo}{6}$. Analogously, one derives the same bound for~\eqref{eqn:gamma:GP:d:9}. 
\qed

\subsubsection{Proof of the bound for $\gamma_f(t)$}
Finally, as a consequence of Lemma~\ref{lem:l}, Lemma~\ref{lem:Phi:1} and Lemma~\ref{lem:g},
\begin{eqnarray*}
|\eqref{gamma:GP:f}|&\ls& N^2\left|\llr{\psi,\gbot p_2\left[b|\Phi(x_1)|^2,\hat{m}^bp_1+\hat{m}^aq_1\right]\psi}\right|\\
&&+N^2\left|\llr{\psi,\gbot\left[b|\Phi(x_1)|^2,p_1\hat{m}^a\right]q_2\psi}\right|\\
&\ls& N^2\norm{\Phi}^2_{L^\infty(\R)}\left(\onorm{p_2\gbot}\left(\onorm{\hat{m}^a}+\onorm{\hat{m}^b}\right)+
\norm{\gbot\psi}\norm{q_2\hat{m}^a\psi}\right)\\
&\ls& \mathfrak{e}^3(t)N^{-\frac{\bo}{2}+\xi}\varepsilon^{1+\bo}+\mathfrak{e}^2(t)\varepsilon
\; \ls\; \mathfrak{e}^2(t)\varepsilon.
\end{eqnarray*}\qed

\subsection{Proof of Proposition~\ref{prop:correction}}
Using Lemma~\ref{lem:l:2} and Lemma~\ref{lem:g:2}, we estimate
$$
N(N-1)\Re\llr{\psi,\gbot\hat{r}\psi}
\ls N^2\onorm{\gbot p_1}\left(\onorm{\hat{m}^a}+\onorm{\hat{m}^b}\right)
\ls \mathfrak{e}(t)N^{\xi-\frac{\bo}{2}}\varepsilon^{1+\bo}.
$$
\qed

\section*{Acknowledgments}
\noindent We thank Serena Cenatiempo, Maximilian Jeblick, Nikolai Leopold and Peter Pickl for helpful discussions. This work was supported by the German Research Foundation within the Research Training Group 1838 ``Spectral Theory and Dynamics of Quantum Systems''.

\renewcommand{\bibname}{References}
\bibliographystyle{abbrv}
    \bibliography{bib_PhD}

\begin{thebibliography}{10}

\bibitem{adami2007}
R.~Adami, F.~Golse, and A.~Teta.
\newblock Rigorous derivation of the cubic {NLS} in dimension one.
\newblock {\em J. Stat. Phys.}, 127(6):1193--1220, 2007.

\bibitem{abdallah2005_2}
N.~Ben~Abdallah, F.~M{\'e}hats, C.~Schmeiser, and R.~Weish{\"a}upl.
\newblock The nonlinear {Schr{\"o}dinger} equation with a strongly anisotropic
  harmonic potential.
\newblock {\em SIAM J. Math. Anal.}, 37(1):189--199, 2005.

\bibitem{benedikter2015}
N.~Benedikter, G.~de~Oliveira, and B.~Schlein.
\newblock Quantitative derivation of the {Gross--Pitaevskii} equation.
\newblock {\em Commun. Pure Appl. Math.}, 68(8):1399--1482, 2015.

\bibitem{NLS}
L.~Bo{\ss}mann.
\newblock Derivation of the 1d {NLS} equation from the 3d quantum many-body
  dynamics of strongly confined bosons.
\newblock {\em arXiv preprint, arXiv:1803.11011}, 2018.

\bibitem{brennecke2017}
C.~Brennecke and B.~Schlein.
\newblock {Gross--Pitaevskii} dynamics for {Bose--Einstein} condensates.
\newblock {\em arXiv preprint, arXiv:1702.05625}, 2017.

\bibitem{chen2013}
X.~Chen and J.~Holmer.
\newblock On the rigorous derivation of the 2d cubic nonlinear
  {Schr{\"o}dinger} equation from 3d quantum many-body dynamics.
\newblock {\em Arch. Ration. Mech. Anal.}, 210(3):909--954, 2013.

\bibitem{chen2016}
X.~Chen and J.~Holmer.
\newblock Focusing quantum many-body dynamics: the rigorous derivation of the
  1d focusing cubic nonlinear {Schr{\"o}dinger} equation.
\newblock {\em Arch. Ration. Mech. Anal.}, 221(2):631--676, 2016.

\bibitem{chen2017}
X.~Chen and J.~Holmer.
\newblock Focusing quantum many-body dynamics {II}: The rigorous derivation of
  the 1d focusing cubic nonlinear {Schr{\"o}dinger} equation from 3d.
\newblock {\em Anal. PDE}, 10(3):589--633, 2017.

\bibitem{erdos2007}
L.~Erd{\H{o}}s, B.~Schlein, and H.-T. Yau.
\newblock Derivation of the cubic non-linear {Schr{\"o}dinger} equation from
  quantum dynamics of many-body systems.
\newblock {\em Invent. Math.}, 167(3):515--614, 2007.

\bibitem{erdos2010}
L.~Erd{\H{o}}s, B.~Schlein, and H.-T. Yau.
\newblock Derivation of the {Gross--Pitaevskii} equation for the dynamics of
  {Bose--Einstein} condensate.
\newblock {\em Ann. Math.}, pages 291--370, 2010.

\bibitem{esteve2006}
J.~Esteve, J.-B. Trebbia, T.~Schumm, A.~Aspect, C.~Westbrook, and I.~Bouchoule.
\newblock Observations of density fluctuations in an elongated {Bose} gas:
  Ideal gas and quasicondensate regimes.
\newblock {\em Phys. Rev. Lett.}, 96, 2006.

\bibitem{evans}
L.~C. Evans.
\newblock {\em Partial Differential Equations}.
\newblock American Mathematical Society, 2010.

\bibitem{gorlitz2001}
A.~G{\"o}rlitz, J.~Vogels, A.~Leanhardt, C.~Raman, T.~Gustavson, J.~Abo-Shaeer,
  A.~Chikkatur, S.~Gupta, S.~Inouye, T.~Rosenband, D.~Pritchard, and
  W.~Ketterle.
\newblock Realization of {Bose--Einstein condensates} in lower dimensions.
\newblock {\em Phys. Rev. Lett.}, 87(13):130402, 2001.

\bibitem{griesemer2004}
M.~Griesemer.
\newblock Exponential decay and ionization thresholds in non-relativistic
  quantum electrodynamics.
\newblock {\em J. Funct. Anal.}, 210(2):321 -- 340, 2004.

\bibitem{henderson2009}
K.~Henderson, C.~Ryu, C.~MacCormick, and M.~Boshier.
\newblock Experimental demonstration of painting arbitrary and dynamic
  potentials for {Bose--Einstein} condensates.
\newblock {\em New J. Phys.}, 11(4):043030, 2009.

\bibitem{jeblick2016}
M.~Jeblick, N.~Leopold, and P.~Pickl.
\newblock Derivation of the time dependent {Gross--Pitaevskii} equation in two
  dimensions.
\newblock {\em arXiv preprint arXiv:1608.05326}, 2016.

\bibitem{jeblick2017}
M.~Jeblick and P.~Pickl.
\newblock Derivation of the time dependent two dimensional focusing {NLS}
  equation.
\newblock {\em arXiv preprint arXiv:1707.06523}, 2017.

\bibitem{jeblick2018}
M.~Jeblick and P.~Pickl.
\newblock Derivation of the time dependent {Gross--Pitaevskii} equation for a
  class of non purely positive potentials.
\newblock {\em arXiv preprint arXiv:1801.04799}, 2018.

\bibitem{keler2016}
J.~v. Keler and S.~Teufel.
\newblock The {NLS} limit for bosons in a quantum waveguide.
\newblock {\em Ann. Henri Poincar{\'e}}, 17(12):3321--3360, 2016.

\bibitem{kinoshita2006}
T.~Kinoshita, T.~Wenger, and D.~Weiss.
\newblock A quantum {Newton's} cradle.
\newblock {\em Nature}, 440, 2006.

\bibitem{kirkpatrick2011}
K.~Kirkpatrick, B.~Schlein, and G.~Staffilani.
\newblock Derivation of the two-dimensional nonlinear {Schr{\"o}dinger}
  equation from many body quantum dynamics.
\newblock {\em Amer. J. of Math.}, 133(1):91--130, 2011.

\bibitem{knowles2010}
A.~Knowles and P.~Pickl.
\newblock Mean-field dynamics: singular potentials and rate of convergence.
\newblock {\em Commun. Math. Phys.}, 298(1):101--138, 2010.

\bibitem{lieb_loss}
E.~H. Lieb and M.~Loss.
\newblock Analysis.
\newblock {\em American Mathematical Society, Providence, RI,}, 4, 2001.

\bibitem{LSSY}
E.~H. Lieb, R.~Seiringer, J.~P. Solovej, and J.~Yngvason.
\newblock {\em The Mathematics of the Bose Gas and its Condensation}.
\newblock Birkh{\"a}user, 2005.

\bibitem{lieb2004}
E.~H. Lieb, R.~Seiringer, and J.~Yngvason.
\newblock One-dimensional behavior of dilute, trapped {Bose} gases.
\newblock {\em Commun. Math. Phys.}, 244(2):347--393, 2004.

\bibitem{mehats2017}
F.~M{\'e}hats and N.~Raymond.
\newblock Strong confinement limit for the nonlinear {Schr{\"o}dinger} equation
  constrained on a curve.
\newblock {\em Ann. Henri Poincar{\'e}}, 18(1):281--306, 2017.

\bibitem{meinert2017}
F.~Meinert, M.~Knap, E.~Kirilov, K.~Jag-Lauber, M.~Zvonarev, E.~Demler, and
  H.-C. N{\"a}gerl.
\newblock Bloch oscillations in the absence of a lattice.
\newblock {\em Science}, 356:945--948, 2017.

\bibitem{pickl2008}
P.~Pickl.
\newblock On the time dependent {Gross--Pitaevskii-} and {Hartree} equation.
\newblock {\em arXiv preprint arXiv:0808.1178}, 2008.

\bibitem{pickl2010}
P.~Pickl.
\newblock Derivation of the time dependent {Gross--Pitaevskii} equation without
  positivity condition on the interaction.
\newblock {\em J. Stat. Phys.}, 140(1):76--89, 2010.

\bibitem{pickl2011}
P.~Pickl.
\newblock A simple derivation of mean field limits for quantum systems.
\newblock {\em Lett. Math. Phys.}, 97(2):151--164, 2011.

\bibitem{pickl2015}
P.~Pickl.
\newblock Derivation of the time dependent {Gross--Pitaevskii} equation with
  external fields.
\newblock {\em Rev. Math. Phys.}, 27(01):1550003, 2015.

\bibitem{rodnianski2009}
I.~Rodnianski and B.~Schlein.
\newblock Quantum fluctuations and rate of convergence towards mean field
  dynamics.
\newblock {\em Comm. Math. Phys.}, 291(1):31--61, 2009.

\bibitem{tao}
T.~Tao.
\newblock {\em Nonlinear {D}ispersive {E}quations: {L}ocal and {G}lobal
  {A}nalysis}.
\newblock Number 106. American Mathematical Soc., 2006.

\end{thebibliography}
\end{document}